\documentclass[11pt, a4paper]{article} 

\usepackage[german,english]{babel} 
\usepackage{booktabs} 
\usepackage{comment} 
\usepackage[utf8]{inputenc} 
\usepackage[T1]{fontenc} 

\usepackage{amsmath,amsfonts,amsthm, amssymb,mathrsfs} 
\usepackage{extpfeil}
\usepackage{tikz} 
\usepackage{tikz-cd}
\usetikzlibrary{positioning,arrows} 
\usetikzlibrary{decorations.pathreplacing} 
\usepackage{tikzsymbols} 
\usepackage{blindtext} 
\usepackage{bm, bbm}
\usepackage{hyperref}
\usepackage[shortlabels]{enumitem}
\usepackage[ruled,vlined]{algorithm2e}

\usepackage{geometry} 

\usepackage{setspace} 
\usepackage[T1]{fontenc} 
\usepackage[utf8]{inputenc} 

\usepackage{XCharter} 

\usepackage{fancyhdr} 
\usepackage{setspace}
\usepackage{graphicx}
\usepackage{amsmath}
\usepackage{tabularx}
\usepackage{multirow}
\usepackage{longtable}
\usepackage{booktabs} 
\usepackage{threeparttable} 
\usepackage{footnote}
\usepackage[round]{natbib}
\usepackage{pdflscape}

\usepackage{tikz}
\usetikzlibrary{shapes.geometric}

\usepackage{natbib} 
\usepackage{har2nat} 

\newtheorem{assumption}{Assumption}

\newtheorem{definition}{Definition}
\newtheorem{theorem}{Theorem}
\newtheorem{lemma}{Lemma}

\newtheorem*{remark}{Remark}

\title{Robust Variance Estimation in Linear Regression: \\ A Projection-Geometry Perspective} 
\author{
Yanping Chen\thanks{Department of Economics, Indiana University Bloomington. Email: \texttt{chenyanp@iu.edu}.}
}
\date{\today}

\begin{document}


\maketitle 

\begin{abstract}
Inference in linear regression commonly treats OLS residuals as proxies for unobserved errors. This approximation can fail when the regression projection is nonlocal relative to the error-dependence structure. Residualization then shifts covariance information across observations and clusters, while conventional heteroskedasticity-consistent (HC) and cluster-robust variance estimators (CRVE) retain only diagonal or within-cluster residual moments and may therefore understate sampling uncertainty.

This paper develops a projection-geometry framework for robust variance estimation. The variance of the OLS estimator is represented exactly as a Riesz functional of latent covariance blocks, and observable residual moments are linked to the target through a linear operator determined by the full regression projection. This formulation reduces variance estimation to a linear inverse problem. I propose a Riesz variance estimator that combines within- and cross-cluster residual moments. Conventional HC and CRVE emerge as restricted approximations whose validity depends on negligible projection spillovers.

The estimator remains well defined when cluster-specific leverage matrices are singular and is computed by an iterative algorithm that avoids explicit matrix inversion. Simulations show substantial undercoverage by conventional methods under projection spillovers, whereas the proposed estimator restores near-nominal coverage. In an application to colonial governor promotions, the correction changes the significance of four of five reported coefficients.

\end{abstract}

\newpage
\section{Introduction} \label{sec:introduction}
Empirical researchers routinely rely on heteroskedasticity-robust (HC) and cluster-robust variance estimators (CRVE) to conduct inference in linear regression. Their validity rests on a plug-in principle: OLS residuals approximate unobserved regression errors, so residual moments estimate the sampling variance of the estimator. In heteroskedastic settings, this logic traces back to \citet{eicker1963asymptotic} and \citet{white1980heteroskedasticity}, with finite-sample refinements yielding HC2/HC3 \citep{mackinnon1985some} that correct for leverage through $(1-h_{ii})^{-1}$ and $(1-h_{ii})^{-2}$ respectively. In clustered settings, the foundational estimator is due to \citet{liang1986longitudinal} and \citet{arellano1987computing}, with bias-corrected extensions CR2 and CR3 \citep{bell2002bias} that adjust for within-cluster leverage through $(I_{N_g} - H_{gg})^{-1/2}$ and $(I_{N_g} - H_{gg})^{-1}$. \citet{imbens2016robust} and \citet{pustejovsky2018small} develop practical small-sample recommendations based on CR2, and \citet{cameron2008bootstrap} propose the wild clustering bootstrap as an alternative to asymptotic inference. \citet{mackinnon2023fast} show that CR3-based inference is substantially more reliable than CR1 when cluster sizes are heterogeneous or the number of regressors is moderate. \citet{cameron2015practitioner} and \citet{mackinnon2023cluster} provide comprehensive empirical guides.

A crucial but often implicit requirement behind this plug-in argument is geometric: the OLS projection must be approximately \emph{local} --- nearly diagonal under heteroskedasticity and block-diagonal under clustering --- so that residuals primarily reflect errors within their own observation or cluster. Many modern empirical designs violate this locality property. Rich control sets, shared latent factor directions, multi-way fixed effects, and non-nested cluster structure induce substantial \emph{projection spillovers}: OLS residuals mechanically mix errors across observations or clusters through non-negligible off-diagonal mass in the residual-maker $M = I - X(X'X)^{-1}X'$. When the number of covariates grows with the sample size, cross-leverage terms remain first-order objects \citep{anatolyev2017asymptotics, anatolyev2024off}; when regressors share latent directions, the projection matrix couples observations globally even when the number of regressors is moderate; and when mobility graphs are dense, off-diagonal mass grows with connectivity regardless of dimension. In all these settings, classical HC and CRVE fail through two distinct channels.

The first failure is \emph{inadequate leverage adjustment}. Classical estimators correct for self-attenuation through diagonal leverage terms --- $(1-h_{ii})$ or $(I_{N_g}-H_{gg})$ --- but use only the diagonal $h_{ii}$, which determines the total off-diagonal mass $h_{ii}(1-h_{ii})=\sum_{j\neq i}M_{ij}^2$ but carries no information about which observations that mass attaches to or what variances $\{\sigma_j^2\}$ they carry. This variance pattern is a property of $\mathrm{col}(X)^{\perp}$, orthogonal to the leverage geometry in $\mathrm{col}(X)$, and is therefore invisible to any correction based on $h_{ii}$ alone: the resulting bias, $\sum_{j\neq i}h_{ij}^2\sigma_j^2/(1-h_{ii})^2$, is an inner product between the off-diagonal projection weights and the heteroskedasticity pattern that no function of $h_{ii}$ can recover. Under homoskedasticity, $\sigma_j^2=\sigma^2$ factors out and the bias reduces to $\sigma^2 h_{ii}(1-h_{ii})$, fully removable by a diagonal adjustment; under heteroskedasticity, the failure is structural.

The second failure is \emph{discarding cross-residual products}. Classical estimators use only squared residuals $\{\hat{u}_i^2\}$, discarding cross-products $\{\hat{u}_i \hat{u}_j\}$ entirely. These cross-products satisfy $\mathbb{E}[ \hat{u}_i \hat{u}_j \mid X] = \sum_k M_{ik} M_{jk} \sigma_k^2$ and carry identifying information about the same variance components through different, \emph{signed} linear combinations of the projection weights --- combinations the squared-residual system, whose coefficient matrix has non-negative entries $\{M_{ik}^2\}$, cannot form. When the diagonal moment system loses rank or is poorly conditioned in the direction relevant for $\mathrm{var} (\hat{\beta})$, cross-moments provide the identifying variation the diagonal system structurally lacks.

Recent work addresses the first failure partially. \citet{cattaneo2018inference} correct for cross-leverage contamination in the scalar heteroskedastic case by inverting the Hadamard square of the residual-maker from controls, $(M_W \odot M_W)^{-1}$. \citet{jochmans2022heteroscedasticity} addresses the same setting through a cross-fit construction: he scales the residual by the diagonal of the controls-only projection, $\tilde{u}_i = \hat{u}_i/(M_W)_{ii}$, and uses the product $y_i\tilde{u}_i$ as an unbiased signal for $\sigma_i^2$, approximating the leave-one-out idea of \citet{kline2020leave} via the Sherman--Morrison formula and extending consistency to $\limsup_n q_n/n<1$ without Hadamard inversion. Both rely only on the projection of controls $M_W$, ignoring leverage from the regressors of interest $X$, and using $y_i$ as a proxy for $u_i$ introduces sensitivity to the conditional mean under misspecification. \citet{anatolyev2026many} extend to the clustered setting via a leave-cluster-out cross-fit estimator using the full residual-maker $M=M_{[X,W]}$, accounting for leverage from both $X$ and $W$, but require $M_{gg}$ to be invertible within each cluster. Moreover, all three methods retain only within-observation or within-cluster residual moments, leaving the second failure --- discarding cross-residual products --- entirely unaddressed.

This paper develops a unified framework that addresses both failures simultaneously. The central insight is that $\mathrm{var}(\hat{\beta})$ admits an exact finite-sample representation as a Riesz functional of latent covariance blocks, with observable residual moments linked to this target through a linear operator determined by projection geometry. Variance estimation therefore reduces to a linear inverse problem. We propose the \emph{Riesz variance estimator} using both within- and cross-cluster residual moments with the full projection geometry. Classical HC, CRVE, and \citet{cattaneo2018inference} all arise as restricted solutions differing in which residual moments they retain and how much projection geometry they incorporate. The Riesz estimator does not require invertibility of the Hadamard square or of $M_{gg}$, remaining valid precisely when \citet{anatolyev2026many} and classical bias-corrected estimators break down.

The variance-recovery system has dimension driven by cluster sizes, making direct inversion of the associated operator systems computationally infeasible at empirically relevant scales. We show that the problem admits a matrix-free formulation amenable to the LSQR algorithm of \citet{paige1982lsqr}, based on Golub--Kahan bidiagonalization \citep{golub1965calculating}, which requires only operator actions and scales to large datasets with dense projection geometry without forming or inverting large matrices.

The paper establishes consistency and asymptotic normality of the resulting $t$-statistics under heteroskedasticity, arbitrary within-cluster dependence, heterogeneous and possibly growing cluster sizes, and high-dimensional controls, assuming independent clusters. Classical HC, CRVE, and the many-controls correction of \citet{cattaneo2018inference} are shown to be restricted solutions within the Riesz framework, with their validity characterized by a projection sparsity condition. Simulations confirm that classical estimators severely undercover under projection spillovers, while the proposed estimator restores nominal coverage. An empirical application to colonial governor promotions \citep{xu2018costs} shows that correcting for projection spillovers reverses the significance of four out of five reported coefficients.

\paragraph{Contributions.}
\textit{First}, this paper provides a geometric diagnosis of why classical HC and CRVE estimators fail, developed fully in Section~\ref{sec:motivation}. The diagnosis reframes the failure identified above --- that $h_{ii}$ alone cannot detect the variance pattern across neighboring observations --- not as a finite-sample bias but as a structural identification problem rooted in the orthogonality of $\mathrm{col}(X)$ and $\mathrm{col}(X)^{\perp}$. This explains why many controls, shared latent factors, and dense mobility graphs all generate the same failure mode despite their superficial differences: each concentrates off-diagonal projection mass on pairs of observations whose error variances may differ, a condition no diagonal correction can detect. To our knowledge, this geometric characterization is new to the literature, which has previously focused on leverage bounds \citep{chesher1987bias} and asymptotic rates \citep{anatolyev2017asymptotics} rather than the identification structure of the problem.

\textit{Second}, this paper develops a unified projection-geometry framework for residual-based variance estimation and establishes a strict existence hierarchy within it. Rather than estimating individual error variances, we characterize $\mathrm{var}(\hat{\beta})$ directly as a Riesz functional of the full residual-moment system indexed by the OLS projection matrix, unifying classical HC, CRVE, and many-controls estimators as restricted solutions within a common operator system. This clarifies the geometric conditions under which classical restrictions are first-order valid and when they are not. The Riesz estimator exists under strictly weaker conditions than classical bias-corrected estimators, remaining valid when cluster-specific leverage matrices are singular --- precisely the settings where CR2, CR3, and \citet{anatolyev2026many} break down. This failure arises in the empirical application \citep{xu2018costs}, where some dyadic clusters are perfectly nested within the high-dimensional fixed effects, making CR2 and CR3 undefined.

\textit{Third}, the paper shows that the Riesz inverse problem admits a matrix-free formulation amenable to the LSQR algorithm of \citet{paige1982lsqr}, making the estimator computationally feasible in large-scale designs with dense projection geometry, requiring only operator actions rather than explicit matrix formation or inversion.

\paragraph{Organization.}
Section~\ref{sec:motivation} provides geometric intuition for projection spillovers and illustrates why classical variance estimators fail.
Section~\ref{sec:model} introduces the model and variance functional.
Section~\ref{sec:framework} develops the operator framework.
Section~\ref{sec:estimation} presents the LSQR estimator.
Section~\ref{sec:theory} establishes asymptotic results.
Section~\ref{sec:partial} links the partial system to classical HC and CRVE methods and characterizes the existence hierarchy.
Section~\ref{sec:simulation} presents simulation results.
Section~\ref{sec:application} presents the empirical application.

\section{Motivation} \label{sec:motivation}
This section develops geometric intuition for the two failures of classical variance estimators, which motivates this paper. For expositional clarity, we work in the scalar case $(N_g=1,\,G=n)$ and abstract from nuisance controls; Section~\ref{sec:model} formalizes the econometric model with both the clustered setting and high-dimensional controls.

OLS splits $\mathbb{R}^n$ into two orthogonal subspaces: the column space $\mathrm{col}(X)$ of dimension $p$, onto which $H=X(X'X)^{-1}X'$ projects, and its orthogonal complement $\mathrm{col}(X)^{\perp}$ of dimension $n-p$, onto which $M=I-H$ projects. Fitted values $\hat{y}=Hy$ live entirely in $\mathrm{col}(X)$; residuals $\hat{u}=Mu$ live entirely in $\mathrm{col}(X)^{\perp}$. The true errors $u\in\mathbb{R}^n$, however, are not confined to either subspace: their component $Hu$ lies in $\mathrm{col}(X)$, contributes to $\hat{\beta} - \beta =(X'X)^{-1} X'u$, and is annihilated by $M$ --- leaving no trace in $\hat{u}$ and making it unrecoverable from residuals alone. Only the component $(I-H) u = \hat{u}$ is observable. The sampling uncertainty in $\hat{\beta}$ therefore originates in $\mathrm{col}(X)$, while the only observable evidence lives in $\mathrm{col}(X)^{\perp}$: inference requires using the observed component to learn about the unobserved one, and the two subspaces are orthogonal.

Let $h_{ij} = [X(X'X)^{-1}X']_{ij}$ denote the $(i,j)$ entry of the hat matrix, and $M_{ij} = \mathbf{1}_{i=j} - h_{ij}$ the corresponding entry of the residual-maker $M$. The OLS residual satisfies $\hat{u} = Mu$, so for each observation $i$,
\begin{equation}
    \hat{u}_i \;=\; M_{ii} \, u_i \;+\; \sum_{j\neq i} M_{ij} \, u_j. \label{eq:residual-decomp}
\end{equation}
The structure of \eqref{eq:residual-decomp} reflects the two-space geometry directly. The diagonal term $M_{ii}=1-h_{ii}$ shrinks $u_i$: the larger $h_{ii}$, the more observation $i$ pulls the fit toward itself and the less of its own error survives in $\hat{u}_i$. The off-diagonal terms $M_{ij}=-h_{ij}$ import the errors of other observations: when $x_i$ and $x_j$ are close in $\mathrm{col}(X)$, the regression cannot cleanly separate their contributions, so $\hat{u}_i$ mechanically absorbs $-h_{ij}u_j$.

Under heteroskedasticity with $\mathrm{Var}(u_k)=\sigma_k^2$, the exact finite-sample moments are
\begin{align}
    \mathbb{E}[\hat{u}_i^2] 
        &= M_{ii}^2 \, \sigma_i^2 + \underbrace{\sum_{j\neq i} M_{ij}^2 \, \sigma_j^2}_{\text{projection spillover}}, \label{eq:diag-moment} \\
    \mathbb{E}[\hat{u}_i \hat{u}_j] 
        &= \sum_{k=1}^{n} M_{ik} M_{jk} \, \sigma_k^2, \qquad i\neq j. \label{eq:cross-moment}
\end{align}

\subsection{Failure 1: Inadequate Leverage Adjustment} \label{subsec:failure1}

The projection spillover in \eqref{eq:diag-moment} is not an incidental nuisance --- its magnitude is pinned down exactly by the geometry of $M$. Since $H$ is symmetric and idempotent, taking the $(i,i)$ entry of $[H^2]_{ii}=H_{ii}$ and splitting diagonal from off-diagonal yields
\begin{equation}
    \sum_{j\neq i} M_{ij}^2 \;=\; \|P_i\|^2 \cdot \|Q_i\|^2 \;=\; h_{ii} (1-h_{ii}), \label{eq:indiv-offdiag}
\end{equation}
where $P_i=He_i$ is the shadow of $e_i$ onto $\mathrm{col}(X)$ and $Q_i=Me_i$ is its residual direction. The off-diagonal mass $h_{ii}(1-h_{ii})$ measures how much of observation $i$'s leverage leaks to other observations. It vanishes at $h_{ii}\in\{0,1\}$: at $h_{ii}=1$, the observation pulls the fit completely toward itself, leaving nothing to share; at $h_{ii}=0$, there is no leverage to spread. The leakage is maximized at $h_{ii}=1/2$. Paradoxically, it is intermediate-leverage observations, not the most extreme ones, that generate the largest off-diagonal mass and therefore the largest potential for projection spillover.

The HC family corrects for the self-attenuation $M_{ii}^2=(1-h_{ii})^2$ in the first term of \eqref{eq:diag-moment}, but uses only $h_{ii}$ and therefore leaves the projection spillover in the second term unaddressed:
\begin{equation}
    \mathbb{E} \! \left[ \frac{\hat{u}_i^2}{(1-h_{ii})^2} \right] \;=\; \sigma_i^2 \;+\; \frac{\sum_{j\neq i} M_{ij}^2 \, \sigma_j^2}{(1-h_{ii})^2}. \label{eq:hc3-bias}
\end{equation}
The residual bias is the inner product between the projection weights $\{M_{ij}^2\}$ and the heteroskedasticity pattern $\{\sigma_j^2\}$. The diagonal $h_{ii}$ determines the total off-diagonal mass $h_{ii} (1-h_{ii}) = \sum_{j\neq i} M_{ij}^2$ but carries no information about which observations that mass attaches to or what variances they carry: the variance pattern across neighbors is a property of $\mathrm{col}(X)^{\perp}$, orthogonal to the leverage geometry encoded in $\mathrm{col}(X)$, and is therefore undetectable from $h_{ii}$ alone regardless of how the correction is designed.

Under homoskedasticity $\sigma_j^2=\sigma^2$, equation \eqref{eq:indiv-offdiag} shows the inner product collapses to $\sigma^2 h_{ii}(1-h_{ii})$, a function of $h_{ii}$ alone that is fully removable by a diagonal adjustment --- HC2 with weight $(1-h_{ii})^{-1}$ is exactly unbiased in this case. Under heteroskedasticity, the failure is structural rather than a matter of finite-sample precision, and is severe precisely when off-diagonal mass concentrates on neighbors with heterogeneous $\sigma_j^2$. Because proximity in regressor space and variance heterogeneity are determined by entirely separate mechanisms and can coincide arbitrarily, this case is empirically common. Three structural mechanisms generate first-order inner product $\sum_{j\neq i}M_{ij}^2\sigma_j^2$.

\paragraph{Many regressor directions.}
Summing \eqref{eq:indiv-offdiag} over all $i$ gives the aggregate off-diagonal mass
\begin{equation}
    \sum_{i=1}^{n} \sum_{i\neq j} M_{ij}^2 \;=\; p - \sum_i h_{ii}^2 \; \le p  \! \left(1-\frac{p}{n}\right), \label{eq:offdiag-mass}
\end{equation}
with equality under uniform leverage $h_{ii}=p/n$ for all $i$. When $p/n\to \alpha \in(0,1)$, the residual space $\mathrm{col}(X)^{\perp}$ has dimension $n-p$ --- a shrinking fraction of $\mathbb{R}^n$. The $n$ residual directions $\{Q_i\}$ are crowded into this thin space and cannot all be mutually orthogonal: by \eqref{eq:offdiag-mass}, the aggregate off-diagonal mass is $O(n)$ even under uniform leverage. With heteroskedasticity bounded away from homoskedasticity, the spillover in \eqref{eq:hc3-bias} then contributes at the same order as the signal $\sigma_i^2$ regardless of the design. This arises in regressions with many controls, saturated interactions, or many instruments \citep{cattaneo2018inference, anatolyev2017asymptotics, anatolyev2024off}.

\paragraph{Shared latent directions.}
When $p$ regressors share a common latent factor --- as in interactive fixed effects \citep{bai2009panel} or common correlated effects \citep{pesaran2006estimation} --- the factor structure makes $X'X$ nearly low-rank, so $(X'X)^{-1}$ has large eigenvalues in the near-collinear directions spanned by the factor. Since $h_{ij} = x_i' (X'X)^{-1} x_j$, this amplifies the inner product between any two observations loading on the shared factor, making $h_{ij}$ globally large regardless of raw Euclidean distance. Ignoring the latent structure and working with the full $p$-dimensional $X$ directly converts a benign $k$-dimensional signal --- where the effective projection rank is $k$ and off-diagonal mass is bounded by $k$ --- into a dangerous $p$-dimensional projection that inverts through near-zero eigenvalues, inflating both $h_{ij}$ and $\mathrm{Var} (\hat{\beta})$ simultaneously. The inner product $\sum_{j\neq i} M_{ij}^2 \sigma_j^2$ is then first-order when pairs with strong factor exposure also have heterogeneous error variances. This mechanism compounds with the many-regressors mechanism when $p/n \to \alpha$: the factor structure concentrates the $O(n)$ aggregate off-diagonal mass on pairs where the inner product is largest, producing bias larger than either mechanism in isolation.

\paragraph{Graph connectivity in non-nested fixed effects.}
In two-way or multi-way fixed effects designs, $[M_D]_{ij}\neq 0$ if and only if observations $i$ and $j$ belong to the same connected component of the mobility graph \citep{verdier2020estimation}, so worker mobility across firm boundaries generates pairwise off-diagonal mass that grows with graph density. As with the shared-directions mechanism, this generates bias only when connected observations have heterogeneous error variances. In dyadic designs --- such as the colonial governor application \citep{xu2018costs} in Section~\ref{sec:application} --- pairs sharing a network node differ systematically in their connection status and hence plausibly in their error variances, making the heteroskedasticity condition empirically relevant.

The formal conditions under which each mechanism makes the inner product $\sum_{j\neq i} M_{ij}^2 \sigma_j^2$ first-order for $\hat{V}$ are characterized in Section~\ref{sec:partial} via the projection-sparsity condition. Failure~1 is therefore neither about projection geometry alone nor heteroskedasticity alone, but their structural interaction: any leverage correction based solely on $h_{ii}$ is blind to it.

\subsection{Failure 2: Discarding Cross-Residual Moments} \label{subsec:failure2}

Besides the projection spillover in \eqref{eq:diag-moment}, classical estimators commit a second, logically independent failure: they discard the cross-product equations \eqref{eq:cross-moment} entirely. This section shows that discarding these equations is not merely inefficient --- it can be identifying, particularly in clustered designs.

The diagonal equations \eqref{eq:diag-moment} form a linear system in the unknown variances $\mathbf{s} = (\sigma_1^2, \ldots, \sigma_n^2)'$:
\begin{equation}
    A\mathbf{s} = \mathbf{q}, \qquad A_{ik} = M_{ik}^2, \quad q_i = \mathbb{E} [\hat{u}_i^2]. \label{eq:diag-system}
\end{equation}
The coefficient matrix $A$ has an algebraic limitation: its rows $\{M_{ik}^2\}_k$ are elementwise non-negative, so the system can form only non-negative linear combinations of $\{\sigma_k^2\}$. The cross-moment equations \eqref{eq:cross-moment} extend this system with rows $B_{(ij),k} = M_{ik} M_{jk}$, providing $\binom{n}{2}$ additional equations for the same $n$ unknowns. Unlike the rows of $A$, the rows of $B$ are \emph{signed}: $M_{ik}M_{jk}$ is positive when $i$ and $j$ have projection weights of the same sign onto direction $k$, and negative otherwise. This sign variation spans directions in $\mathbb{R}^n$ that the non-negative rows of $A$ cannot reach, which is the algebraic source of additional identifying power. Once $\mathbf{s}$ is recovered, $V$ is a known linear combination of its entries determined by how much each observation contributes to variation in the regressor of interest. The diagonal system is sufficient if and only if the target combination lies in the row space of $A$ and $A$ is well-conditioned in that direction. Three failure modes arise when this condition fails.

\paragraph{Failure mode 2a: Rank deficiency.}
The diagonal system loses rank whenever some row $A_{i\cdot}$ is a linear combination of other rows --- the squared projection weights of observation $i$ carry no information about $\{\sigma_k^2\}$ beyond what the other diagonal equations already contain. This occurs when rows $i$ and $j$ of $M$ are nearly proportional ($M_{i\cdot}\approx cM_{j\cdot}$), which happens precisely when $i$ and $j$ are close in $\mathrm{col}(X)$ --- the same condition that drives Failure~1. Certain contrasts in $\{\sigma_k^2\}$ entering $V$ then become invisible to all diagonal equations simultaneously, and no bias correction can recover them: the squared residuals $\{\hat{u}_i^2\}$ are insufficient for identifying $V$ regardless of how they are reweighted. Cross-moment equations involve $M_{ik}M_{jk}$ rather than $M_{ik}^2$, generating new signed linear combinations of $\{\sigma_k^2\}$ that can restore identification when the diagonal system fails. Geometrically, a third observation $l$ with $M_{l\cdot}$ not aligned with $M_{i\cdot}$ or $M_{j\cdot}$ acts as a probe: the cross-moment $\mathbb{E} [\hat{u}_i \hat{u}_l] - \mathbb{E} [\hat{u}_j  \hat{u}_l]$ amplifies the difference between rows $i$ and $j$ even when $\| M_{i\cdot} - M_{j\cdot} \|$ is small.

\paragraph{Failure mode 2b: Ill-conditioning.}
Even when $A$ has full rank, $\hat{\mathbf{s}}=A^{-1}\mathbf{q}$ may be imprecisely estimated when $A$ is ill-conditioned --- precisely the setting where observations are close in $\mathrm{col}(X)$ and $M_{ij}$ is large, as in Failure~1. The rows of the diagonal system generate quadratic forms $(M_{ik}^2)_k$ that are always non-negative and therefore span a restricted subset of all linear forms in $\{\sigma_k^2\}$. When observations concentrate on few directions in $\mathrm{col}(X)$, these squared forms become nearly collinear. Cross-moment rows $\{M_{ik}M_{jk}\}_k$ generate mixed quadratic forms with sign variation, providing the linearly independent equations that stabilize the system. Cross-moments are therefore most valuable in exactly the designs where Failure~1 is most severe: clustered regressors, shared factor loadings, and overlapping group memberships. The two failures compound rather than operate independently.

\paragraph{Failure mode 2c: Low power in the variance-relevant direction.}
Failure mode 2b concerns global conditioning of $A$; this is its directional sharpening. Even when $A$ is full rank and globally well-conditioned, the variance of $\hat{V} = \mathbf{r}' A^{-1} \hat{\mathbf{q}}$ is governed by how well $A$ is conditioned specifically in the direction of $\mathbf{r}$ --- a system can be stable on average while still weak in the one direction that determines $V$. When the regressor of interest varies primarily between clusters, $V$ loads on between-cluster contrasts in $\{\sigma_k^2\}$. The diagonal weights $M_{ik}^2$ are nearly identical for observations within the same cluster, making $A$ nearly block-constant within clusters and poorly conditioned in the between-cluster directions that $\mathbf{r}$ loads on. Cross-moment weights $M_{ik}M_{jk}$ for observations in \emph{different} clusters involve between-cluster projection weights directly, providing the signal the diagonal system lacks in precisely those directions.

Failure~2 is therefore not resolved by better bias corrections to $\hat{u}_i^2$, but requires cross-residual products $\hat{u}_i\hat{u}_j$ that carry identifying information through signed linear combinations of $\{\sigma_k^2\}$ the diagonal system cannot form. Both failures originate in the same geometry: $M$ projects $u$ onto $\mathrm{col}(X)^{\perp}$, permanently discarding the component $Hu$ that drives $\mathrm{Var}(\hat{\beta})$ and entangling the observable residuals through $h_{ij}$. Failure~1 is the consequence of this entanglement for individual $\hat{u}_i^2$; Failure~2 is the consequence of discarding the additional identifying information the entanglement encodes in $\hat{u}_i\hat{u}_j$. The Riesz framework in Section~\ref{sec:framework} addresses both simultaneously by treating the full residual moment array as a linear inverse problem rather than a collection of scalars to be reweighted.

\section{Econometric Model} \label{sec:model}
We consider a clustered linear regression model. Let clusters be indexed by $g = 1, \ldots, G$, each containing $N_g$ observations, with total sample size $n = \sum_{g=1}^G N_g$. For a generic variable $Z$, let $z_{gi}$ denote the $i$th observation in cluster $g$, $Z_g$ the stacked observations in cluster $g$, and $Z = (Z_1', \ldots, Z_G')'$ the full sample stacked across clusters. Following \citet{cattaneo2018inference}, the data satisfy
\begin{equation} \label{eq:model}
    Y_g = X_g \beta + W_g \gamma + U_g,
\end{equation}
where $Y_g \in \mathbb{R}^{N_g}$ is the outcome vector, $X_g \in \mathbb{R}^{N_g \times k}$ collects the regressors of primary interest, $W_g \in \mathbb{R}^{N_g \times \ell}$ denotes a (possibly high-dimensional) set of control variables, and $U_g \in \mathbb{R}^{N_g}$ represents unobserved disturbances that may exhibit arbitrary dependence within clusters. The parameter of interest is $\beta \in \mathbb{R}^k$, while $\gamma \in \mathbb{R}^\ell$ is treated as a nuisance parameter. The dimension $\ell$ may increase with the sample size $n$, accommodating high-dimensional control designs. We assume clusters are independent conditional on $(X,W)$.

This specification encompasses several well-known cases. When $\ell = 0$, it reduces to the classical low-dimensional regression model; when $\ell/n$ does not vanish, it corresponds to high-dimensional regression settings. The framework includes heteroskedasticity-robust inference as the case $G = n$ and $N_g = 1$, and standard cluster-robust inference when $G \to \infty$ with bounded cluster sizes $N_g$.

We handle the high-dimensional controls by partialling out $W$. Let $M_W := I_n - W(W'W)^{-1}W'$ denote the residual-maker projecting onto the orthogonal complement of $\operatorname{span}(W)$, and define $\widetilde X = M_W X$ and $\widetilde Y = M_W Y$. Applying ordinary least squares to the residualized model yields
\begin{equation} \label{eq:beta_hat}
    \widehat\beta = (X'M_WX)^{-1}(X'M_WY) = \beta + \big( \widetilde X'\widetilde X \big)^{-1}  \big( \widetilde X'U \big),
\end{equation}
and the residuals
\begin{equation} \label{eq:U_hat}
    \widehat U = M_W (Y - X \widehat\beta) = M U,
\end{equation}
where $M := M_{\widetilde X} M_W$ with $M_{\widetilde X} = I_n - \widetilde X(\widetilde X'\widetilde X)^{-1}\widetilde X'$ is the effective OLS residual-maker. All coupling across observations or clusters arises through the projection geometry of $M$.

To conduct inference, an appropriate normalization for $(\widehat\beta - \beta)$ is required, which generally depends on cluster sizes, within-cluster dependence, and regressor variation \citep{carter2017asymptotic, hansen2019asymptotic, djogbenou2019asymptotic}. Let $\mu_n$ denote the effective scaling of the estimator, with $\mu_n = n$ under heteroskedastic sampling and $\mu_n$ equal to $G$, $n$, or an intermediate rate under clustered dependence. To avoid estimating $\mu_n$ directly, we focus on the scalar $t$-statistic
\begin{equation} \label{eq:t-statistic}
    t = \frac{\sqrt{\mu_n}\, a'(\widehat\beta - \beta)}{\sqrt{\mu_n \widehat V}} = \widehat V^{-1/2} a'(\widehat\beta - \beta),
\end{equation}
where $a \in \mathbb{R}^k$ is a fixed unit vector and $\widehat V$ is a variance estimator satisfying
\begin{equation} \label{eq:v-target}
    \widehat V = V + o_p(\mu_n^{-1}),
\end{equation}
with
\begin{equation} \label{eq:var}
    V = a' \Gamma^{-1} \bigg( \frac{1}{n^2} \sum_{g=1}^{G} \widetilde X_g \Sigma_g \widetilde X_g' \bigg) \Gamma^{-1} a, \quad \Gamma = \frac{1}{n} \sum_{g=1}^{G} \widetilde X_g' \widetilde X_g, \quad \Sigma_g = \mathbb{E}[U_g U_g' \mid X,W].
\end{equation}
Under suitable regularity conditions, $t \xrightarrow{d} \mathcal{N}(0,1)$.

In practice, variance estimation is challenging because the covariance blocks $\{\Sigma_g\}$ are unknown and the errors $\{U_g\}$ are unobserved, so $V$ in \eqref{eq:var} cannot be computed directly. By \eqref{eq:U_hat},
    \[ \widehat U_g = [MU]_g = \sum_{h=1}^G M_{gh} U_h = M_{gg} U_g + \sum_{h \neq g} M_{gh} U_h, \]
where $M_{gh}$ denotes the $(g,h)$-th block of $M$ conformable with the cluster structure. Classical HC and CRVE correct only the diagonal blocks $\{M_{gg}\}$ and therefore implicitly treat $\widehat U_g \widehat U_g'$ as informative about $\Sigma_g$ alone. This is justified only when the off-block terms $\sum_{h\neq g} M_{gh}U_h$ are negligible, or when they are not but the clusters they draw on share the same covariance structure: in the latter case the contamination is a function of the projection geometry alone and can, in principle, be absorbed into a correction based on $\{ M_{gg} \}$. The difficulty is structural rather than a matter of finite-sample precision when both conditions fail simultaneously --- when $M_{gh}$ is non-negligible \emph{and} $\Sigma_g \neq \Sigma_h$ for the clusters $h$ on which $g$ draws --- so that $\widehat U_g \widehat U_g'$ becomes an entangled mixture of $\Sigma_g$ and the covariance blocks of other clusters $\{\Sigma_h\}_{h\neq g}$, in a way no correction based on $M_{gg}$ alone can disentangle.

Designs with non-negligible off-block mass $\{M_{gh} : g \neq h\}$ are not exotic; examples include the three mechanisms developed in Section~\ref{sec:motivation} --- many controls ($\ell/n \to \alpha \in (0,1)$), shared latent factor directions among the regressors, and dense connectivity in non-nested fixed-effects designs --- though these are illustrative rather than exhaustive, and other sources of non-local projection geometry (e.g., spatial or network-based control structures) can generate the same difficulty. None of these mechanisms guarantees first-order bias on its own: each makes non-negligible $M_{gh}$ possible, but whether the resulting off-block mass actually produces first-order bias for $V$ depends additionally on whether it aligns with heterogeneity in $\{ \Sigma_g \}$, a joint condition we make precise via the projection-sparsity criterion in Section~\ref{sec:partial}.

Rather than attempting to estimate the latent covariance blocks $\{\Sigma_g\}$ or to reconstruct the unobserved errors $\{U_g\}$, the approach developed in this paper targets the variance functional $V$ directly, exploiting the structure of the full residual moment system implied by the projection geometry.

For a scalar contrast, the variance admits the linear representation
\begin{equation} \label{eq:oracle-var}
    V = \frac{1}{n^2} \sum_{g=1}^G \mathrm{tr}\!\big(\Sigma_g \widetilde X_g Q \widetilde X_g'\big), \qquad Q := \Gamma^{-1} a a' \Gamma^{-1}.
\end{equation}
This representation rewrites $V$ as a quantity that is \emph{linear} in the unknown covariance array $\{\Sigma_g\}$, rather than as the quadratic form in \eqref{eq:var}. Once $\{\Sigma_g\}$ is equipped with a suitable inner product, linearity is precisely what permits $V$ to admit a Riesz representer, which we construct in Section~\ref{sec:framework}.

Although $\{\Sigma_g\}$ is unobserved, the feasible residual moments
   \[ \widehat C_{hh'} := \widehat U_h \widehat U_{h'}' \]
have population counterparts, for cluster indices $h,h' \in \{1,\ldots,G\}$,
    \[ C_{hh'} = \mathbb{E}[\widehat C_{hh'} \mid X,W] = \sum_{g=1}^G M_{hg}\Sigma_g M_{gh'}, \]
using cluster independence and symmetry of $M$. This linear system links the full residual moment array $\{C_{hh'}\}$ --- including both within- and cross-cluster moments --- to latent covariance blocks $\{\Sigma_g\}$ through the projection geometry of $M$, motivating the Riesz representation framework developed in the next section. In summary, this section has introduced three objects on which everything that follows depends: the projection geometry of $M$, governing how errors mix across observations and clusters; the unknown target $\{\Sigma_g\}$, which we never attempt to estimate directly; and the observable link $\{C_{hh'}\}$, the full residual moment array through which $\{\Sigma_g\}$ is related to $V$.

\section{The Riesz Representation of the Variance Functional} \label{sec:framework}
Our goal is to recover the scalar variance functional $V=\mathrm{Var}(a'\widehat\beta)$ directly, without estimating the latent covariance blocks $\{\Sigma_g\}_{g=1}^G$. The key observation is that $V$ is a linear functional of $\{\Sigma_g\}$ and therefore admits a Riesz representer under a suitable inner product. We then exploit the linear operator linking $\{\Sigma_g\}$ to the residual moment array $\{C_{hh'}\}$ to obtain a feasible representation and develop a corresponding estimator via observable sample residual moments $\{\widehat C_{hh'}\}$.

Throughout this section, the sample size $n$ and cluster structure $\{N_g\}_{g=1}^G$ are fixed, and we suppress the subscript $n$ on the spaces $\mathcal H$ and $\mathcal K$, and the objects built from them, restoring it from Section~\ref{sec:theory} onward when asymptotic behavior as $n\to\infty$ is the object of interest.

\subsection{Latent covariance space and variance functional}

For a given sample of size $n$ with known cluster structure $\{N_g\}_{g=1}^G$, define the finite-dimensional Hilbert space
    \[ \mathcal H := \bigoplus_{g=1}^G \mathbb R^{N_g \times N_g}, \qquad \langle A,B \rangle_{\mathcal H} := \frac{1}{n^2} \sum_{g=1}^G \mathrm{tr}(A_g'B_g), \qquad \|A\|_{\mathcal H}^2 := \langle A,A \rangle_{\ mathcal H}. \]
The covariance array $\Sigma := \{\Sigma_g\}_{g=1}^G$ is an element of $\mathcal H$.

\begin{lemma}[Riesz applicability] \label{lem:riesz}
Under Assumption~\ref{asmp:A4} (well-conditioned design), the variance functional $\Phi:\mathcal H\to\mathbb R$ defined by 
    \[ \Phi(\Sigma) := V = \frac{1}{n^2} \sum_{g=1}^{G} \mathrm{tr} \! \big( \Sigma_g \, \widetilde X_g Q \widetilde X_g' \big) \]
is linear and bounded on $(\mathcal H, \langle \cdot, \cdot \rangle_{\mathcal H})$. Hence, by the Riesz representation theorem, there exists a unique element $R^\star = \{R_g^\star\} \in \mathcal H$ such that
    \[ \Phi(\Sigma) = \langle \Sigma, R^\star \rangle_{\mathcal H} \quad \text{for all } \Sigma \in \mathcal H. \]
Moreover, $R_g^\star = \widetilde X_g Q \widetilde X_g'$ for all $g$.
\end{lemma}

Lemma~\ref{lem:riesz} expresses the target variance as an inner product between the latent covariance array and a known representer $R^\star$. To obtain a feasible representation, we next link $\Sigma$ to observable residual moments via the projection geometry.

\subsection{Residual moment space and the projection operator}

Let $M$ denote the effective OLS residual-maker from Section~\ref{sec:model}, so that $\widehat U=MU$. Define the space of the full residual moment array
    \[ \mathcal K := \bigoplus_{h,h'=1}^G \mathbb R^{N_h\times N_{h'}}, \qquad \langle C,D \rangle_{\mathcal K} := \frac{1}{n^2} \sum_{h,h'=1}^G \mathrm{tr}(C_{hh'}'D_{hh'}), \qquad \|C\|_{\mathcal K}^2 := \langle C,C \rangle_{\mathcal K}. \]
Let $C := \{C_{hh'}\} \in \mathcal K$ denote the population residual moments 
    \[ C_{hh'} := \mathbb E[\widehat U_h\widehat U_{h'}'\mid X,W]. \]
Under cluster independence, these moments satisfy the linear system
    \[ C_{hh'} = \sum_{g=1}^G M_{hg} \Sigma_g M_{gh'}. \] 
This motivates the linear operator $\mathcal A:\mathcal H\to\mathcal K$ defined by
\begin{equation} \label{eq:A-def}
    [\mathcal A(\Sigma)]_{hh'} := \sum_{g=1}^G M_{hg}\Sigma_g M_{gh'}.
\end{equation}
Thus $C = \mathcal A(\Sigma)$.

\begin{lemma}[Adjoint operator] \label{lem:adjoint}
The operator $\mathcal A$ admits a unique Hilbert-space adjoint $\mathcal A^\ast:\mathcal K \to \mathcal H$ satisfying\footnote{Formally, for a bounded linear operator $\mathcal A:\mathcal H\to\mathcal K$ between Banach spaces, the adjoint $\mathcal A'$ is defined on the dual spaces
$\mathcal K^*\to\mathcal H^*$. Because $\mathcal H$ and $\mathcal K$ are finite-dimensional Hilbert spaces, the Riesz representation theorem provides canonical isometric identifications $\mathcal H\simeq\mathcal H^*$ and $\mathcal K\simeq\mathcal K^*$. Throughout, we implicitly use these identifications and work with the induced Hilbert-space adjoint $\mathcal A^\ast:\mathcal K\to\mathcal H$.}
    \[ \langle \mathcal A(\Sigma),D\rangle_{\mathcal K} = \langle \Sigma,\mathcal A^\ast(D) \rangle_{\mathcal H} \quad \text{for all } \Sigma \in \mathcal H,\; D \in \mathcal K. \]
Moreover, $\mathcal A^\ast$ has the explicit block form
\begin{equation}\label{eq:Astar-def}
    [\mathcal A^\ast(D)]_g = \sum_{h,h'=1}^G M_{gh} \, D_{hh'} \, M_{h'g}.
\end{equation}
\end{lemma}

Lemma~\ref{lem:adjoint} allows us to transfer the Riesz representer from the latent-covariance space $\mathcal H$ to the residual-moment space $\mathcal K$.

\subsection{Residual-moment representer and the Riesz variance estimator}

We seek weights $D=\{D_{hh'}\}\in\mathcal K$ such that the residual-moment inner product  $\langle C,D \rangle_{\mathcal K}$ reproduces the target variance functional $\Phi(\Sigma)$ for all admissible covariance arrays $\Sigma$ satisfying  $C = \mathcal A(\Sigma)$. By Lemmas~\ref{lem:riesz}--\ref{lem:adjoint},
    \[ \langle C,D \rangle_{\mathcal K} = \langle \mathcal A(\Sigma),D \rangle_{\mathcal K} = \langle \Sigma, \mathcal A^\ast(D) \rangle_{\mathcal H}. \]
Hence, any $D\in\mathcal K$ satisfying
    \begin{equation}\label{eq:adjoint-system}
        \mathcal A^\ast(D) = R^\star = \{\widetilde X_g Q\widetilde X_g'\}_{g=1}^G
    \end{equation}
ensures
    \[ \Phi(\Sigma) = \langle \Sigma,R^\star\rangle_{\mathcal H} = \langle C,D\rangle_{\mathcal K}. \]

Equation~\eqref{eq:adjoint-system} is the central identification equation of the paper: it characterizes all residual-moment weighting schemes that exactly recover the target variance functional. Classical HC and CRVE estimators correspond to particular restricted solutions to this equation (see Section~\ref{subsec:HC-CRVE}).

Because the residual-moment space $\mathcal K$ collects all within- and cross-cluster residual moments, it has higher dimension than the latent covariance space $\mathcal H$, the operator $\mathcal A: \mathcal H \to \mathcal K$ is generically non-square and tall. Its adjoint $\mathcal A^\ast: \mathcal K \to \mathcal H$ is therefore fat, so the adjoint equation \eqref{eq:adjoint-system} need not admit a unique solution. Moreover, the projection geometry encoded in $\{M_{gh}\}$ may induce linear dependencies among the residual moment equations, so $\mathcal A$ need not have full column rank.
If $\ker(\mathcal A^\ast)\neq\{0\}$, the solution set of \eqref{eq:adjoint-system} is an affine subspace: for any $Z\in\ker(\mathcal A^\ast)$, we have $\langle C,Z\rangle_{\mathcal K}=0$ for all $C\in\mathrm{range}(\mathcal A)$, so that $D+Z$ induces the same functional on $\mathrm{range}(\mathcal A)$.

The variance recovery problem is therefore a linear inverse problem: observable residual moments overidentify the latent covariance blocks, so the adjoint equation admits multiple solutions unless a normalization is imposed. We select the representer of smallest Frobenius norm.

\begin{definition}[Minimum-Frobenius-norm representer] \label{def:minnorm}
Let $\|\cdot\|_F$ denote the Frobenius norm induced by  $\langle \cdot, \cdot \rangle_{\mathcal K}$. Define
\begin{equation} \label{eq:min-norm}
     D^\star := \arg\min_{D\in\mathcal K} \Big\{\|D\|_F:\ \mathcal A^\ast(D)=R^\star\Big\}.
\end{equation}
Equivalently, $D^\star=(\mathcal A^\ast)^+R^\star$, where $(\mathcal A^\ast)^+$ denotes the Moore--Penrose pseudoinverse under the chosen inner products \citep{golub1965calculating}.
\end{definition}

The minimum-norm choice selects the unique representer orthogonal to $\ker(\mathcal A^\ast)$, eliminating redundant residual-moment weightings that do not affect the variance functional.

\begin{definition}[Riesz variance estimator] \label{def:riesz-est}
Let $D^\star$ denote the minimum-norm representer defined in \eqref{eq:min-norm}. Given sample residual moments  $\widehat C_{hh'}=\widehat U_h\widehat U_{h'}'$, define the Riesz variance estimator
\begin{equation}\label{eq:Vhat-Riesz}
    \boxed{ \widehat V_{\mathrm{Riesz}} := \langle \widehat C, D^\star \rangle_{\mathcal K} = \frac{1}{n^2} \sum_{h,h'=1}^G \widehat U_h' D_{hh'}^\star \widehat U_{h'}. }
\end{equation}
\end{definition}

The characterization above reduces variance recovery to solving the linear system $\mathcal A^\ast(D)=R^\star$, a potentially least-square inverse problem. Because $\mathcal K$ can be high-dimensional and the operator is available only through evaluations of $(\mathcal A,\mathcal A^\ast)$, we compute the canonical minimum-norm solution using an iterative matrix-free Krylov method. Numerical implementation is discussed in Section~\ref{sec:estimation}.

\begin{remark}[Finite-sample Hilbert spaces and asymptotic sequences] 
All objects in this section are defined at each fixed $n$ in finite-dimensional Hilbert spaces, so existence and uniqueness of the Riesz representer are immediate. Asymptotic results in Section~\ref{sec:theory} analyze the sequence $\{(\mathcal H_n,\mathcal K_n,\mathcal A_n,R_n^\star,D_n^\star)\}_{n\ge1}$ under uniform regularity conditions. No infinite-dimensional embedding is required.
\end{remark}

\section{Feasible Estimation and Numerical Implementation} \label{sec:estimation}
This section describes how to compute the minimum-norm residual-moment representer $D^\star$ defined in \eqref{eq:min-norm}, and hence the feasible variance estimator \eqref{eq:Vhat-Riesz}. In principle, this linear inverse system could be solved via direct pseudoinversion (e.g.\ via QR or SVD). However, such approaches would require explicit construction of large projection-induced operators, thus computationally infeasible at the cluster scales considered here. Instead, we exploit the operator structure of the problem and compute the canonical minimum-norm solution using a matrix-free Krylov method.

Specifically, we employ the LSQR algorithm of \citet{paige1982lsqr}. LSQR is mathematically equivalent, in exact arithmetic, to applying conjugate gradients to the normal equation \eqref{eq:normal-residmom}, but it never forms the composite operator $\mathcal A\mathcal A^\ast$ explicitly and offers improved numerical stability in the ill-conditioned regimes we consider. LSQR computes the minimum-norm solution using only repeated evaluations of the operators $\mathcal A$ and $\mathcal A^\ast$, without explicitly forming the associated normal operators. This matrix-free structure makes LSQR particularly well suited to our setting, where the operators are large and potentially ill-conditioned.
\subsection{Normal equations and Krylov geometry}

Recall from Section~\ref{sec:framework} that the canonical residual-moment representer $D^\star \in \mathcal K$ is the minimum-norm solution to the adjoint equation $\mathcal A^\ast(D) = R^\star$. Equivalently, $D^\star$ solves the least-squares problem
\begin{equation} \label{eq:ls-prob}
    D^\star = \arg \min_{D\in\mathcal K} \|\mathcal{A}^\ast D - R^\star \|_{\mathcal H}, \qquad \text{with } \|D^\star\|_F \text{ minimal among minimizers.}
\end{equation}
The associated normal equation expressed in the residual-moment space is
\begin{equation}\label{eq:normal-residmom}
    (\mathcal A\mathcal A^\ast)\,D^\star = \mathcal A R^\star, \qquad D^\star\in\mathrm{range}(\mathcal A).
\end{equation}

Because $\mathcal A\mathcal A^\ast$ is self-adjoint and positive semidefinite, and positive definite on $\mathrm{range}(\mathcal A)$, the normal equation \eqref{eq:normal-residmom} has a unique solution on $\mathrm{range}(\mathcal A)$, which coincides with the Moore--Penrose solution $D^\star=(\mathcal A^\ast)^+R^\star$.

\begin{lemma}[Positive semidefiniteness of $\mathcal A\mathcal A^\ast$] \label{lem:AAstar-psd}
Let $\mathcal A:\mathcal H\to\mathcal K$ be a bounded linear operator with adjoint $\mathcal A^\ast$. Then:
\begin{enumerate}[(i)]
    \item $\mathcal A\mathcal A^\ast$ is self-adjoint and positive semidefinite, since for any $D\in\mathcal K$,
        \[ \langle D,\mathcal A\mathcal A^\ast D\rangle_{\mathcal K} = \|\mathcal A^\ast D\|_{\mathcal H}^2 \ge 0. \]
    \item $\mathcal A\mathcal A^\ast$ is positive definite on $\mathrm{range}(\mathcal A)$: if $D\in \mathrm{range}(\mathcal A)$ and $\langle D,\mathcal A\mathcal A^\ast D\rangle_{\mathcal K}=0$, then $\mathcal A^\ast D=0$ and hence $D=0$.
\end{enumerate}
\end{lemma}

Direct solution of \eqref{eq:normal-residmom} is computationally infeasible because the composite operator $\mathcal A\mathcal A^\ast$ is large and may be ill-conditioned. A key structural feature of our framework is that the operators $\mathcal A$ and $\mathcal A^\ast$ can be evaluated without explicitly forming large matrices. Their actions admit the blockwise expressions
\begin{equation}\label{eq:actions}
\begin{aligned}
    [\mathcal A(H)]_{hh'} &= \sum_{g=1}^G M_{hg}H_gM_{gh'},  && H\in\mathcal H,\\[0.3em]
    [\mathcal A^\ast(D)]_g  &= \sum_{h,h'=1}^G M_{gh}D_{hh'}M_{h'g},  && D\in\mathcal K.
\end{aligned}
\end{equation}
Thus all computations reduce to repeated multiplications involving the projection blocks $\{ M_{gh} \}$. This matrix-free structure naturally motivates Krylov subspace methods, which approximate the solution by exploring subspaces generated by successive applications of the normal operator $\mathcal T=\mathcal A\mathcal A^\ast$.

Since the solution $D^\star$ lies in the residual-moment space $\mathcal K$, while the least-squares residual $r(D)=R^\star-\mathcal A^\ast D$ lies in the covariance space $\mathcal H$, the geometry of the problem naturally involves both spaces. In particular, for the operator $\mathcal T=\mathcal A\mathcal A^\ast$ and the initial iterate $D^{(0)}=0$, the residual of the normal equation equals $r_0=\mathcal A R^\star$. The associated residual-moment Krylov subspace is
    \[ \mathscr K_t(\mathcal T,r_0) = \mathrm{span}\{r_0,\mathcal T r_0,\ldots,\mathcal T^{t-1}r_0\}. \]
Since the normal equation~\eqref{eq:normal-residmom} also admits a dual formulation in the covariance space,
    \[ (\mathcal A^\ast\mathcal A)H^\star = R^\star, \qquad D^\star=\mathcal A(H^\star), \]
the associated Krylov directions lie in
    \[ \mathscr K_t(\mathcal A^\ast\mathcal A,R^\star). \]

Thus the operators $\mathcal A\mathcal A^\ast$ and $\mathcal A^\ast\mathcal A$ govern the geometry of the inverse problem in the residual-moment and covariance spaces, respectively. This dual geometry is exactly what LSQR exploits and a naive application of conjugate gradients to a single normal equation would not: rather than fixing one space and forming $\mathcal T$ or its dual explicitly, the algorithm alternates between the two spaces at every iteration, applying $\mathcal A$ and $\mathcal A^\ast$ in turn. It is this alternation --- detailed next --- that allows the composite operators $\mathcal A\mathcal A^\ast$ and $\mathcal A^\ast \mathcal A$ to govern the convergence theory without ever being formed.

\subsection{LSQR algorithm and numerical implementation}

Based on the Krylov geometry described above, we compute the representer $D^\star$ using the classical LSQR algorithm of \citet{paige1982lsqr}. LSQR is designed for large-scale least-squares problems and is mathematically equivalent (in exact arithmetic) to applying conjugate gradients to the normal equations, while avoiding explicit formation of the composite operator $\mathcal A\mathcal A^\ast$ and providing improved numerical stability.

LSQR applies Golub--Kahan bidiagonalization to the operator pair $(\mathcal A,\mathcal A^\ast)$, generating orthonormal bases in the covariance space $\mathcal H$ and residual-moment space $\mathcal K$. Starting from the normalized vector
    \[ u_1 = \frac{R^\star}{\|R^\star\|_{\mathcal H}}, \]
the algorithm alternates between applications of $\mathcal A$ and $\mathcal A^\ast$ to construct sequences of orthonormal vectors $\{u_t\}_{t\ge1}\subset\mathcal H$ and $\{v_t\}_{t\ge1}\subset\mathcal K$. These vectors respectively span the Krylov subspaces
    \[ \mathscr K_t(\mathcal A^\ast\mathcal A, R^\star) \qquad \text{and}\qquad \mathscr K_t (\mathcal A\mathcal A^\ast, \mathcal A R^\star). \]
    
At iteration $t$, the LSQR iterate satisfies
    \[ D_t \in \mathscr K_t(\mathcal A\mathcal A^\ast, \mathcal A R^\star), \]
so that the approximate representer is expressed as a linear combination of the residual-moment Krylov directions
    \[ D_t = \sum_{j=1}^t y_{t,j}\, v_j . \]
The coefficients $y_t$ are obtained by solving a small least-squares problem associated with the bidiagonal matrix produced by the Golub–Kahan bidiagonalization.

Standard Krylov theory implies that the error decreases geometrically with the effective condition number $\kappa(\mathcal T)=\lambda_{\max}(\mathcal T)/\lambda_{\min}^+(\mathcal T)$ (where $\lambda_{\min}^+$ denotes the smallest \emph{nonzero} eigenvalue, since $\mathcal T$ is only positive semidefinite on all of $\mathcal K$ by Lemma~\ref{lem:AAstar-psd}) on $\mathrm{range}(\mathcal A)$. In particular,
    \[ \|D_t - D^\star\|_{\mathcal T} \le 2 \Big( \frac{\sqrt{\kappa(\mathcal T)}-1} {\sqrt{\kappa(\mathcal T)}+1} \Big)^t \|D_0 - D^\star\|_{\mathcal T}, \]
where $\|x\|_{\mathcal T}^2 = \langle x,\mathcal T x\rangle_{\mathcal K} = \| \mathcal{A}^\ast x \|^2_{\mathcal H}$. Since $D_0=0$ and $\mathcal A^\ast D^\star = R^\star$, we have $\|D_0-D^\star\|_{\mathcal T}=\|R^\star\|_{\mathcal H}$. Therefore, we terminate the iteration when the relative residual satisfies
    \[ \frac{\|\mathcal A^\ast D_t - R^\star \|_{\mathcal H}} {\|R^\star\|_{\mathcal H}} \le \tau_n. \]
Under the conditions in Section~\ref{sec:theory}, choosing $\tau_n$ to decrease sufficiently fast ensures that the numerical error introduced by the iterative solution is asymptotically negligible.

\begin{lemma}[LSQR residual--error relationship] \label{lem:lsqr-residual-error}
Let $\mathcal T := \mathcal A\mathcal A^\ast$. There exists a constant $C>0$, depending only on spectral bounds of $\mathcal T$ on $\mathrm{range}(\mathcal A)$, such that for every iteration $t$,
\begin{equation} \label{eq:lsqr-det-bound}
    \|D_t-D^\star\|_{\mathcal T} \;\le\; C\, \|R^\star - \mathcal A^\ast D_t \|_{\mathcal H}.
\end{equation}
Hence the stopping rule above implies
\begin{equation} \label{eq:lsqr-det-bound-tau}
    \|D_t-D^\star\|_{\mathcal T} \;\le\; C\,\tau_n\,\|R^\star\|_{\mathcal H}.
\end{equation}
Moreover, under Assumptions~\ref{asmp:A2} and~\ref{asmp:A4},
\begin{equation} \label{eq:lsqr-det-bound-mu}
    \|D_t-D^\star\|_{\mathcal T} = O_p \! \left( \tau_n \sqrt{\frac{\sup_g N_g}{n}} \right).
\end{equation}
\end{lemma}

Equation~\eqref{eq:lsqr-det-bound} is a deterministic bound, valid at every fixed $t$ regardless of sample size: it converts the readily computable relative residual $\|\mathcal A^\ast D_t-R^\star\|_{\mathcal H}$ into a bound on the estimation error $\|D_t-D^\star\|_{\mathcal T}$ via the spectral properties of $\mathcal T$ alone. Equation~\eqref{eq:lsqr-det-bound-tau} then specializes this bound to the terminal iteration, using only the stopping rule itself. Equation~\eqref{eq:lsqr-det-bound-mu} further translates this into a probabilistic rate by bounding $\|R^\star\|_{\mathcal H}$ under the sampling assumptions of Section~\ref{sec:theory}, and is the form used in the asymptotic results that follow.

By Theorem~\ref{thm:feasible} and Assumption~\ref{asmp:A3}, a practical sufficient condition is $\tau_n \ \mu_n^{\frac{2+3\lambda}{4+4\lambda}}\to 0$. Since $\mu_n\le n$, this condition is satisfied for any $\tau_n=n^{-\alpha}$ with $\alpha\ge 3/4$. In ill-conditioned settings, numerical stability may be improved by standard techniques such as Tikhonov regularization or preconditioning; see, e.g., \citet{saad2003iterative,hansen1998rank}. Algorithm~\ref{alg:lsqr-riesz} summarizes the LSQR implementation used in our estimator. \\

\begin{algorithm}[H]
\caption{LSQR implementation for the minimum-norm Riesz representer}
\label{alg:lsqr-riesz}
\DontPrintSemicolon
\KwIn{Operator actions $\mathcal A:\mathcal H\to\mathcal K$ and $\mathcal A^\ast:\mathcal K\to\mathcal H$; target $R^\star\in\mathcal H$; tolerance $\tau_n$}
\KwOut{Approximate representer $D_t \approx D^\star=(\mathcal A^\ast)^+R^\star$}

\BlankLine
\textbf{(0) Initialization:}\;
\Indp
    $v_0 \gets 0 \in \mathcal K$, \quad $D_0 \gets 0 \in \mathcal K$ \;
    $u_1 \gets R^\star/\|R^\star\|_{\mathcal H}$, \quad $b_1 \gets \|R^\star\|_{\mathcal H}$\;
\Indm

\BlankLine
\For{$t=1,2,\ldots$}{
\Indp
\textbf{(1) Golub--Kahan bidiagonalization update:}\;
\Indp
    $\tilde v_t \gets \mathcal A \, u_t - b_t v_{t-1}$ \;
    $a_t \gets \|\tilde v_t\|_{\mathcal K}$,\quad $v_t \gets \tilde v_t/a_t$\;
    $\tilde u_{t+1} \gets \mathcal A^\ast \, v_t - a_t u_t$\;
    $b_{t+1} \gets \|\tilde u_{t+1}\|_{\mathcal H}$, \quad $u_{t+1} \gets \tilde u_{t+1}/b_{t+1}$\;
\Indm

\textbf{(2) Small least-squares update:}\;
\Indp
Form the bidiagonal matrix
    \[ B_t= \begin{bmatrix}  a_1 \\ b_2 & a_2 \\ & b_3 & \ddots \\ && \ddots & a_t \\ &&& b_{t+1}  \end{bmatrix},  \qquad e_1=(1,0,\ldots,0)'. \]

Compute $y_t \gets \arg\min_y \|B_t y - b_1 e_1\|$\;

Update $D_t \gets \sum_{j=1}^t y_{t,j}\, v_j$\;
\Indm

\textbf{(3) Stopping rule:}\;
\Indp
    \If{$\|\mathcal A^\ast D_t - R^\star \|_{\mathcal H}/\|R^\star\|_{\mathcal H} \le \tau_n$}{
        \textbf{break}\;
    }
\Indm
\Indm
}
\Return{$D_t$}\;
\end{algorithm}

\section{Asymptotic Theory}\label{sec:theory}

\subsection{Assumptions}\label{subsec:assumption}
Let $\|\cdot\|_F$ denote the Frobenius norm and $\|\cdot\|_{\mathrm{op}}$ the operator (spectral) norm.
\begin{assumption}[Sampling]\label{asmp:A1}
The sample consists of $G=G_n$ clusters with $\sum_{g=1}^G N_g=n$ and $G\to\infty$. Conditional on $(X, W)$, errors $\{U_g\}$ are independent across clusters  and satisfy
    \[ \mathbb E[U_g\mid X,W]=0, \qquad \mathbb E[U_gU_g'\mid X,W]=\Sigma_g, \]
where $\Sigma_g$ is positive semi-definite.
\end{assumption}

\begin{assumption}[Moment bounds]\label{asmp:A2}
For some $0<\lambda\leq 2$,
    \[ \sup_{g,i}\mathbb E\!\left[ |u_{gi}|^{4+2\lambda}\mid X,W\right] < \infty, \qquad \sup_g \frac{1}{N_g} \sum_{i=1}^{N_g} \bigl(\mathbb{E} \|\tilde x_{gi}\|^{4+2\lambda} + \mathbb{E} \|w_{gi}\|^{4+2\lambda} \bigr) < \infty. \]
\end{assumption}

\begin{assumption}[Cluster sizes]\label{asmp:A3}
There exists a deterministic sequence $\mu_n\to\infty$ and a finite constant $\nu_a>0$ such that 
    \[ \mu_n\, V\;\to\; \nu_a, \]
and the maximal cluster size obeys
    \[ \mu_n^{\frac{2+\lambda}{2+2\lambda}} \, \sup_g \frac{N_g}{n} \;\to\; 0. \]
\end{assumption}
\begin{assumption}[Design regularity]\label{asmp:A4}
The control and residualized design matrices satisfy
    \[ \lambda_{\min} \! \Big( \frac{1}{n}\sum_{g=1}^G W_g'W_g \Big) \ge c_W > 0, \qquad \Gamma_n := \frac{1}{n} \sum_{g=1}^G \widetilde X_g' \widetilde X_g \xrightarrow{p} \Gamma_0 \succ 0. \]
\end{assumption}

\begin{assumption}[Riesz identification and stability] \label{asmp:A5}
Let $R_n^\star := \{\widetilde X_g Q_n \widetilde X_g'\}_{g=1}^G$ with $Q_n := \Gamma_n^{-1} a a' \Gamma_n^{-1}$, and let $\mathcal A_n^\ast$ be the adjoint operator in~\eqref{eq:Astar-def}.
\begin{enumerate}[(i)]
    \item (Exact solvability.) For all sufficiently large $n$,
        \[ R_n^\star \in \mathrm{range}(\mathcal A_n^\ast). \]
    \item (Stability.) The Moore--Penrose pseudoinverse of $\mathcal A_n^\ast$ is uniformly bounded on its range:
        \[ \sup_n \big\| (\mathcal A_n^\ast|_{\mathrm{Range} (\mathcal A_n^\ast)})^{+} \big\|_{\mathrm{op}} \;\le\; c < \infty. \]
\end{enumerate}
\end{assumption}

Assumptions~\ref{asmp:A1}--\ref{asmp:A5} jointly ensure that the variance functional $\Phi(\Sigma) = \mathrm{Var}(a'\widehat\beta)$ is identifiable from the full observable residual moment array and that the resulting Riesz-based $t$-statistic is asymptotically normal. Assumptions~\ref{asmp:A1}--\ref{asmp:A4} provide standard conditions for cluster asymptotics, following \citet{djogbenou2019asymptotic} and related work, and guarantee that the target representer $R_n^\star$ is well defined and $\|R_n^\star\|_{\mathcal H_n}$ is stochastically bounded. While Assumptions~\ref{asmp:A1}--\ref{asmp:A4} restrict the data-generating process, Assumption~\ref{asmp:A5} restricts the deterministic operator $\mathcal A_n^\ast$ induced by the projection geometry of the design. Part~(i) ensures existence of an exact solution to $\mathcal A_n^\ast(D)=R_n^\star$; part~(ii) ensures well-posedness through a uniform bound on the Moore--Penrose pseudoinverse restricted to $\mathrm{range}(\mathcal A_n^\ast)$. Together, these conditions guarantee that the canonical minimum-norm solution $D_n^\star$ exists, is unique, and remains stochastically bounded.

Part~(i) admits a simple sufficient condition that depends only on the dimensions of the design: it holds whenever $\mathrm{rank}(\mathcal A_n^\ast)=\dim(\mathcal H_n)$, i.e., when the adjoint operator is surjective, so that every element of $\mathcal H_n$ --- including $R_n^\star$ --- lies in its range. This rank condition is generic when $G$ is large relative to $\sup_g N_g$, since $\dim(\mathcal K_n)=\sum_{h,h'}N_hN_{h'}\gg\sum_g N_g^2=\dim(\mathcal H_n)$ in that regime, though it is a sufficient and not a necessary condition: exact solvability can hold for the specific $R_n^\star$ at hand even when $\mathcal A_n^\ast$ fails to be surjective.

Assumption~\ref{asmp:A5}(i) is in fact stronger than what the asymptotic results in this section require: an approximate version, replacing exact range membership with a sufficiently fast-decaying approximation error, suffices for both Theorem~\ref{thm:riesz-t} and Theorem~\ref{thm:feasible}. We use the exact version throughout this section for simplicity, and return to the distinction between exact and approximate solvability --- and to its relation to first-order approximation more generally --- when relating the Riesz estimator to classical restricted solutions in Section~\ref{sec:partial}.

\subsection{Main Results}

\begin{theorem}[Asymptotic normality]\label{thm:riesz-t}
Suppose Assumptions~\ref{asmp:A1}--\ref{asmp:A5} hold. Then,
\begin{align}
    & V^{-1/2}\, a'(\widehat\beta-\beta) \;\xrightarrow{d}\; \mathcal N(0,1), \label{eq:normality} \\
    & \widehat V_{\mathrm{Riesz}}/V \;\xrightarrow{p}\; 1, \label{eq:consistency} \\
    & \widehat V_{\mathrm{Riesz}}^{-1/2}\, a'(\widehat\beta-\beta) \;\xrightarrow{d}\; \mathcal N(0,1). \label{eq:t-Riesz}
\end{align}
\end{theorem}

Theorem~\ref{thm:riesz-t} establishes consistency and asymptotic normality for the Riesz variance estimator and the associated $t$-statistic when the minimum-norm representer $D_n^\star$ is used. As discussed following Assumption~\ref{asmp:A5}, Part~(i) of that assumption is stronger than necessary for this result; Section~\ref{sec:partial} states the weaker approximate-solvability condition the proof in fact needs.

\begin{theorem}[Feasible implementation via LSQR]\label{thm:feasible}
Let $\widehat D_n$ be the LSQR approximation to the minimum-norm solution $D_n^\star$ of $\mathcal A_n^\ast(D_n)=R_n^\star$, computed as in Algorithm~\ref{alg:lsqr-riesz} with the \emph{relative} stopping rule
\begin{equation}\label{eq:lsqr-stop}
    \frac{\big\|R_n^\star - \mathcal A_n^\ast(\widehat D_n)\big\|_{\mathcal H_n}}{\|R_n^\star\|_{\mathcal H_n}}
    \;\le\; \tau_n,
\end{equation}
where the tolerance sequence $\tau_n$ satisfies
\begin{equation}\label{eq:tau-rate}
    \mu_n\,\tau_n\,\sqrt{\frac{\sup_g N_g}{n}} \;\longrightarrow\; 0.
\end{equation}
Define $\widehat V^{\mathrm{LSQR}}_{\mathrm{Riesz}}$ as in \eqref{eq:Vhat-Riesz} with $D_n^\star$ replaced by $\widehat D_n$. Then, under Assumptions~\ref{asmp:A1}--\ref{asmp:A5},
\begin{align}
    & \widehat V^{\mathrm{LSQR}}_{\mathrm{Riesz}}/V \;\xrightarrow{p}\; 1, \label{eq:consistency-lsqr} \\
    & \big(\widehat V^{\mathrm{LSQR}}_{\mathrm{Riesz}}\big)^{-1/2}\, a'(\widehat\beta-\beta) \; \xrightarrow{d} \; \mathcal N(0,1). \label{eq:t-lsqr}
\end{align}
\end{theorem}

Theorem~\ref{thm:feasible} shows that the LSQR iterate $\widehat D_n$ stays close enough to the exact minimum-norm solution $D_n^\star$, at the rate controlled by $\tau_n$, that the feasible estimator $\widehat V^{\mathrm{LSQR}}_{\mathrm{Riesz}}$ inherits the conclusion of Theorem~\ref{thm:riesz-t}. The rate condition~\eqref{eq:tau-rate} governing $\tau_n$ controls the estimator-level error $| \widehat V^{\mathrm{LSQR}}_{\mathrm{Riesz}} -\widehat V_{\mathrm{Riesz}}|$ directly; the proof, detailed in Appendix~D, proceeds by bounding this error in terms of $\| \widehat D_n - D_n^\star \|_{\mathcal K_n}$ and the stochastic order of $\| \widehat C_n \|_{\mathcal K_n}$. Section~\ref{sec:partial} revisits this same approximation logic --- how close a representer must be to $D_n^\star$ to preserve first-order validity --- in the context of classical restricted estimators, where the relevant representer is not an LSQR iterate but a diagonal or block-diagonal weighting scheme.

\section{Partial Riesz Equivalence and Limits} \label{sec:partial}

This section relates the full Riesz estimator to the classical estimators it generalizes. First, we formalize the diagonal partial-information restriction underlying HC and CRVE and characterize when it reproduces the full Riesz representer \emph{exactly} at fixed $n$ (Lemma~\ref{lem:part-exact}), and when it does so only \emph{first-order} (Lemma~\ref{lem:part-firstorder}) --- two related but distinct questions that recur throughout the section. Second, we compare the full and partial systems at the level of matrix rank, isolating the scalar-cluster (heteroskedastic) case in which the two systems coincide under a mild richness condition. Third, we place HC2--HC5 and CR0--CR3 within the framework as specific diagonal or block-diagonal weighting schemes, give a computable diagnostic for when such restrictions are reliable, and close with the resulting existence hierarchy --- making precise how Failures 1 and 2 of Section~\ref{sec:motivation} are resolved by the full Riesz system, and re-emerge, in restricted form, under each classical estimator.

\subsection{The diagonal restriction: exact and first-order equivalence} \label{subsec:diag-restriction}

This subsection studies \emph{partial Riesz representations} of the variance functional, obtained by restricting the full residual moment array used to construct the Riesz representer. Throughout, the \emph{full Riesz} refers to the representation based on the entire residual moment array $\{ C_{hh'} \}_{h,h'=1}^G$, while the \emph{partial Riesz} focuses on the most common partial-information restriction: retaining only the \emph{within--cluster} residual moments $\{C_{hh}\}$ and discards the cross-cluster blocks $\{C_{hh'}:h\neq h'\}$. This is precisely the information set underlying classical HC and CR variance estimators, which treat feasible residuals as informative only through within-cluster quadratic forms.

To formalize this diagonal partial-information restriction, define the diagonal moment space
     \[ \mathcal K_n^{\mathrm{diag}} := \bigoplus_{h=1}^G \mathbb R^{N_h\times N_h}, \qquad [P_{\mathrm{diag}}(C)]_{hh'} := \begin{cases} C_{hh}, & h=h',\\ 0, & h\neq h', \end{cases} \]
and the associated \emph{partial Riesz operator}
    \[ \mathcal A_n^{\mathrm{diag}} := P_{\mathrm{diag}} \circ \mathcal A_n : \mathcal H_n \to \mathcal K_n^{\mathrm{diag}}, \]
which maps a covariance collection $\{\Sigma_g\}$ to the induced diagonal residual cross-moments $\{C_{hh}\}$. Note that although $P_{\mathrm{diag}}$ discards empirical cross-cluster moments $\{ C_{hh': h \neq h'} \}$, the diagonal moments $\{C_{hh}\}$ still aggregate information about all covariance blocks $\{\Sigma_g\}$ through the off-diagonal project blocks $\{M_{gh}: h \neq g\}$.

With the same inner-product geometry as in Section~\ref{sec:framework}, the adjoint $(\mathcal A_n^{\mathrm{diag}})^\ast:\mathcal K_n^{\mathrm{diag}}\to \mathcal H_n$ admits the block representation
    \[ \big[(\mathcal A_n^{\mathrm{diag}})^\ast(\{D_{hh}\})\big]_g = \sum_{h=1}^G M_{gh}\, D_{hh}\, M_{hg}, \qquad g=1,\dots,G. \]
A partial Riesz representer under the diagonal restriction is any $\{D_{hh}\} \in \mathcal K_n^{\mathrm{diag}}$ solving
    \begin{equation} \label{eq:part-Riesz}
        (\mathcal A_n^{\mathrm{diag}})^\ast(\{D_{hh}\}) = R_n^\star := \{\widetilde X_g Q_n \widetilde X_g'\}_{g=1}^G.
    \end{equation}
Vectorizing it yields the block--Khatri--Rao system
    \begin{equation} \label{eq:Khatri-Rao}
        \operatorname{vec}(\widetilde X_g Q_n \widetilde X_g') = \sum_{h=1}^G (M_{gh} \otimes M_{gh})\, \operatorname{vec}(D_{hh}) \quad \Longleftrightarrow\quad (M\!\ast\! M)\, d = r,
    \end{equation}
where $M\!\ast\!M$ stacks the block Kronecker products $\{M_{gh}\otimes M_{gh}\}$, $d := \operatorname{vec} \big( \{D_{hh}\}_{h=1}^G \big)$ denotes the column vector formed by stacking $\operatorname{vec} (D_{hh})$ over $h=1,\dots,G$, and $r := \operatorname{vec} \big( \{\widetilde X_g Q_n \widetilde X_g'\}_{g=1}^G \big)$ is the analogous stacked column for the target blocks. When all clusters have size one, this reduces to the inverse Hadamard system of \citet{cattaneo2018inference}; for general clusters, it becomes a natural ``inverse Khatri--Rao'' problem.

\paragraph{Exact equivalence.}

The partial and full Riesz systems coincide \emph{exactly} at fixed $n$ whenever
    \begin{equation} \label{eq:part-exact}
        R_n^\star \in \mathrm{Range} \big((\mathcal A_n^{\mathrm{diag}})^\ast\big),
    \end{equation}
so that~\eqref{eq:part-Riesz} admits at least one solution $D_n^{\mathrm{diag}}\in\mathcal K_n^{\mathrm{diag}}$. In this case, the diagonal weights $D_n^{\mathrm{diag}}$ provide an alternative representation of the same Riesz functional as the full-system solution $D_n^\star$. Intuitively, this corresponds to designs in which all variance-relevant directions of $R_n^\star$ are already visible in the diagonal residual moments $\{C_{hh}\}$, so that off-diagonal residual blocks $\{C_{hh'}:h\neq h'\}$ do not contribute additional information for the variance functional.

In terms of Section~\ref{sec:motivation}, condition~\eqref{eq:part-exact} is precisely a statement that Failure 2 is absent for this design: the cross-cluster moments $\{C_{hh'}: h\neq h'\}$ are redundant with, rather than informative beyond, the diagonal moments $\{C_{hh}\}$, so discarding them costs nothing. This is a property of the projection geometry alone --- it does not require homoskedasticity or any restriction on $\{\Sigma_g\}$ --- and it fails precisely when the mechanisms of Section~\ref{sec:motivation} (rich controls, shared latent factors, dense non-nested fixed effects) place genuine identifying information for $V$ in the off-diagonal blocks.

\begin{lemma}[Exact equivalence of full and partial Riesz systems]\label{lem:part-exact}
Suppose Assumption~\ref{asmp:A5} holds and there exists $D_n^{\mathrm{diag}}\in\mathcal K_n^{\mathrm{diag}}$ such that $(\mathcal A_n^{\mathrm{diag}})^\ast D_n^{\mathrm{diag}} = R_n^\star$. Then for any covariance array $\Sigma_n$,
    \[ \big\langle \mathcal A_n(\Sigma_n), D_n^\star \big\rangle_{\mathcal K_n} = \big\langle \mathcal A_n^{\mathrm{diag}}(\Sigma_n), D_n^{\mathrm{diag}} \big\rangle_{\mathcal K_n^{\mathrm{diag}}}. \]
In particular, for $C_n := \mathcal A_n(\Sigma_n)$ and $C_n^{\mathrm{diag}} := \mathcal A_n^{\mathrm{diag}}(\Sigma_n)$,
    \[ \mathrm{Var}(a'\widehat\beta \mid X, W) = \big\langle C_n, D_n^\star\big\rangle_{\mathcal K_n} = \big\langle C_n^{\mathrm{diag}}, D_n^{\mathrm{diag}}\big\rangle_{\mathcal K_n^{\mathrm{diag}}}. \]
\end{lemma}

Thus, whenever~\eqref{eq:part-Riesz} is solvable, the partial and full Riesz representations of the variance coincide exactly at fixed $n$. This includes, for example, block-diagonal projections $M$ (no cross-cluster residual leakage), in which case all off-diagonal cross-moments are identically zero and $R_n^\star$ automatically lies in $\mathrm{Range}\big((\mathcal A_n^{\mathrm{diag}})^\ast\big)$.

\paragraph{First-order equivalence.}

Beyond exact equality, in more general designs the partial system may fail to reproduce $R_n^\star$ exactly: certain components of $R_n^\star$ may be driven primarily by off-diagonal residual interactions and therefore be difficult or impossible to represent using diagonal weights alone. In such cases the full Riesz system remains well defined, while the diagonal restriction~\eqref{eq:part-Riesz} must be interpreted as an approximation.

Nevertheless, the partial system can still deliver a \emph{first-order equivalent} Riesz estimator when it reproduces the variance-relevant components of $R_n^\star$ at the scale governing studentization and when the associated diagonal weights remain well behaved. Formally, this requires that the diagonal adjoint approximates the full Riesz representer up to an asymptotically negligible error:
\begin{assumption}[Partial approximation] \label{asmp:PartApprox}
\begin{equation} \label{eq:PartApprox}
    \big\| (\mathcal A_n^{\mathrm{diag}})^\ast D_{n}^{\mathrm{diag}} - R_n^\star \big\|_{\mathcal H_n} = o_p \Big( \mu_n^{-1} \! \sqrt{\tfrac{n}{\sup_g N_g}} \Big), \qquad \|D_n^{\mathrm{diag}}\|_{\mathcal K_n} = O_p\!\Big(\sqrt{\tfrac{\sup_g N_g}{n}} \Big).
\end{equation}
\end{assumption}

The first norm bound guarantees that $(\mathcal A_n^{\mathrm{diag}})^\ast D_{n}^{\mathrm{diag}}$ approximates the full Riesz representer at the $\mu_n^{-1}$ scale that determines the first-order behavior of the studentized statistic. Importantly, this condition is automatically implied by closeness of the weights in the $\mathcal K_n$ geometry. Indeed, the bound holds whenever\footnote{ Since $\mathcal A_n^{\mathrm{diag}} = P_{\mathrm{diag}} \circ \mathcal A_n$ with $P_{\mathrm{diag}}$ an orthogonal projection, we have $(\mathcal A_n^{\mathrm{diag}})^\ast = \mathcal A_n^\ast \circ P_{\mathrm{diag}}$ and $\|(\mathcal A_n^{\mathrm{diag}})^\ast \|_{\mathrm{op}} \le \|\mathcal A_n^\ast \|_{\mathrm{op}} \, \|P_{\mathrm{diag}}\|_{\mathrm{op}} \le 1$. Therefore, $\big\|(\mathcal A_n^{\mathrm{diag}})^\ast D_n^{\mathrm{diag}} - R_n^\star \big\|_{\mathcal H_n} = \big\| (\mathcal A_n^{\mathrm{diag}})^\ast (D_n^{\mathrm{diag}} - D_n^\star) \big\|_{\mathcal H_n} \le \|D_n^{\mathrm{diag}} - D_n^\star\|_{\mathcal K_n}$.}
    \begin{equation} \label{eq:PartApprox-suff}
        \|D_n^{\mathrm{diag}} - D_n^\star\|_{\mathcal K_n} = o_p \Big( \mu_n^{-1} \! \sqrt{\tfrac{n}{\sup_g N_g}} \Big).
    \end{equation}

The second norm bound on $D_n^{\mathrm{diag}}$ is mild and auxiliary: it controls the second-order sampling variability of the partial estimator and ensures that empirical cross-moment fluctuations do not dominate the studentized statistic. It does not restrict the form of the partial estimator and does not require the partial system to be fully identified. As a sufficient condition, it is implied by a diagonal analogue of Assumption~\ref{asmp:A5}---namely, restricted invertibility and stability of $(\mathcal A_n^{\mathrm{diag}})^\ast$ on its range---but such a condition is not imposed here.

\begin{lemma}[First-order equivalence of full and partial Riesz estimators] \label{lem:part-firstorder}
Let $D_n^\star$ be the minimum-norm solution to $\mathcal A_n^\ast D = R_n^\star$, and let $D_n^{\mathrm{diag}}\in\mathcal K_n^{\mathrm{diag}}$ satisfy Assumption~\ref{asmp:PartApprox}. Define
    \[ \widehat V_{\mathrm{Riesz},n} := \langle \widehat C_n, D_n^\star\rangle_{\mathcal K_n}, \qquad \widehat V_{\partial,n} := \langle \widehat C_n^{\mathrm{diag}}, D_n^{\mathrm{diag}}\rangle_{\mathcal K_n^{\mathrm{diag}}}, \]
where $\widehat C_n^{\mathrm{diag}} = P_{\mathrm{diag}}(\widehat C_n)$. Under Assumptions~\ref{asmp:A1}--\ref{asmp:A5} and~\ref{asmp:PartApprox},
    \[ \widehat V_{\partial,n} - \widehat V_{\mathrm{Riesz},n} = o_p(\mu_n^{-1}), \qquad \mu_n\big(\widehat V_{\partial,n} - \widehat V_{\mathrm{Riesz},n}\big) = o_p(1). \]
Hence $\widehat V_{\partial,n}$ and $\widehat V_{\mathrm{Riesz},n}$ share the same probability limit and induce the same asymptotic distribution for the studentized statistic $\widehat V^{-1/2} a' (\widehat\beta_n - \beta)$.
\end{lemma}

Importantly, Assumption~\ref{asmp:PartApprox} is \emph{not} required for validity of the full Riesz estimator. Rather, it characterizes regimes in which diagonal restrictions are asymptotically adequate. When this approximation fails---most notably when the residual projection matrix exhibits substantial cross-cluster leakage---the partial system may discard first-order variance information, while the full Riesz system remains well posed and continues to deliver correct inference.

For reference, Table~\ref{tab:notation-part} fixes notation for the partial-Riesz objects used throughout the remainder of this section, since closely related quantities at different stages of approximation are easy to conflate on a first read.

\begin{table}[h]
\centering
\caption{Notation for partial-Riesz objects (Section~\ref{sec:partial})}
\label{tab:notation-part}
\begin{tabular}{ll}
\toprule
Symbol & Meaning \\
\midrule
$D_n^\star$ & Minimum-norm \emph{full} Riesz representer, $\mathcal A_n^\ast D_n^\star = R_n^\star$ \\
$D_n^{\mathrm{diag}}$ & Any representer satisfying Assumption~\ref{asmp:PartApprox} (exact or first-order) \\
$d_{\mathrm{full}}, d_{\mathrm{part}}$ & Vectorized full / diagonal weight arrays (Section~\ref{subsec:matrix-rank}) \\
$d^{\mathrm{diag}}_i$ & Exact scalar diagonal Riesz weight, $d = (M\!\odot\! M)^{-1}r$ \\
$d_0$ (CRVE) & Block-diagonal approximation $d_0 = B_0^{-1} r$, ignoring cross-cluster blocks $E$ \\
$d^{HC2}, d^{HC3}$ & Classical leverage-correction weights, $(1-h_{ii})^{-1}$ and $(1-h_{ii})^{-2}$ \\
\bottomrule
\end{tabular}
\end{table}

\subsection{Matrix representations and rank comparison} \label{subsec:matrix-rank}
To make the relationship between the full and partial Riesz systems more explicit, it is useful to write both adjoints in matrix form. Fix an ordering of blocks and define
   \[ r := \operatorname{vec}(R_n^\star) := \begin{pmatrix} \operatorname{vec} (\widetilde X_1 Q \widetilde X'_1) \\[-2pt] \vdots \\[-2pt] \operatorname{vec} (\widetilde X_G Q \widetilde X'_G) \end{pmatrix} \in \mathbb R^{\sum_g N_g^2}, \]
and weight vectors
    \[ d_{\mathrm{full}} := \begin{pmatrix} \operatorname{vec} (D_{11}) \\[-2pt] \operatorname{vec} (D_{12})\\[-2pt] \vdots\\[-2pt] \operatorname{vec}(D_{GG}) \end{pmatrix} \in \mathbb R^{\sum_{h,h'} N_h N_{h'}},
    \qquad
    d_{\mathrm{part}} := \begin{pmatrix} \operatorname{vec}(D_{11})\\[-2pt] \vdots \\[-2pt] \operatorname{vec}(D_{GG}) \end{pmatrix} \in \mathbb R^{\sum_h N_h^2}. \]

Using $\operatorname{vec}(ABC) = (C' \otimes A) \operatorname{vec} (B)$, the full adjoint satisfies
    \[ \operatorname{vec} \big( \big[ \mathcal A_n^\ast(D) \big]_g \big) = \sum_{h=1}^G \sum_{h'=1}^G (M_{gh'}\otimes M_{gh}) \operatorname{vec}(D_{hh'}), \]
so that stacked over $g$ we obtain
    \[ r = H_{\mathrm{full}} \, d_{\mathrm{full}}, \]
where $H_{\mathrm{full}}$ is the block matrix with $(g,(h,h'))$ block
    \[ H_{\mathrm{full}}(g;h,h') := M_{gh'} \otimes M_{gh} \in \mathbb R^{N_g^2 \times N_h N_{h'}}. \]

Similarly, the partial adjoint satisfies
    \[ \operatorname{vec} \big( \big[ (\mathcal A_n^{\mathrm{diag}})^\ast (D^{\mathrm{diag}}) \big]_g \big) = \sum_{h=1}^G (M_{gh} \otimes M_{gh}) \operatorname{vec}(D_{hh}), \]
so that
    \[ r = H_{\mathrm{part}}\, d_{\mathrm{part}}, \qquad H_{\mathrm{part}}(g;h) := M_{gh} \otimes M_{gh} \in \mathbb R^{N_g^2 \times N_h^2}. \]
Thus
    \[ \mathcal A_n^\ast \; \longleftrightarrow \; H_{\mathrm{full}}, \qquad (\mathcal A_n^{\mathrm{diag}})^\ast \; \longleftrightarrow \; H_{\mathrm{part}}, \]
and
    \[ \mathrm{Range} (\mathcal A_n^\ast) = \mathrm{col} (H_{\mathrm{full}}), \qquad \mathrm{Range} \big( (\mathcal A_n^{\mathrm{diag}})^\ast \big)  = \mathrm{col}(H_{\mathrm{part}}). \]

By construction, $H_{\mathrm{part}}$ is obtained from $H_{\mathrm{full}}$ by retaining only the columns corresponding to diagonal blocks $D_{hh}$. Hence it always hold that
    \[ \mathrm{col} (H_{\mathrm{part}}) \subseteq \mathrm{col}(H_{\mathrm{full}}), \qquad \operatorname{rank} \big( (\mathcal A_n^{\mathrm{diag}})^\ast \big) \le \operatorname{rank}(\mathcal A_n^\ast). \]

Whether this inclusion is strict depends on the projection geometry.
In scalar designs ($N_g=1$), the two column spaces may coincide under a mild richness condition on the rows of the projection factor $Q$, in which case the diagonal restriction does not discard any variance-relevant directions. In contrast, in clustered designs, the inclusion is generically strict when the cross-cluster projection blocks $\{M_{gh}:g\neq h\}$ carry non-negligible mass. In such settings, the off-diagonal columns of $H_{\mathrm{full}}$ generate additional directions that cannot be reproduced by diagonal blocks alone, reflecting the extra degrees of freedom provided by off-diagonal weights $\{D_{hh'}:h\neq h'\}$.

\paragraph{Independent-observations case: when do full and partial systems coincide?}

Consider the scalar-cluster design ($G=n$, $N_g=1$), so that all covariance and residual blocks are scalars. In this case, the projection matrix admits the factorization $M=QQ'$, where $Q\in\mathbb R^{n\times m}$ has orthonormal columns and $m=\operatorname{rank}(M)$. Let $q_i'\in\mathbb R^m$ denote the $i$th row of $Q$, so that
    \[ M_{ij}=q_i' q_j. \]

The columns of the partial and full adjoint systems admit a simple geometric interpretation. Specifically,
    \[ H_{\mathrm{part}}(i;j)=M_{ij}^2=(q_i' q_j)^2, \qquad H_{\mathrm{full}}(i;j,j')=M_{ij}M_{ij'}=(q_i'q_j)(q_i'q_{j'}), \]
which correspond, respectively, to evaluating the quadratic form $q\mapsto (q'q_j)^2$ and the mixed quadratic form $q\mapsto (q'q_j)(q'q_{j'})$ at $q=q_i$. Thus, both systems generate vectors of the form $i\mapsto q_i' S q_i$, where $S$ ranges over symmetric matrices. The distinction lies in the set of admissible quadratic forms: the partial system is restricted to rank-one forms $S=q_j q_j'$, while the full system additionally allows mixed forms generated by pairs $(q_j,q_{j'})$.

\begin{assumption}[Richness] \label{cond:richness}
    \[ \mathrm{span}\{q_j q_j'\}_{j=1}^n = \mathbb S^m. \]
\end{assumption}

Condition~\ref{cond:richness} requires that the rank-one forms generated by the rows of $Q$ span the entire space of $m\times m$ symmetric matrices. When it holds, every mixed quadratic form $q\mapsto (q'q_j)(q'q_{j'})$ can be written as a linear combination of squared forms $q\mapsto (q'q_k)^2$, so the additional columns in $H_{\mathrm{full}}$ generate no new directions beyond those already spanned by $H_{\mathrm{part}}$, and
    \[ \operatorname{col}(H_{\mathrm{full}})=\operatorname{col}(H_{\mathrm{part}}), \qquad \operatorname{rank}(H_{\mathrm{full}})=\operatorname{rank}(H_{\mathrm{part}}). \]

Rank equality, however, is \emph{not automatic}. If the rows $\{q_i\}$ are highly structured or concentrated on a small number of directions, the squared forms $\{(q_i' q_j)^2\}$ may fail to span all mixed quadratic interactions. In such cases, the full system can generate variance-relevant directions that are unavailable under the diagonal restriction, and $H_{\mathrm{full}}$ may have strictly larger column rank than $H_{\mathrm{part}}$.

The richness condition above holds generically once $n \ge \binom{m+1}{2}$, the dimension of $\mathbb S^m$. Indeed, $\mathrm{span}\{q_jq_j'\}_{j=1}^n \subsetneq \mathbb S^m$ if and only if the symmetric tensors $\{q_j q_j'\}_{j=1}^n$ are jointly contained in some proper linear subspace of $\mathbb S^m$ --- a closed, Lebesgue-null algebraic condition on the rows $\{q_j\}$ whenever $n\ge\binom{m+1}{2}$ and the $q_j$ are not algebraically constrained to lie on a common low-dimensional variety. Under any regressor distribution that is absolutely continuous with respect to Lebesgue measure on $\mathbb R^k$, the induced rows $q_j = Q'e_j$ admit a density on their support, so the failure set has probability zero. The condition can fail only when the regressors are deterministically structured --- for example, dummy-variable designs in which many rows $q_j$ are exactly repeated, concentrating $\{q_jq_j'\}$ on a small number of distinct directions (see the discussion of high-dimensional HC immediately below).

In summary, for independent observations, diagonal (HC-type) information is typically sufficient to recover the full Riesz representer. Exact equivalence between the full and partial systems, however, relies on the richness condition rather than holding universally, and the condition can fail in designs with repeated or near-repeated rows of $Q$.

\subsection{HC and CRVE as partial Riesz solutions/approximations} \label{subsec:HC-CRVE}

The partial Riesz system above provides a natural bridge between our general variance estimator and the classical HC and CRVE families. In this subsection, we show that these familiar estimators arise as special \emph{diagonal} or block-diagonal weighting schemes within the Riesz framework --- concrete instances of the exact and first-order equivalence results of Lemmas~\ref{lem:part-exact}--\ref{lem:part-firstorder} --- and we clarify when such restrictions yield exact solutions and when they operate only as approximations.

\subsubsection{High-dimensional HC as scalar-cluster partial Riesz solution}

When each cluster contains a single observation ($G=n$, $N_g=1$), the covariance blocks reduce to scalars $\{\sigma_i^2\}_{i=1}^n$, and the residual cross-moments satisfy
    \[ c_j \;=\; \sum_{i=1}^{n} M_{ji}^2 \sigma_i^2. \]
Under the diagonal moment restriction, the partial Riesz system \eqref{eq:Khatri-Rao} collapses to the inverse--Hadamard system
    \[ \tilde x_i' Q_n \tilde x_i \;=\; \sum_{j=1}^n M_{ij}^2 d_j \quad \Longleftrightarrow \quad (M\!\odot\! M)\, d = r, \]
where $M\!\odot\!M$ denotes the Hadamard square of the projection matrix $M$, $d_j$ is the partial Riesz weight attached to observation $j$, and $r_i=\tilde x_i'Q_n\tilde x_i$. When $M\!\odot\!M$ is invertible, the exact solution is
    \begin{equation} \label{eq:inv-Hadamard}
        d = (M\!\odot\! M)^{-1} r, \qquad d_i = \sum_{j=1}^{n} \kappa_{ij} r_j, \qquad \kappa_{ij} = \big[(M\!\odot\! M)^{-1}\big]_{ij}.
    \end{equation}
And the resulting partial Riesz variance estimator is
    \[ \widehat V_{\partial} = \frac{1}{n^2} \sum_{i=1}^n d_i \hat u_i^2 = a' \Gamma^{-1} \left( \frac{1}{n^2} \sum_{i,j=1}^n \kappa_{ij}\, \tilde x_i \tilde x_i'\, \hat u_j^2 \right) \Gamma^{-1} a. \]

This estimator has the same inverse--Hadamard structure as that of \citet{cattaneo2018inference}. Their implementation constructs diagonal weights using $(M_W\!\odot\!M_W)^{-1}$, where $M_W$ partials out the high-dimensional controls $W$. In contrast, the natural projection matrix entering the Riesz system is, by the Frisch--Waugh--Lovell theorem, the full OLS residual-maker
    \[ M:=M_{\widetilde X} M_W = M_{[X,W]}. \]

In the scalar-cluster case, partial Riesz estimators constructed from different residual projections are nevertheless first-order equivalent. Appendix Lemma~\ref{applem:M-MW} shows that the full OLS residual-maker $M_{[X,W]}$ is a low-rank perturbation of $M_W$, implying
    \[ \| (M\!\odot\! M) - (M_W\!\odot\! M_W) \|_F^2 = o_p(n). \]
Under standard inverse--Hadamard stability conditions---namely, uniform boundedness of the operator norms of $(M\!\odot\!M)^{-1}$ and $(M_W\!\odot\!M_W)^{-1}$---this perturbation yields
    \[ \| d(M) - d(M_W) \|_2 = o_p(n), \]
where $d(M)$ and $d(M_W)$ denote the scalar Riesz weight vectors obtained using $M_{[X,W]}$ and $M_W$, respectively. Since in the scalar design the Riesz geometry satisfies $\|D\|_{\mathcal K_n} = n^{-1}\|d\|_2$, it follows that
    \[ \| D(M) - D(M_W) \|_{\mathcal K_n} = o_p(1), \]
which is equivalent to the sufficient partial-approximation condition \eqref{eq:PartApprox-suff} in the scalar-cluster case, where $\mu_n=n$ and $\sup_g N_g=1$.

Consequently, replacing the full OLS projection $M_{[X,W]}$ by the residual-maker $M_W$ is asymptotically negligible at the level of the variance functional in the scalar-cluster setting. The high-dimensional HC estimator of \citet{cattaneo2018inference} is therefore first-order equivalent to a particular scalar partial Riesz solution within the general Riesz framework, although the two need not coincide exactly in finite samples.

\subsubsection{Classical HC as diagonal approximation to partial Riesz}

The exact diagonal Riesz solution $d_i$ in \eqref{eq:inv-Hadamard} aggregates information from all observations through the inverse--Hadamard square of the OLS projection matrix. Under fixed-dimensional and well-conditioned designs, standard projection geometry implies $\max_i h_{ii}=O_p(n^{-1})$, so the inverse--Hadamard operator is asymptotically diagonally dominant. In this regime, the exact diagonal Riesz weight admits the expansion
    \[ d_i = \frac{r_i}{(1-h_{ii})^2} + O_p\!\left( n^{-1} \|r\|_1 \right), \qquad \|r\|_1 := \sum_{j=1}^n |r_j|.  \]
The leading term is governed by the local leverage correction $(1-h_{ii})^{-2}$, while the remainder captures projection spillovers induced by off-diagonal elements of $M$. A formal statement and proof are given in Appendix Lemma~\ref{applem:diag-expansion}.

Classical HC estimators arise by replacing the exact diagonal Riesz weight with progressively simpler diagonal approximations to this leading term. In particular,
\begin{align*}
    \text{HC2:} &\quad w_i = (1-h_{ii})^{-1},
        &\qquad& \text{first-order leverage correction},\\
    \text{HC3:} &\quad w_i = (1-h_{ii})^{-2},
        &\qquad& \text{second-order (diagonal Riesz) approximation},\\
    \text{HC4--HC5:} &\quad w_i = (1-h_{ii})^{-\delta_i},
        &\qquad& \text{inflated diagonal envelope}.
\end{align*}

HC3 exactly matches the leading diagonal component of the Riesz solution. Under our assumptions, $\|r\|_2 = O_p(\sqrt n)$, implying $\|r\|_1 \le \sqrt n\,\|r\|_2 = O_p(n)$. Consequently, the remainder term $n^{-1}\|r\|_1$ is generically $O_p(1)$, so that, at the \emph{weight level},
    \[ d^{HC3}_i - d_i^{\mathrm{diag}} = O_p(1) \quad\text{uniformly in } i, \]
and hence
    \[ \|d^{HC3} - d^{\mathrm{diag}}\|_2 = O_p(\sqrt n) = o_p(n). \]
Because the scalar-cluster Riesz geometry satisfies $\|D\|_{\mathcal K_n} = n^{-1}\|d\|_2$, this implies
   \[ \|D^{HC3} - D^{\mathrm{diag}}\|_{\mathcal K_n} = o_p(1), \]
the same sufficient condition~\eqref{eq:PartApprox-suff} verified for high-dimensional HC above. Therefore, although HC3 does not recover the exact diagonal Riesz weights, it is sufficiently close in the $\mathcal K_n$ geometry to yield a first-order equivalent variance estimator under fixed-dimensional asymptotics.

HC2 retains only the linear leverage correction and therefore differs from the leading diagonal Riesz term by an $O_p(n^{-1})$ amount since a Taylor expansion yields
    \[ (1-h_{ii})^{-2}=(1-h_{ii})^{-1}+h_{ii}+O(h_{ii}^2), \qquad h_{ii}=O_p(n^{-1}) \]
This difference at the weight level is also negligible so the HC2 estimator is first-order equivalent under fixed-dimensional asymptotics.

HC4--HC5 further inflate the diagonal leverage correction as conservative finite-sample safeguards against extreme leverage points. These estimators no longer aim to approximate the Riesz weights themselves; instead, they modify the variance estimator directly by amplifying leverage--variance interactions, without a direct correspondence to any Riesz representer.

\subsubsection{CRVEs as block-diagonal approximation to partial Riesz}

In clustered designs, the diagonal moment restriction leads to the block-diagonal partial Riesz system \eqref{eq:Khatri-Rao}. Writing $B$ with $(g,h)$ block $B_{gh}:=M_{gh}\otimes M_{gh}$, the system can be stacked as
    \[ B\, d = r . \]
Decompose $B=B_0+E$, where $(B_0)_{gg}=M_{gg}\otimes M_{gg}$ and $E_{gh}=B_{gh}$ for $g\neq h$. Provided that the block-diagonal mapping $B_0$ is uniformly stable in the sense that $ \| B_0^{-1} \|_{\mathrm{op}} = O_p(1)$, or equivalently $\sup_g \|M_{gg}^{-1}\|_{\mathrm{op}} = O_p(1)$, so that $d_0 = B_0^{-1} r$ is well defined, a standard perturbation argument yields
    \begin{equation}\label{eq:pert-bound}
        \| d - d_0 \|_2 \; \le \; \| (I+B_0^{-1} E)^{-1} \|_{\mathrm{op}} \, \| B_0^{-1} E \|_{\mathrm{op}} \, \| B_0^{-1} r \|_2.
    \end{equation}

In words, \eqref{eq:pert-bound} says that block-diagonal CRVE is reliable exactly when cross-cluster projection leakage is small relative to how invertible the within-cluster blocks are, scaled by how quickly the target variance itself shrinks: leakage that is large in absolute terms can still be harmless if it is small relative to $\mu_n^{-1}$, while even modest leakage becomes first-order when $\mu_n$ grows slowly or within-cluster leverage is itself close to singular.

In particular, if $\|B_0^{-1}E\|_{\mathrm{op}} = o_p(1)$, then $\| (I + B_0^{-1} E)^{-1} \|_{\mathrm{op}} = O_p(1)$ and
   \[ \|d - d^{(0)} \|_2 = O_p(1) \, \| B_0^{-1} E\|_{\mathrm{op}} \, \|r\|_2.\]
Since $\|B_0^{-1}r\|_2 \;\le\; \|B_0^{-1}\|_{\mathrm{op}}\,\|r\|_2$.  By  $\|r\|_2 = O_p (\sqrt{n \sup_g N_g})$ and $\|D\|_{\mathcal K_n} = n^{-1} \|d\|_2$, it follows that a sufficient condition for the sufficient partial-approximation condition  \eqref{eq:PartApprox-suff} is
    \[  \|B_0^{-1}E\|_{\mathrm{op}} = o_p\!\Big(\frac{n}{\mu_n \sup_g N_g} \Big). \]

To ensure both invertibility of the perturbed system and the required variance-level rate, we therefore impose the joint leakage condition
    \begin{equation} \label{eq:small-leakage}
          \|B_0^{-1}E\|_{\mathrm{op}} = o_p(1) \qquad \text{and} \qquad \|B_0^{-1}E\|_{\mathrm{op}} = o_p\!\Big(\frac{n}{\mu_n \sup_g N_g} \Big).
    \end{equation}
The first condition guarantees that the Neumann expansion underlying \eqref{eq:pert-bound} is well behaved, while the second ensures that the induced approximation error is negligible at the $\mu_n^{-1}$ scale governing studentization.

This condition is the cluster-level analogue of Failure 1: $\|B_0^{-1}E\|_{\mathrm{op}}$ measures exactly the cross-cluster projection leakage identified in Section~\ref{sec:motivation} as the source of projection spillovers, now normalized by within-cluster invertibility rather than by $(1-h_{ii})$. The three mechanisms generating non-negligible off-diagonal mass there --- many controls, shared latent factor directions, and dense connectivity in non-nested fixed effects --- are exactly the mechanisms that inflate $E$ and threaten~\eqref{eq:small-leakage} here.

A convenient sufficient (row-sum) condition implying this bound is
    \begin{equation}\label{eq:block-sparse-rate}
       \max_{g}\sum_{h\neq g} \big\|(M_{gg}^{-1}M_{gh})\otimes(M_{gg}^{-1}M_{gh})\big\|_{\mathrm{op}} \;=\; o_p(1) \quad \text{and} \quad o_p\!\Big(\frac{n}{\mu_n \sup_g N_g} \Big), \tag{Proj-Sparsity}
    \end{equation}
which requires the off-diagonal projection leakage, normalized by within-cluster leverage, to be negligible both absolutely and at the variance-relevant rate. Section~\ref{subsec:diagnostic} below gives a computable diagnostic for assessing this condition without forming $M$ explicitly.

Under this strengthened block-sparsity condition, the partial Riesz system decouples at first order across clusters, yielding the cluster-level approximation
    \[ \operatorname{vec}(D_{gg}) \approx (M_{gg}\otimes M_{gg})^{-1} \operatorname{vec}(\widetilde X_g Q_n \widetilde X_g') \quad\Longleftrightarrow\quad D_{gg} \approx M_{gg}^{-1}(\widetilde X_g Q_n \widetilde X_g')M_{gg}^{-1}. \]

Classical CRVEs arise by replacing the exact cluster-level inverse $M_{gg}^{-1}$ with progressively simpler leverage corrections. In particular, CR0 and CR1 ignore leverage adjustments (up to degrees-of-freedom scaling), CR2 corresponds to the approximation $M_{gg}^{-1} \approx (I_{N_g}-H_{gg})^{-1/2}$, and CR3 corresponds to using the full inverse $M_{gg}^{-1}=(I_{N_g}-H_{gg})^{-1}$. Here $H_{gg}=\widetilde X_g(\widetilde X'\widetilde X)^{-1}\widetilde X_g'$ denotes the $(g,g)$ block of the global hat matrix, so that $M_{gg} = I_{N_g} - H_{gg}$. Substituting these block-diagonal approximations into \eqref{eq:Vhat-Riesz} reproduces the familiar CR0--CR3 formulas.

However, the OLS projection matrix $M$ is generally not block diagonal across clusters---even under fixed-dimensional regressions. In the scalar case, the two rates coincide as $\mu_n = n$ and $N_g = 1$ for all $g$, and projection identities imply $\sum_{j\neq i}M_{ij}^2=O_p(n^{-1})$, and hence $o_p(1)$, automatically under fixed-dimensional asymptotics,
so the inverse--Hadamard system is diagonally dominant without further restrictions.
In contrast, with multi-unit clusters, fixed dimensionality alone does not control the normalized aggregate mass of the off-diagonal Khatri--Rao blocks $(M_{gg}^{-1} M_{gh}) \otimes  (M_{gg}^{-1} M_{gh})$. Additional projection-sparsity or rate conditions are therefore required for block-diagonal CRVEs to approximate the partial Riesz solution.

When such projection-sparsity conditions fail, block-diagonal CRVEs may discard first-order variance information contained in the diagonal moment equations through cross-cluster projection leakage. In contrast, the exact partial Riesz estimator remains valid for the diagonal information set by construction, and the full Riesz estimator remains valid without any diagonal restriction.

\subsection{A practical diagnostic for projection sparsity} \label{subsec:diagnostic}

The validity of block-diagonal CRVE approximations hinges on projection sparsity, i.e.\ the negligibility of the normalized off-diagonal Khatri--Rao blocks $\{(M_{gg}^{-1} M_{gh}) \otimes (M_{gg}^{-1} M_{gh}): g\neq h\}$. The $\mu_n$-free row-sum bound in~\eqref{eq:block-sparse-rate} above already gives a sufficient condition that does not require knowing the normalization rate $\mu_n$; what remains is making that bound computable in practice.

Although the projection matrix $M$ is typically too large to inspect directly, projection sparsity can be assessed indirectly using computable diagnostics. In particular, large off-diagonal leverage patterns---such as high average or maximum values of $\|M_{gg}^{-1/2}M_{gh}\|_F^2$ across $h\neq g$---signal substantial cross-cluster projection spillovers. Empirically, such spillovers are likely in designs with overlapping or non-nested fixed effects, global regressors, or weak common factors. We apply this diagnostic to the colonial governors application in Section~\ref{sec:application}, where several dyadic clusters exhibit exactly this pattern.

\subsection{The existence hierarchy} \label{subsec:hierarchy}

The results of this section order the full Riesz estimator and its classical restrictions by the strength of the condition each requires to exist and remain first-order valid:

\begin{itemize}
    \item \textbf{Full Riesz.} Exists and is asymptotically valid under Assumption~\ref{asmp:A5} alone---exact range membership of $R_n^\star$ in $\mathrm{Range}(\mathcal A_n^\ast)$, the weakest condition in the hierarchy, since $\mathcal A_n^\ast$ uses the entire residual moment array.
    \item \textbf{Partial (diagonal) Riesz.} Exists exactly only under the strictly stronger range condition~\eqref{eq:part-exact}, or first-order only under Assumption~\ref{asmp:PartApprox}; both fail generically once cross-cluster projection mass is non-negligible and informative (Failure 2 of Section~\ref{sec:motivation}).
    \item \textbf{Block-diagonal CRVE (CR2/CR3).} Requires, in addition, that each within-cluster block $M_{gg}$ be invertible and that the projection-sparsity condition~\eqref{eq:block-sparse-rate} hold---strictly stronger still, and the condition that fails in Section~\ref{sec:application}'s nested-dyad clusters, where $M_{gg}$ is singular and CR2/CR3 are undefined.
\end{itemize}

Each restriction in this hierarchy trades robustness for computational simplicity by discarding part of the residual moment array; the conditions under which that trade is costless are precisely those characterized by Lemmas~\ref{lem:part-exact}--\ref{lem:part-firstorder} and the projection-sparsity condition~\eqref{eq:block-sparse-rate} above. The full Riesz estimator is the unique member of the hierarchy that requires none of these additional restrictions, which is why it remains valid in Section~\ref{sec:application} when CR2 and CR3 are not.
\section{Simulations} \label{sec:simulation}
This section evaluates the finite-sample performance of the proposed Riesz variance estimator against classical alternatives in designs that reflect the three sources of projection spillovers identified in Section~\ref{sec:motivation}: (i) many controls; and (ii) shared latent factor directions.The simulation designs draw on \citet{cattaneo2018inference} and \citet{djogbenou2019asymptotic}.


\subsection{Data Generating Process}

Following \citet{djogbenou2019asymptotic}, we allocate $n$ observations among $G$ clusters according to
    \[ N_g = \left[ \frac{n \exp(\lambda g/G)}{\sum_{j=1}^{G} \exp(\lambda j/G)} \right], \qquad g=1,\ldots,G-1, \qquad N_G = n - \sum_{g=1}^{G-1} N_g, \]
where $[\cdot]$ denotes the integer part of the argument and $\lambda\ge 0$ controls size imbalance. When $\lambda=0$ and $n/G$ is an integer, $N_g=n/G$ for all $g$; as $\lambda$ increases, cluster sizes become more unequal.

For unit $i$ in cluster $g$, the outcome is generated as
    \[ y_{gi} = x_{gi}\beta + w_{gi}'\gamma + u_{gi}, \]
where $\beta$ is the scalar parameter of interest, $\gamma \in \mathbb R^{\ell}$ collects nuisance coefficients, and $w_{gi} \in \mathbb R^{\ell}$ includes an intercept. We set $\beta=1$ and $\gamma = 0 \cdot \mathbf 1_{\ell}$.

Within-cluster dependence is induced by a common cluster shock:
    \[ u_{gi} = \sqrt{\rho_u} \, \varepsilon_g + \sqrt{1-\rho_u} \, \varepsilon_{gi}, \qquad \rho_u \in [0,1], \qquad \varepsilon_g \stackrel{iid}{\sim} \mathcal N(0,1), \]
with the idiosyncratic component
    \[ \varepsilon_{gi} \mid X,W \stackrel{iid}{\sim} \mathcal N \! \left( 0, \sigma_\varepsilon^2 (x_{gi},w_{gi}) \right). \]
Define $t(a)=a$ if $|a|\le 2$ and $t(a)=2\,\mathrm{sgn}(a)$ otherwise. For $\vartheta\in\{0,1\}$ (0 = homoskedastic, 1 = moderate heteroskedasticity), let
    \[ \sigma^2_{\varepsilon} (x_{gi},w_{gi}) = \kappa_{\varepsilon} \left[ 1 + \big(t(x_{gi}) + \iota'w_{gi}\big)^2 \right]^{\vartheta}, \]
where $\iota$ is a conformable vector of ones. The scaling constant $\kappa_{\varepsilon}$ is chosen so that $\mathrm{var} (\varepsilon_{gi})\approx 1$ unconditionally, ensuring that differences in estimator performance are not driven by trivial changes in noise scale.

\subsubsection*{Design 1: High-dimensional controls (HD)}

This design isolates the high-dimensional channel $\ell/n\to\alpha$ in the spirit of \citet{cattaneo2018inference}. Let $\ell=\ell(n)$ grow with $n$\footnote{More precisely, $\ell$ should scale with the effective sample size $\mu_n$, which depends on $n$, $G$, $\{N_g\}$, and within-cluster variation, and is bounded by $n/\sup_g N_g$.}. Fix $\rho_w\in[0,1]$ and generate controls as
    \[ w_{gi} = \big( 1, w_{gi}^{(2)}, \ldots, w_{gi}^{(\ell)} \big)', \qquad w_{gi}^{(j)} = \sqrt{\rho_w} \, z_g^{(j)} + \sqrt{1-\rho_w} \, z_{gi}^{(j)}, \quad j\ge 2,  \]
where $\{z_g^{(j)}, z_{gi}^{(j)}\}$ are i.i.d.\ across $(g,i,j)$ from one of the
following distributions: (i) uniform, $\mathcal U(-1,1)$; (ii) normal, $\mathcal N(0,1)$; and (iii) discrete, $\mathbf 1\{\tilde z\ge 2.5\}$ with $\tilde z\sim\mathcal N(0,1)$.

The regressor of interest is generated as
     \[ x_{gi} = \sqrt{\rho_x} \, v_g + \sqrt{1-\rho_x} \, v_{gi}, \qquad \rho_x \in [0,1], \]
where $v_g \stackrel{iid}{\sim} \mathcal N(0,1)$ and $v_{gi} \mid W \stackrel{iid}{\sim} \mathcal N(0,\sigma_v^2(w_{gi}))$ with
    \[ \sigma_v^2(w_{gi}) = \kappa_v \left[ 1 + (\iota'w_{gi})^2 \right]^{\vartheta}. \]
The scaling constant $\kappa_v$ is chosen so that $\mathrm{var}(v_{gi})\approx 1$ unconditionally.

\subsubsection*{Design 2: Common latent factors (LF)}
This design generates projection spillovers through shared low-rank directions, in the spirit of interactive factor structures. A global factor is realized at the unit level and enters regressors with cluster-specific exposures (loadings), so that clusters are aligned along common directions in regressor space even when $\ell$ is moderate.

Let $f_{n(g,i)} \stackrel{iid}{\sim} \mathcal{N}(0,1)$ denote a common factor shared across clusters.
For each cluster $g$, draw loadings
   \[ \lambda_{g}^{(x)} \stackrel{iid}{\sim} Rademacher, \qquad \lambda_{g}^{(j)} \stackrel{iid}{\sim}  Rademacher, \quad j = 2,\ldots,\ell, \]
independent of $\{z_g^{(j)},z_{gi}^{(j)},v_g,v_{gi}\}$.

For $j\ge 2$, generate controls as
    \[ w_{gi}^{(j)} = \sqrt{\delta_w}\,\lambda_g^{(j)} f_{n(g,i)} + \sqrt{1-\delta_w}\,\tilde w_{gi}^{(j)}, \quad \tilde w_{gi}^{(j)} = \sqrt{\rho_w}\, z_g^{(j)} + \sqrt{1-\rho_w}\, z_{gi}^{(j)}. \]
Similarly, generate the regressor of interest as
    \[ x_{gi} = \sqrt{\delta_x}\,\lambda_g^{(x)} f_{n(g,i)} + \sqrt{1-\delta_x}\,\tilde x_{gi}, \quad \tilde x_{gi} = \sqrt{\rho_x}\, v_g + \sqrt{1-\rho_x}\, v_{gi}. \]
Here, $\delta_w, \delta_x \in [0,1]$ control the strength of the common factor component. When $\delta_w=\delta_x=0$, the design reduces to the baseline without shared directions, while larger values of $\delta_w$ and $\delta_x$ induce stronger alignment across observations and generate increasing dense projection structures.

\subsection{Simulation Results}

We compare the proposed Riesz estimators with CR0--CR3 (which reduce to HC0--HC3 when clusters are singletons) and the HCK estimator of \citet{cattaneo2018inference}. We report two variants: \textit{Riesz-F}, the full estimator based on the complete residual moment system, and \textit{Riesz-P}, which imposes the diagonal restriction. All methods use the same OLS point estimator $\widehat\beta$ and studentization scheme in~\eqref{eq:t-statistic}, differing only in the variance estimator.

Across all designs we conduct $R = 5{,}000$ Monte Carlo replications and report empirical size at the 5\% level and mean confidence interval length at the 95\% level. Where CR2, CR3, or HCK require matrix inversions that may fail, we also report the \textit{invalid rate} (fraction of replications where the estimator is undefined) and \textit{tail instability} (fraction of valid replications where $\widehat{SE}/SD(\hat\beta) > 10$). Size and CI length are computed over valid replications only.

\textbf{Design 1, independent observations (Tables~\ref{tab:HD-IID-1}--\ref{tab:HD-IID-2}).}
Table~\ref{tab:HD-IID-1} considers normal controls. Panel B shows that size distortions of CR0 grow sharply with $\ell/n$: at $\ell/n = 0.10$, CR0 rejects at 8.0\% against a nominal 5\%; by $\ell/n = 0.50$, the rejection rate reaches 23.4\%. CR1 and CR2 improve but remain above nominal for large $\ell/n$, while CR3 becomes increasingly conservative, falling to 1.8\% at $\ell/n = 0.50$. HCK controls size well across all $\ell/n$ values. The Riesz estimators maintain rejection rates between 5.2\% and 6.2\% throughout, with moderately wider confidence intervals than HCK, reflecting their use of the full projection geometry from both $X$ and $W$. Riesz-P and Riesz-F produce identical results in all configurations, consistent with Lemmas 5--6.

Table~\ref{tab:HD-IID-2} considers discrete controls, where implementation failures are pronounced. Panel A shows that CR2 is undefined in 49\% of replications at $\ell/n = 0.10$, rising to 97\% at $\ell/n = 0.50$; CR3 is undefined in 16\% and 62\% of replications at the same values. HCK exhibits a non-negligible invalid rate and increasing tail instability, reaching a mean CI length of 1.94 at $\ell/n = 0.50$ compared to 0.21 for the Riesz estimators. The Riesz estimators remain well-defined in all replications with stable coverage close to the nominal level throughout.

\textbf{Design 1, clustered observations (Tables~\ref{tab:HD-cluster-1}--\ref{tab:HD-cluster-2}).}
Table~\ref{tab:HD-cluster-1} considers homogeneous cluster sizes with no regressor correlation across clusters. Size distortions mirror the independent case: CR0 over-rejects substantially as $\ell/n$ grows, while Riesz estimators maintain near-nominal size. HCK performs well in this configuration where the regressor structure is simple.

Table~\ref{tab:HD-cluster-2} introduces heterogeneous cluster sizes and within-cluster regressor correlation, which is the empirically relevant setting where HCK breaks down. Under moderate correlation ($\rho_w = \rho_x = \rho_u = 0.6$), HCK rejects at 25.0\% at $\ell/n = 0.01$ --- before any many-controls problem appears --- declining to 7.6\% at $\ell/n = 0.50$ as the Hadamard correction partially compensates. The Riesz estimators maintain rejection rates between 5.2\% and 6.8\% throughout. The Riesz advantage over HCK is largest when regressor clustering is strong, because HCK uses only the projection of controls $M_W$ and is not designed to handle leverage from clustered regressors of interest.

\textbf{Design 2, independent observations (Table~\ref{tab:LF-IID}, independent).} In the low-dimensional regime ($\ell/n = 0.01$), all estimators perform similarly when the factor loading is zero. As $\delta_x = \delta_w$ increases to 0.8, CR0 over-rejects at 6.6\% while Riesz estimators remain at 5.6\% --- a modest but consistent pattern. In the moderate-dimensional regime ($\ell/n = 0.10$), the interaction between factor structure and dimension becomes visible: CR0 over-rejects at 12.2\% at $\delta = 0.8$ compared to 7.5\% at $\delta = 0.0$, while Riesz estimators remain between 4.9\% and 5.8\% across all factor strengths. This confirms that the shared-directions mechanism compounds with the many-regressors mechanism as predicted by the motivation in Section~\ref{sec:motivation}.

\textbf{Design 2, clustered observations (Table~\ref{tab:LF-cluster}, clustered).} Under clustering with weak within-cluster correlation ($\rho_u = 0.3$), size distortions follow a similar pattern. At $\ell/n = 0.10$ and $\delta = 0.8$, CR0 rejects at 12.0\% while CR3 under-rejects at 1.9\%. Riesz estimators maintain rejection rates near 5.3\%, with CI lengths that are moderately larger than CR1 and closer to CR2.

\textbf{Overall.} Three patterns emerge consistently across designs. First, CR0 and CR1 over-reject substantially as projection spillovers increase, with CR0 reaching rejection rates above 20\% in high-dimensional settings. Second, CR3 and HCK address some of this distortion but in opposite directions: CR3 becomes increasingly conservative (under-rejecting), while HCK over-rejects under clustered regressors. Third, the Riesz estimators maintain rejection rates close to the nominal 5\% level across all configurations examined, at the cost of confidence intervals that are 5--15\% wider than CR1. The equivalence of Riesz-P and Riesz-F throughout is consistent with the asymptotic equivalence established in Lemmas 5--6 for correctly specified clustered settings with continuous regressors.

\begin{table}[htbp] 
\centering
\caption{Design 1 (HD) with independent observations ($n=1000$, $G=1000$). \\ 
Controls $W$ follow the normal distribution and errors are heteroskedastic ($\vartheta=1$). \\}
\label{tab:HD-IID-1}
\begin{threeparttable}
\setlength{\tabcolsep}{1pt}
\begin{tabularx}{\textwidth}{l@{\hspace{15pt}}*{7}{>{\centering\arraybackslash}X}}
\toprule
 & CR0 & CR1 & CR2 & CR3 & HCK & Riesz-P & Riesz-F \\
\midrule
\multicolumn{8}{l}{\textit{Panel A: Invalid Rate}} \\
\addlinespace
$\ell/n = 0.01$ & 0.0000 & 0.0000 & 0.0000 & 0.0000 & 0.0000 & 0.0000 & 0.0000 \\
$\ell/n = 0.05$ & 0.0000 & 0.0000 & 0.0000 & 0.0000 & 0.0000 & 0.0000 & 0.0000 \\
$\ell/n = 0.10$ & 0.0000 & 0.0000 & 0.0000 & 0.0000 & 0.0000 & 0.0000 & 0.0000 \\
$\ell/n = 0.20$ & 0.0000 & 0.0000 & 0.0000 & 0.0000 & 0.0000 & 0.0000 & 0.0000 \\
$\ell/n = 0.30$ & 0.0000 & 0.0000 & 0.0000 & 0.0000 & 0.0000 & 0.0000 & 0.0000 \\
$\ell/n = 0.40$ & 0.0000 & 0.0000 & 0.0000 & 0.0000 & 0.0000 & 0.0000 & 0.0000 \\
$\ell/n = 0.50$ & 0.0000 & 0.0000 & 0.0000 & 0.0000 & 0.0000 & 0.0000 & 0.0000 \\
\midrule

\multicolumn{8}{l}{\textit{Panel B: Size (5\%)}} \\
\addlinespace
$\ell/n = 0.01$ & 0.0578 & 0.0562 & 0.0558 & 0.0546 & 0.0554 & 0.0540 & 0.0540 \\
$\ell/n = 0.05$ & 0.0700 & 0.0622 & 0.0620 & 0.0550 & 0.0578 & 0.0554 & 0.0554 \\
$\ell/n = 0.10$ & 0.0802 & 0.0638 & 0.0638 & 0.0478 & 0.0548 & 0.0528 & 0.0528 \\
$\ell/n = 0.20$ & 0.1060 & 0.0700 & 0.0696 & 0.0452 & 0.0536 & 0.0524 & 0.0524 \\
$\ell/n = 0.30$ & 0.1416 & 0.0782 & 0.0780 & 0.0380 & 0.0552 & 0.0530 & 0.0530 \\
$\ell/n = 0.40$ & 0.1926 & 0.0966 & 0.0970 & 0.0316 & 0.0646 & 0.0618 & 0.0616 \\
$\ell/n = 0.50$ & 0.2342 & 0.0934 & 0.0940 & 0.0180 & 0.0582 & 0.0562 & 0.0562 \\
\midrule

\multicolumn{8}{l}{\textit{Panel C: Median CI Length (95\%)}} \\
\addlinespace
$\ell/n = 0.01$ & 0.1961 & 0.1972 & 0.1974 & 0.1986 & 0.1980 & 0.1993 & 0.1993 \\
$\ell/n = 0.05$ & 0.1983 & 0.2035 & 0.2038 & 0.2095 & 0.2071 & 0.2088 & 0.2088 \\
$\ell/n = 0.10$ & 0.1928 & 0.2033 & 0.2035 & 0.2148 & 0.2100 & 0.2116 & 0.2116 \\
$\ell/n = 0.20$ & 0.1814 & 0.2029 & 0.2029 & 0.2271 & 0.2158 & 0.2175 & 0.2175 \\
$\ell/n = 0.30$ & 0.1694 & 0.2026 & 0.2027 & 0.2426 & 0.2217 & 0.2238 & 0.2238 \\
$\ell/n = 0.40$ & 0.1586 & 0.2050 & 0.2048 & 0.2646 & 0.2296 & 0.2319 & 0.2319 \\
$\ell/n = 0.50$ & 0.1486 & 0.2103 & 0.2100 & 0.2971 & 0.2389 & 0.2415 & 0.2415 \\

\midrule
\multicolumn{8}{l}{\textit{Panel D: Mean CI Length (95\%)}} \\
\addlinespace
$\ell/n = 0.01$ & 0.1993 & 0.2004 & 0.2007 & 0.2020 & 0.2014 & 0.2030 & 0.2030 \\
$\ell/n = 0.05$ & 0.2013 & 0.2066 & 0.2069 & 0.2126 & 0.2103 & 0.2122 & 0.2122 \\
$\ell/n = 0.10$ & 0.1952 & 0.2058 & 0.2060 & 0.2176 & 0.2128 & 0.2149 & 0.2149 \\
$\ell/n = 0.20$ & 0.1839 & 0.2057 & 0.2059 & 0.2305 & 0.2193 & 0.2215 & 0.2215 \\
$\ell/n = 0.30$ & 0.1713 & 0.2048 & 0.2049 & 0.2452 & 0.2246 & 0.2270 & 0.2270 \\
$\ell/n = 0.40$ & 0.1605 & 0.2073 & 0.2072 & 0.2677 & 0.2329 & 0.2356 & 0.2356 \\
$\ell/n = 0.50$ & 0.1500 & 0.2124 & 0.2121 & 0.3002 & 0.2425 & 0.2456 & 0.2456 \\
\midrule

\multicolumn{8}{l}{\textit{Panel E: Tail Instability $P(\widehat{SE}/SD(\hat{\beta})> 10)$}} \\
\addlinespace
$\ell/n = 0.01$ & 0.0000 & 0.0000 & 0.0000 & 0.0000 & 0.0000 & 0.0000 & 0.0000 \\
$\ell/n = 0.05$ & 0.0000 & 0.0000 & 0.0000 & 0.0000 & 0.0000 & 0.0000 & 0.0000 \\
$\ell/n = 0.10$ & 0.0000 & 0.0000 & 0.0000 & 0.0000 & 0.0000 & 0.0000 & 0.0000 \\
$\ell/n = 0.20$ & 0.0000 & 0.0000 & 0.0000 & 0.0000 & 0.0000 & 0.0000 & 0.0000 \\
$\ell/n = 0.30$ & 0.0000 & 0.0000 & 0.0000 & 0.0000 & 0.0000 & 0.0000 & 0.0000 \\
$\ell/n = 0.40$ & 0.0000 & 0.0000 & 0.0000 & 0.0000 & 0.0000 & 0.0000 & 0.0000 \\
$\ell/n = 0.50$ & 0.0000 & 0.0000 & 0.0000 & 0.0000 & 0.0000 & 0.0000 & 0.0000 \\

\bottomrule
\end{tabularx}
\end{threeparttable}
\end{table}

\begin{table}[htbp] 
\centering
\caption{Design 1 (HD) with independent observations ($n=1000$, $G=1000$). \\ 
Controls $W$ follow the discrete distribution and errors are heteroskedastic ($\vartheta=1$). \\}
\label{tab:HD-IID-2}
\begin{threeparttable}
\setlength{\tabcolsep}{1pt}
\begin{tabularx}{\textwidth}{l@{\hspace{15pt}}*{7}{>{\centering\arraybackslash}X}}
\toprule
 & CR0 & CR1 & CR2 & CR3 & HCK & Riesz-P & Riesz-F \\
\midrule
\multicolumn{8}{l}{\textit{Panel A: Invalid Rate}} \\
\addlinespace
$\ell/n = 0.01$ & 0.0000 & 0.0000 & 0.0660 & 0.0274 & 0.0002 & 0.0000 & 0.0000 \\
$\ell/n = 0.05$ & 0.0000 & 0.0000 & 0.2894 & 0.1038 & 0.0032 & 0.0000 & 0.0000 \\
$\ell/n = 0.10$ & 0.0000 & 0.0000 & 0.4896 & 0.1562 & 0.0196 & 0.0000 & 0.0000 \\
$\ell/n = 0.20$ & 0.0000 & 0.0000 & 0.7478 & 0.2824 & 0.0200 & 0.0000 & 0.0000 \\
$\ell/n = 0.30$ & 0.0000 & 0.0000 & 0.8754 & 0.4090 & 0.0132 & 0.0000 & 0.0000 \\
$\ell/n = 0.40$ & 0.0000 & 0.0000 & 0.9404 & 0.5154 & 0.0142 & 0.0000 & 0.0000 \\
$\ell/n = 0.50$ & 0.0000 & 0.0000 & 0.9680 & 0.6244 & 0.0122 & 0.0000 & 0.0000 \\
\midrule

\multicolumn{8}{l}{\textit{Panel B: Size (5\%)}} \\
\addlinespace
$\ell/n = 0.01$ & 0.0536 & 0.0524 & 0.0490 & 0.0452 & 0.0484 & 0.0488 & 0.0488 \\
$\ell/n = 0.05$ & 0.0730 & 0.0660 & 0.0529 & 0.0353 & 0.0482 & 0.0506 & 0.0508 \\
$\ell/n = 0.10$ & 0.0920 & 0.0780 & 0.0564 & 0.0275 & 0.0488 & 0.0532 & 0.0534 \\
$\ell/n = 0.20$ & 0.1216 & 0.0834 & 0.0563 & 0.0156 & 0.0447 & 0.0538 & 0.0538 \\
$\ell/n = 0.30$ & 0.1618 & 0.0894 & 0.0562 & 0.0108 & 0.0442 & 0.0520 & 0.0520 \\
$\ell/n = 0.40$ & 0.1986 & 0.0938 & 0.0705 & 0.0054 & 0.0452 & 0.0566 & 0.0564 \\
$\ell/n = 0.50$ & 0.2302 & 0.0866 & 0.0875 & 0.0027 & 0.0447 & 0.0538 & 0.0538 \\
\midrule

\multicolumn{8}{l}{\textit{Panel C: Median CI Length (95\%)}} \\
\addlinespace
$\ell/n = 0.01$ & 0.1513 & 0.1522 & 0.1536 & 0.1566 & 0.1539 & 0.1541 & 0.1541 \\
$\ell/n = 0.05$ & 0.1511 & 0.1551 & 0.1622 & 0.1761 & 0.1648 & 0.1641 & 0.1641 \\
$\ell/n = 0.10$ & 0.1484 & 0.1565 & 0.1678 & 0.1934 & 0.1729 & 0.1711 & 0.1711 \\
$\ell/n = 0.20$ & 0.1422 & 0.1591 & 0.1734 & 0.2189 & 0.1836 & 0.1801 & 0.1801 \\
$\ell/n = 0.30$ & 0.1356 & 0.1622 & 0.1782 & 0.2437 & 0.1908 & 0.1855 & 0.1855 \\
$\ell/n = 0.40$ & 0.1315 & 0.1699 & 0.1855 & 0.2752 & 0.2015 & 0.1942 & 0.1942 \\
$\ell/n = 0.50$ & 0.1290 & 0.1826 & 0.1998 & 0.3195 & 0.2146 & 0.2062 & 0.2062 \\

\midrule
\multicolumn{8}{l}{\textit{Panel D: Mean CI Length (95\%)}} \\
\addlinespace
$\ell/n = 0.01$ & 0.1520 & 0.1528 & 0.1552 & 0.2090 & 0.3808 & 0.1551 & 0.1551 \\
$\ell/n = 0.05$ & 0.1518 & 0.1558 & 0.1647 & 0.3469 & 0.6511 & 0.1652 & 0.1652 \\
$\ell/n = 0.10$ & 0.1492 & 0.1573 & 0.1728 & 0.5054 & 0.8597 & 0.1725 & 0.1725 \\
$\ell/n = 0.20$ & 0.1426 & 0.1596 & 0.1780 & 0.6815 & 1.0568 & 0.1811 & 0.1811 \\
$\ell/n = 0.30$ & 0.1361 & 0.1628 & 0.1846 & 0.8381 & 1.2508 & 0.1868 & 0.1868 \\
$\ell/n = 0.40$ & 0.1317 & 0.1702 & 0.1940 & 1.0368 & 1.5277 & 0.1951 & 0.1951 \\
$\ell/n = 0.50$ & 0.1291 & 0.1828 & 0.2073 & 1.2793 & 1.9358 & 0.2071 & 0.2071 \\
\midrule

\multicolumn{8}{l}{\textit{Panel E: Tail Instability $P(\widehat{SE}/SD(\hat{\beta})> 10)$}} \\
\addlinespace
$\ell/n = 0.01$ & 0.0000 & 0.0000 & 0.0002 & 0.0072 & 0.0116 & 0.0000 & 0.0000 \\
$\ell/n = 0.05$ & 0.0000 & 0.0000 & 0.0000 & 0.0272 & 0.0424 & 0.0000 & 0.0000 \\
$\ell/n = 0.10$ & 0.0000 & 0.0000 & 0.0000 & 0.0466 & 0.0672 & 0.0000 & 0.0000 \\
$\ell/n = 0.20$ & 0.0000 & 0.0000 & 0.0000 & 0.0508 & 0.0806 & 0.0000 & 0.0000 \\
$\ell/n = 0.30$ & 0.0000 & 0.0000 & 0.0000 & 0.0512 & 0.0948 & 0.0000 & 0.0000 \\
$\ell/n = 0.40$ & 0.0000 & 0.0000 & 0.0000 & 0.0518 & 0.1124 & 0.0000 & 0.0000 \\
$\ell/n = 0.50$ & 0.0000 & 0.0000 & 0.0000 & 0.0452 & 0.1254 & 0.0000 & 0.0000 \\

\bottomrule
\end{tabularx}
\end{threeparttable}
\end{table}

\begin{table}[htbp] 
\centering
\caption{Design 1 (HD) with homogeneous cluster sizes ($n=1000$, $G=200$, $\lambda=0$). \\
Controls $W$ follow the uniform distribution and errors are heteroskedastic ($\vartheta=1$). \\}
\label{tab:HD-cluster-1}
\begin{threeparttable}
\setlength{\tabcolsep}{1pt}
\begin{tabularx}{\textwidth}{l@{\hspace{15pt}}*{7}{>{\centering\arraybackslash}X}}
\toprule
 & CR0 & CR1 & CR2 & CR3 & HCK & Riesz-P & Riesz-F \\
\midrule 

\multicolumn{8}{c}{\textit{Size (5\%)}} \\
\midrule

\multicolumn{8}{l}{\textit{Weak correlation ($\rho_w=0$, $\rho_x = 0$, $\rho_u = 0.3$)}} \\
\addlinespace
$\ell/n = 0.01$ & 0.0554 & 0.0528 & 0.0534 & 0.0514 & 0.0530 & 0.0510 & 0.0510 \\
$\ell/n = 0.05$ & 0.0646 & 0.0580 & 0.0582 & 0.0516 & 0.0536 & 0.0532 & 0.0532 \\
$\ell/n = 0.10$ & 0.0734 & 0.0612 & 0.0618 & 0.0480 & 0.0538 & 0.0518 & 0.0518 \\
$\ell/n = 0.20$ & 0.1098 & 0.0732 & 0.0740 & 0.0442 & 0.0574 & 0.0576 & 0.0576 \\
$\ell/n = 0.30$ & 0.1278 & 0.0678 & 0.0688 & 0.0300 & 0.0498 & 0.0490 & 0.0490 \\
$\ell/n = 0.40$ & 0.1866 & 0.0860 & 0.0870 & 0.0240 & 0.0594 & 0.0602 & 0.0602 \\
$\ell/n = 0.50$ & 0.2154 & 0.0818 & 0.0832 & 0.0136 & 0.0534 & 0.0568 & 0.0570 \\
\midrule

\multicolumn{8}{l}{\textit{Moderate correlation ($\rho_w=0$, $\rho_x = 0$, $\rho_u = 0.6$)}} \\
\addlinespace
$\ell/n = 0.01$ & 0.0530 & 0.0510 & 0.0512 & 0.0496 & 0.0478 & 0.0490 & 0.0490 \\
$\ell/n = 0.05$ & 0.0612 & 0.0546 & 0.0548 & 0.0496 & 0.0522 & 0.0516 & 0.0516 \\
$\ell/n = 0.10$ & 0.0740 & 0.0610 & 0.0614 & 0.0460 & 0.0540 & 0.0558 & 0.0556 \\
$\ell/n = 0.20$ & 0.0988 & 0.0652 & 0.0658 & 0.0410 & 0.0540 & 0.0552 & 0.0552 \\
$\ell/n = 0.30$ & 0.1182 & 0.0628 & 0.0642 & 0.0284 & 0.0516 & 0.0514 & 0.0514 \\
$\ell/n = 0.40$ & 0.1658 & 0.0716 & 0.0732 & 0.0180 & 0.0570 & 0.0566 & 0.0566 \\
$\ell/n = 0.50$ & 0.2002 & 0.0664 & 0.0676 & 0.0128 & 0.0510 & 0.0524 & 0.0524 \\
\midrule

\multicolumn{8}{c}{\textit{Mean CI Length (95\%)}} \\
\midrule

\multicolumn{8}{l}{\textit{Weak correlation ($\rho_w=0$, $\rho_x = 0$, $\rho_u = 0.3$)}} \\
\addlinespace
$\ell/n = 0.01$ & 0.1705 & 0.1717 & 0.1714 & 0.1724 & 0.1724 & 0.1732 & 0.1732 \\
$\ell/n = 0.05$ & 0.1793 & 0.1844 & 0.1841 & 0.1890 & 0.1871 & 0.1885 & 0.1885 \\
$\ell/n = 0.10$ & 0.1759 & 0.1859 & 0.1856 & 0.1957 & 0.1912 & 0.1926 & 0.1926 \\
$\ell/n = 0.20$ & 0.1677 & 0.1880 & 0.1876 & 0.2098 & 0.1982 & 0.2000 & 0.2000 \\
$\ell/n = 0.30$ & 0.1586 & 0.1900 & 0.1894 & 0.2265 & 0.2050 & 0.2067 & 0.2067 \\
$\ell/n = 0.40$ & 0.1505 & 0.1949 & 0.1941 & 0.2506 & 0.2142 & 0.2160 & 0.2160 \\
$\ell/n = 0.50$ & 0.1425 & 0.2021 & 0.2011 & 0.2843 & 0.2243 & 0.2265 & 0.2265 \\
\midrule

\multicolumn{8}{l}{\textit{Moderate correlation ($\rho_w=0$, $\rho_x=0$, $\rho_u = 0.6$)}} \\
\addlinespace
$\ell/n = 0.01$ & 0.1522 & 0.1533 & 0.1531 & 0.1541 & 0.1540 & 0.1545 & 0.1544 \\
$\ell/n = 0.05$ & 0.1579 & 0.1624 & 0.1622 & 0.1666 & 0.1643 & 0.1653 & 0.1653 \\
$\ell/n = 0.10$ & 0.1556 & 0.1644 & 0.1641 & 0.1732 & 0.1679 & 0.1689 & 0.1689 \\
$\ell/n = 0.20$ & 0.1505 & 0.1687 & 0.1683 & 0.1884 & 0.1752 & 0.1766 & 0.1766 \\
$\ell/n = 0.30$ & 0.1447 & 0.1734 & 0.1729 & 0.2068 & 0.1828 & 0.1841 & 0.1841 \\
$\ell/n = 0.40$ & 0.1397 & 0.1809 & 0.1802 & 0.2328 & 0.1928 & 0.1942 & 0.1942 \\
$\ell/n = 0.50$ & 0.1347 & 0.1910 & 0.1901 & 0.2689 & 0.2043 & 0.2059 & 0.2059 \\
\bottomrule
\end{tabularx}
\end{threeparttable}
\end{table}

\begin{table}[htbp] 
\centering
\caption{Design 1 (HD) with heterogeneous cluster sizes ($n=1000$, $G=200$, $\lambda=1$). \\
Controls $W$ follow the normal distribution and errors are heteroskedastic ($\vartheta=1$). \\}
\label{tab:HD-cluster-2}
\begin{threeparttable}
\setlength{\tabcolsep}{1pt}
\begin{tabularx}{\textwidth}{l@{\hspace{15pt}}*{7}{>{\centering\arraybackslash}X}}
\toprule
 & CR0 & CR1 & CR2 & CR3 & HCK & Riesz-P & Riesz-F \\
\midrule 

\multicolumn{8}{c}{\textit{Size (5\%)}} \\
\midrule

\multicolumn{8}{l}{\textit{Weak correlation ($\rho_w = 0.3$, $\rho_x = 0.3$, $\rho_u = 0.3$)}} \\
\addlinespace
$\ell/n = 0.01$ & 0.0746 & 0.0726 & 0.0684 & 0.0630 & 0.1140 & 0.0600 & 0.0598 \\
$\ell/n = 0.05$ & 0.0738 & 0.0662 & 0.0618 & 0.0494 & 0.0802 & 0.0522 & 0.0522 \\
$\ell/n = 0.10$ & 0.0930 & 0.0776 & 0.0746 & 0.0538 & 0.0790 & 0.0636 & 0.0636 \\
$\ell/n = 0.20$ & 0.1252 & 0.0876 & 0.0810 & 0.0430 & 0.0750 & 0.0594 & 0.0592 \\
$\ell/n = 0.30$ & 0.1560 & 0.0840 & 0.0790 & 0.0270 & 0.0596 & 0.0554 & 0.0546 \\
$\ell/n = 0.40$ & 0.1822 & 0.0904 & 0.0890 & 0.0256 & 0.0652 & 0.0626 & 0.0624 \\
$\ell/n = 0.50$ & 0.2322 & 0.0936 & 0.0946 & 0.0164 & 0.0660 & 0.0672 & 0.0660 \\
\midrule

\multicolumn{8}{l}{\textit{Moderate correlation ($\rho_w = 0.6$, $\rho_x = 0.6$, $\rho_u = 0.6$)}} \\
\addlinespace
$\ell/n = 0.01$ & 0.0846 & 0.0816 & 0.0680 & 0.0526 & 0.2502 & 0.0520 & 0.0520 \\
$\ell/n = 0.05$ & 0.0870 & 0.0778 & 0.0644 & 0.0424 & 0.1734 & 0.0530 & 0.0528 \\
$\ell/n = 0.10$ & 0.1166 & 0.0960 & 0.0752 & 0.0378 & 0.1532 & 0.0604 & 0.0606 \\
$\ell/n = 0.20$ & 0.1730 & 0.1278 & 0.0894 & 0.0268 & 0.1322 & 0.0626 & 0.0620 \\
$\ell/n = 0.30$ & 0.1966 & 0.1172 & 0.0938 & 0.0186 & 0.1010 & 0.0558 & 0.0554 \\
$\ell/n = 0.40$ & 0.2074 & 0.1020 & 0.0944 & 0.0150 & 0.0816 & 0.0588 & 0.0598 \\
$\ell/n = 0.50$ & 0.2358 & 0.0930 & 0.0968 & 0.0096 & 0.0756 & 0.0662 & 0.0684 \\
\midrule

\multicolumn{8}{c}{\textit{Mean CI Length (95\%)}} \\
\midrule

\multicolumn{8}{l}{\textit{Weak correlation ($\rho_w = 0.3$, $\rho_x = 0.3$, $\rho_u = 0.3$)}} \\
\addlinespace
$\ell/n = 0.01$ & 0.1867 & 0.1881 & 0.1910 & 0.1964 & 0.1689 & 0.1999 & 0.1999 \\
$\ell/n = 0.05$ & 0.1796 & 0.1848 & 0.1875 & 0.1979 & 0.1772 & 0.1963 & 0.1963 \\
$\ell/n = 0.10$ & 0.1741 & 0.1840 & 0.1871 & 0.2031 & 0.1823 & 0.1973 & 0.1973 \\
$\ell/n = 0.20$ & 0.1639 & 0.1838 & 0.1871 & 0.2165 & 0.1906 & 0.2018 & 0.2018 \\
$\ell/n = 0.30$ & 0.1559 & 0.1868 & 0.1895 & 0.2352 & 0.2001 & 0.2088 & 0.2092 \\
$\ell/n = 0.40$ & 0.1486 & 0.1924 & 0.1933 & 0.2590 & 0.2105 & 0.2170 & 0.2172 \\
$\ell/n = 0.50$ & 0.1430 & 0.2029 & 0.2014 & 0.2947 & 0.2239 & 0.2295 & 0.2296 \\
\midrule

\multicolumn{8}{l}{\textit{Moderate correlation ($\rho_w = 0.6$, $\rho_x = 0.6$, $\rho_u = 0.6$)}} \\
\addlinespace
$\ell/n = 0.01$ & 0.2123 & 0.2139 & 0.2228 & 0.2372 & 0.1434 & 0.2411 & 0.2411 \\
$\ell/n = 0.05$ & 0.1891 & 0.1945 & 0.2054 & 0.2282 & 0.1505 & 0.2181 & 0.2181 \\
$\ell/n = 0.10$ & 0.1745 & 0.1844 & 0.2006 & 0.2362 & 0.1569 & 0.2161 & 0.2161 \\
$\ell/n = 0.20$ & 0.1553 & 0.1741 & 0.1926 & 0.2519 & 0.1700 & 0.2149 & 0.2156 \\
$\ell/n = 0.30$ & 0.1476 & 0.1769 & 0.1892 & 0.2670 & 0.1835 & 0.2174 & 0.2178 \\
$\ell/n = 0.40$ & 0.1443 & 0.1868 & 0.1907 & 0.2850 & 0.1978 & 0.2211 & 0.2211 \\
$\ell/n = 0.50$ & 0.1431 & 0.2030 & 0.1991 & 0.3160 & 0.2143 & 0.2312 & 0.2310 \\
\bottomrule
\end{tabularx}
\end{threeparttable}
\end{table}

\begin{table}[htbp] 
\centering
\caption{Design 2 (LF) with independent observations ($n=1000$, $G=1000$). \\
Controls $W$ follow the uniform distribution and errors are heteroskedastic ($\vartheta=1$). \\}
\label{tab:LF-IID}
\begin{threeparttable}
\setlength{\tabcolsep}{1pt}
\begin{tabularx}{\textwidth}{l@{\hspace{15pt}}*{7}{>{\centering\arraybackslash}X}}
\toprule
 & CR0 & CR1 & CR2 & CR3 & HCK & Riesz-P & Riesz-F \\
\midrule
\multicolumn{8}{c}{\textit{Size (5\%)}} \\
\midrule
\multicolumn{8}{l}{\textit{Panel BI: $\ell/n = 0.01$}} \\
\addlinespace
$\delta_x=\delta_w=0.00$ & 0.0580 & 0.0560 & 0.0560 & 0.0550 & 0.0556 & 0.0534 & 0.0534 \\ 
$\delta_x=\delta_w=0.20$ & 0.0614 & 0.0612 & 0.0602 & 0.0574 & 0.0588 & 0.0566 & 0.0568 \\
$\delta_x=\delta_w=0.40$ & 0.0672 & 0.0654 & 0.0634 & 0.0594 & 0.0614 & 0.0588 & 0.0588 \\
$\delta_x=\delta_w=0.60$ & 0.0670 & 0.0656 & 0.0636 & 0.0596 & 0.0620 & 0.0600 & 0.0600 \\
$\delta_x=\delta_w=0.80$ & 0.0664 & 0.0650 & 0.0610 & 0.0554 & 0.0576 & 0.0558 & 0.0558 \\
\midrule

\multicolumn{8}{l}{\textit{Panel BII: $\ell/n = 0.05$}} \\
\addlinespace
$\delta_x=\delta_w=0.00$ & 0.0654 & 0.0578 & 0.0578 & 0.0496 & 0.0534 & 0.0508 & 0.0510 \\ 
$\delta_x=\delta_w=0.20$ & 0.0752 & 0.0666 & 0.0624 & 0.0504 & 0.0546 & 0.0520 & 0.0520 \\
$\delta_x=\delta_w=0.40$ & 0.0798 & 0.0720 & 0.0636 & 0.0476 & 0.0524 & 0.0504 & 0.0504 \\
$\delta_x=\delta_w=0.60$ & 0.0836 & 0.0754 & 0.0600 & 0.0424 & 0.0510 & 0.0470 & 0.0470 \\
$\delta_x=\delta_w=0.80$ & 0.0874 & 0.0800 & 0.0628 & 0.0390 & 0.0520 & 0.0494 & 0.0494 \\
\midrule

\multicolumn{8}{l}{\textit{Panel BIII: $\ell/n = 0.10$}} \\
\addlinespace
$\delta_x=\delta_w=0.00$ & 0.0754 & 0.0616 & 0.0616 & 0.0492 & 0.0540 & 0.0518 & 0.0518 \\ 
$\delta_x=\delta_w=0.20$ & 0.0952 & 0.0828 & 0.0756 & 0.0506 & 0.0638 & 0.0604 & 0.0604 \\
$\delta_x=\delta_w=0.40$ & 0.1136 & 0.0956 & 0.0772 & 0.0458 & 0.0588 & 0.0556 & 0.0556\\
$\delta_x=\delta_w=0.60$ & 0.1220 & 0.1010 & 0.0750 & 0.0400 & 0.0564 & 0.0530 & 0.0530 \\
$\delta_x=\delta_w=0.80$ & 0.1216 & 0.1008 & 0.0666 & 0.0348 & 0.0504 & 0.0490 & 0.0490 \\
\midrule

\multicolumn{8}{c}{\textit{Mean CI Length (95\%)}} \\
\midrule

\multicolumn{8}{l}{\textit{Panel DI: $\ell/n = 0.01$}} \\
\addlinespace
$\delta_x=\delta_w=0.00$ & 0.1865 & 0.1876 & 0.1876 & 0.1886 & 0.1880 & 0.1893 & 0.1892 \\ 
$\delta_x=\delta_w=0.20$ & 0.1940 & 0.1951 & 0.1958 & 0.1975 & 0.1967 & 0.1985 & 0.1985 \\
$\delta_x=\delta_w=0.40$ & 0.2044 & 0.2056 & 0.2071 & 0.2098 & 0.2086 & 0.2110 & 0.2110 \\
$\delta_x=\delta_w=0.60$ & 0.2084 & 0.2096 & 0.2118 & 0.2153 & 0.2135 & 0.2159 & 0.2159 \\
$\delta_x=\delta_w=0.80$ & 0.2047 & 0.2058 & 0.2085 & 0.2124 & 0.2102 & 0.2120 & 0.2120 \\
\midrule

\multicolumn{8}{l}{\textit{Panel DII: $\ell/n = 0.05$}} \\
\addlinespace
$\delta_x=\delta_w=0.00$ & 0.1978 & 0.2031 & 0.2031 & 0.2085 & 0.2063 & 0.2081 & 0.2081 \\ 
$\delta_x=\delta_w=0.20$ & 0.2035 & 0.2089 & 0.2125 & 0.2221 & 0.2184 & 0.2212 & 0.2212 \\
$\delta_x=\delta_w=0.40$ & 0.2090 & 0.2145 & 0.2224 & 0.2374 & 0.2313 & 0.2348 & 0.2348 \\
$\delta_x=\delta_w=0.60$ & 0.2072 & 0.2127 & 0.2235 & 0.2421 & 0.2334 & 0.2366 & 0.2366 \\
$\delta_x=\delta_w=0.80$ & 0.1983 & 0.2035 & 0.2159 & 0.2363 & 0.2248 & 0.2272 & 0.2272 \\
\midrule

\multicolumn{8}{l}{\textit{Panel DIII: $\ell/n = 0.10$}} \\
\addlinespace
$\delta_x=\delta_w=0.00$ & 0.1944 & 0.2050 & 0.2050 & 0.2162 & 0.2115 & 0.2135 & 0.2135 \\ 
$\delta_x=\delta_w=0.20$ & 0.1941 & 0.2047 & 0.2103 & 0.2288 & 0.2214 & 0.2242 & 0.2242 \\
$\delta_x=\delta_w=0.40$ & 0.1939 & 0.2045 & 0.2166 & 0.2442 & 0.2323 & 0.2358 & 0.2358 \\
$\delta_x=\delta_w=0.60$ & 0.1898 & 0.2001 & 0.2169 & 0.2510 & 0.2339 & 0.2371 & 0.2371 \\
$\delta_x=\delta_w=0.80$ & 0.1814 & 0.1914 & 0.2108 & 0.2485 & 0.2259 & 0.2283 & 0.2283 \\
\bottomrule

\end{tabularx}
\end{threeparttable}
\end{table}

\begin{table}[htbp] 
\centering
\caption{Design 2 (LF) with homogeneous cluster sizes ($n=1000$, $G=200$, $\lambda=0$). \\
Controls $W$ follow the uniform distribution and errors are heteroskedastic ($\vartheta=1$). \\
Within-cluster correlation is weak: $\rho_w=0$, $\rho_x=0$, $\rho_u = 0.3$. \\}
\label{tab:LF-cluster}
\begin{threeparttable}
\setlength{\tabcolsep}{1pt}
\begin{tabularx}{\textwidth}{l@{\hspace{15pt}}*{7}{>{\centering\arraybackslash}X}}
\toprule
 & CR0 & CR1 & CR2 & CR3 & HCK & Riesz-P & Riesz-F \\
\midrule

\multicolumn{8}{c}{\textit{Size (5\%)}} \\
\midrule

\multicolumn{8}{l}{\textit{Panel BI: $\ell/n = 0.01$}} \\
\addlinespace
$\delta_x=\delta_w=0.00$ & 0.0588 & 0.0574 & 0.0578 & 0.0564 & 0.0556 & 0.0560 & 0.0560 \\ 
$\delta_x=\delta_w=0.20$ & 0.0618 & 0.0612 & 0.0610 & 0.0574 & 0.0542 & 0.0562 & 0.0562 \\
$\delta_x=\delta_w=0.40$ & 0.0590 & 0.0572 & 0.0552 & 0.0508 & 0.0522 & 0.0510 & 0.0510 \\
$\delta_x=\delta_w=0.60$ & 0.0628 & 0.0616 & 0.0568 & 0.0520 & 0.0532 & 0.0536 & 0.0536 \\
$\delta_x=\delta_w=0.80$ & 0.0608 & 0.0590 & 0.0546 & 0.0492 & 0.0518 & 0.0520 & 0.0520 \\
\midrule

\multicolumn{8}{l}{\textit{Panel BII: $\ell/n = 0.05$}} \\
\addlinespace
$\delta_x=\delta_w=0.00$ & 0.0670 & 0.0590 & 0.0596 & 0.0534 & 0.0572 & 0.0538 & 0.0538 \\ 
$\delta_x=\delta_w=0.20$ & 0.0718 & 0.0646 & 0.0616 & 0.0508 & 0.0546 & 0.0532 & 0.0532 \\
$\delta_x=\delta_w=0.40$ & 0.0836 & 0.0760 & 0.0662 & 0.0486 & 0.0564 & 0.0572 & 0.0572 \\
$\delta_x=\delta_w=0.60$ & 0.0892 & 0.0822 & 0.0650 & 0.0420 & 0.0514 & 0.0532 & 0.0532 \\
$\delta_x=\delta_w=0.80$ & 0.0938 & 0.0850 & 0.0596 & 0.0336 & 0.0500 & 0.0530 & 0.0528 \\
\midrule

\multicolumn{8}{l}{\textit{Panel BIII: $\ell/n = 0.10$}} \\
\addlinespace
$\delta_x=\delta_w=0.00$ & 0.0818 & 0.0666 & 0.0672 & 0.0554 & 0.0590 & 0.0584 & 0.0584 \\ 
$\delta_x=\delta_w=0.20$ & 0.0948 & 0.0768 & 0.0716 & 0.0488 & 0.0582 & 0.0556 & 0.0556 \\
$\delta_x=\delta_w=0.40$ & 0.1002 & 0.0840 & 0.0718 & 0.0410 & 0.0554 & 0.0524 & 0.0524 \\
$\delta_x=\delta_w=0.60$ & 0.1082 & 0.0864 & 0.0660 & 0.0318 & 0.0530 & 0.0480 & 0.0480 \\
$\delta_x=\delta_w=0.80$ & 0.1204 & 0.0960 & 0.0626 & 0.0194 & 0.0514 & 0.0526 & 0.0526 \\
\midrule

\multicolumn{8}{c}{\textit{Mean CI Length (95\%)}} \\
\midrule

\multicolumn{8}{l}{\textit{Panel DI: $\ell/n = 0.01$}} \\
\addlinespace
$\delta_x=\delta_w=0.00$ & 0.1703 & 0.1715 & 0.1712 & 0.1722 & 0.1723 & 0.1730 & 0.1730 \\ 
$\delta_x=\delta_w=0.20$ & 0.1761 & 0.1774 & 0.1778 & 0.1796 & 0.1795 & 0.1806 & 0.1805 \\
$\delta_x=\delta_w=0.40$ & 0.1841 & 0.1855 & 0.1871 & 0.1901 & 0.1895 & 0.1909 & 0.1909 \\
$\delta_x=\delta_w=0.60$ & 0.1877 & 0.1891 & 0.1918 & 0.1962 & 0.1945 & 0.1959 & 0.1960 \\
$\delta_x=\delta_w=0.80$ & 0.1853 & 0.1867 & 0.1907 & 0.1963 & 0.1930 & 0.1942 & 0.1942 \\
\midrule

\multicolumn{8}{l}{\textit{Panel DII: $\ell/n = 0.05$}} \\
\addlinespace
$\delta_x=\delta_w=0.00$ & 0.1791 & 0.1843 & 0.1839 & 0.1889 & 0.1869 & 0.1883 & 0.1883 \\ 
$\delta_x=\delta_w=0.20$ & 0.1853 & 0.1906 & 0.1933 & 0.2023 & 0.1989 & 0.2009 & 0.2009 \\
$\delta_x=\delta_w=0.40$ & 0.1933 & 0.1989 & 0.2060 & 0.2214 & 0.2144 & 0.2169 & 0.2169 \\
$\delta_x=\delta_w=0.60$ & 0.1958 & 0.2014 & 0.2135 & 0.2361 & 0.2229 & 0.2252 & 0.2252 \\
$\delta_x=\delta_w=0.80$ & 0.1901 & 0.1956 & 0.2143 & 0.2458 & 0.2228 & 0.2242 & 0.2242 \\
\midrule

\multicolumn{8}{l}{\textit{Panel DIII: $\ell/n = 0.10$}} \\
\addlinespace
$\delta_x=\delta_w=0.00$ & 0.1762 & 0.1862 & 0.1859 & 0.1961 & 0.1915 & 0.1930 & 0.1930 \\ 
$\delta_x=\delta_w=0.20$ & 0.1817 & 0.1920 & 0.1955 & 0.2122 & 0.2048 & 0.2068 & 0.2068 \\
$\delta_x=\delta_w=0.40$ & 0.1926 & 0.2035 & 0.2117 & 0.2386 & 0.2245 & 0.2272 & 0.2272 \\
$\delta_x=\delta_w=0.60$ & 0.2012 & 0.2126 & 0.2269 & 0.2678 & 0.2412 & 0.2440 & 0.2440 \\
$\delta_x=\delta_w=0.80$ & 0.2018 & 0.2133 & 0.2399 & 0.3055 & 0.2533 & 0.2554 & 0.2554 \\
\bottomrule

\end{tabularx}
\end{threeparttable}
\end{table}

\newpage
\section{Empirical Application} \label{sec:application}

We revisit \citet{xu2018costs}, which studies the role of social connections in shaping the promotion of colonial governors in the British Empire. We focus on the salary regressions in Table 2 of the original paper, which provide a canonical setting combining high-dimensional fixed effects with dyadic clustering.

Specifically, we estimate the model
    \[ \log(w_{ist}) = \alpha + \beta\, c_{it} + \gamma\, \mathit{served}_{it} + \theta_i + \tau_t + \delta_{it} + \varepsilon_{ist}, \]
where $w_{ist}$ is the salary of governor $i$ in colony $s$ at time $t$, and $c_{it}$ indicates social connection to the Secretary of State, measured along four dimensions: (1) shared ancestry, (2) aristocratic membership, (3) attendance at Eton, and (4) attendance at Oxford or Cambridge. Regression (6) uses a composite indicator. The parameter of interest is $\beta$. Controls include the number of previously served colonies, governor fixed effects $\theta_i$, year fixed effects $\tau_t$, and tenure fixed effects $\delta_{it}$. Standard errors are clustered at the governor--Secretary dyad level.

The estimation sample contains $3{,}510$ observations grouped into $1{,}518$ dyadic clusters, with $572$ controls ($\ell/n = 0.163$). This places the design firmly in the many-controls regime. Some dyadic clusters are perfectly nested within the high-dimensional fixed effects, making $(I_{N_g} - H_{gg})$ singular for those clusters and rendering CR2 and CR3 undefined throughout. The HCK estimator, while designed for many-controls settings, ignores within-cluster dependence and produces standard errors even smaller than CR1 --- reflecting its reliance on the diagonal projection $M_W$ alone rather than the full residual-maker $M_{[X,W]}$.

Table~\ref{tab:empirical_1} reports the replication results. Point estimates are identical across all methods, as expected. The main finding is that the Riesz estimator produces systematically larger standard errors than CR1 and HCK, with material consequences for inference. Under CR1, shared ancestry is significant at the 5\% level; under the Riesz estimator it is only marginally significant at the 10\% level ($p = 0.070$). For aristocratic membership, Eton attendance, and Oxbridge attendance, significance disappears entirely under the Riesz estimator. The composite connection measure remains significant at the 5\% level but with a larger $p$-value ($p = 0.019$ versus $p = 0.006$ under CR1). Riesz-P and Riesz-F produce identical results, consistent with Lemmas 5--6 given the continuous variation in connection measures across the 1,518 dyads.

These results reflect two compounding sources of projection spillovers in this design: the many-controls channel ($\ell/n = 0.163$) and the dyadic clustering channel, where governor and Secretary fixed effects cross the dyad boundaries and create non-negligible off-diagonal mass in $M$. Classical estimators that rely on local projection information omit the variance components flowing through these cross-cluster residual interactions, leading to understated standard errors and over-rejection. The Riesz estimator recovers these components by solving the full residual moment system, providing more reliable inference in this setting.

\begin{table}[htbp]
\centering
\caption{Governor Salary and Connectedness of Secretary of State}
\label{tab:empirical_1}
\begin{threeparttable}
\setlength{\tabcolsep}{1pt}
\begin{tabularx}{\textwidth}{l@{\hspace{15pt}}*{7}{>{\centering\arraybackslash}X}}
\toprule
Regression & CR0 & CR1 & CR2 & CR3 & HCK & Riesz-P & Riesz-F \\
\midrule
\multicolumn{8}{l}{\textit{Regression (1): Shared ancestors}} \\
\addlinespace
Estimate    & 0.1035 & 0.1035 & 0.1035 & 0.1035 & 0.1035 & 0.1035 & 0.1035 \\
Std.Error   & 0.0430 & 0.0470 & NA     & NA     & 0.0308 & 0.0570 & 0.0570 \\
t-statistic & 2.4047 & 2.1993 & NA     & NA     & 3.3557 & 1.8150 & 1.8150 \\
P-value     & 0.0163 & 0.0280 & NA     & NA     & 0.0008 & 0.0697 & 0.0697 \\
\midrule
\multicolumn{8}{l}{\textit{Regression (2): Both aristocrats}} \\
\addlinespace
Estimate    & 0.2150 & 0.2150 & 0.2150 & 0.2150 & 0.2150 & 0.2150 & 0.2150 \\
Std.Error   & 0.1132 & 0.1238 & NA     & NA     & 0.0642 & 0.1838 & 0.1837 \\
t-statistic & 1.8977 & 1.7356 & NA     & NA     & 3.3471 & 1.1693 & 1.1696 \\
P-value     & 0.0579 & 0.0828 & NA     & NA     & 0.0008 & 0.2425 & 0.2424 \\
\midrule
\multicolumn{8}{l}{\textit{Regression (3): Both Eton}} \\
\addlinespace
Estimate    & 0.1328 & 0.1328 & 0.1328 & 0.1328 & 0.1328 & 0.1328 & 0.1328 \\
Std.Error   & 0.0706 & 0.0772 & NA     & NA     & 0.0427 & 0.1002 & 0.1002 \\
t-statistic & 1.8806 & 1.7199 & NA     & NA     & 3.1086 & 1.3256 & 1.3257 \\
P-value     & 0.0602 & 0.0856 & NA     & NA     & 0.0019 & 0.1852 & 0.1851 \\
\midrule
\multicolumn{8}{l}{\textit{Regression (4): Both Oxbridge}} \\
\addlinespace
Estimate    & 0.0717 & 0.0717 & 0.0717 & 0.0717 & 0.0717 & 0.0717 & 0.0717 \\
Std.Error   & 0.0428 & 0.0468 & NA     & NA     & 0.0278 & 0.0562 & 0.0562 \\
t-statistic & 1.6747 & 1.5316 & NA     & NA     & 2.5743 & 1.2762 & 1.2762 \\
P-value     & 0.0942 & 0.1258 & NA     & NA     & 0.0101 & 0.2021 & 0.2021 \\
\midrule
\multicolumn{8}{l}{\textit{Regression (6): Connected}} \\
\addlinespace
Estimate    & 0.0973 & 0.0973 & 0.0973 & 0.0973 & 0.0973 & 0.0973 & 0.0973 \\
Std.Error   & 0.0325 & 0.0356 & NA     & NA     & 0.0229 & 0.0414 & 0.0414 \\
t-statistic & 2.9925 & 2.7368 & NA     & NA     & 4.2396 & 2.3489 & 2.3490 \\
P-value     & 0.0028 & 0.0062 & NA     & NA     & 0.0000 & 0.0190 & 0.0189 \\
\bottomrule
\end{tabularx}
\end{threeparttable}
\end{table}

\newpage
\section{Conclusion} \label{sec:conclusion}

This paper develops a unified projection-geometry framework for variance estimation in linear regression. The central insight is that the variance of the OLS estimator admits an exact finite-sample representation as a Riesz functional of latent covariance blocks, with observable residual moments linked to this target through a linear operator determined by projection geometry. This formulation unifies classical HC, CRVE, and many-controls estimators as restricted solutions within a common operator system, with accuracy governed by projection sparsity. The framework applies regardless of whether the projection matrix is sparse or dense --- classical methods emerge as the limiting case when projection is local, while the Riesz estimator provides valid inference throughout.

The proposed estimator uses both within- and cross-cluster residual moments with the full projection geometry, is implemented via a scalable matrix-free algorithm based on Golub--Kahan bidiagonalization and LSQR, and remains valid when classical bias-corrected estimators break down due to singular cluster-specific leverage matrices. Simulations confirm that classical methods severely undercover under projection spillovers while the Riesz estimator restores nominal coverage. Revisiting \citet{xu2018costs}, conventional standard errors understate uncertainty substantially, reversing significance on four of five reported coefficients. The Riesz representation provides a flexible foundation for variance estimation in a broad class of regression designs beyond those considered here.

\newpage
\pagenumbering{roman} 


\nocite{*}
\bibliographystyle{plainnat}
\bibliography{reference}

\newpage 
\section*{Appendix} 
\appendix 
\pagenumbering{Roman} 
\fancyhf{} 
\renewcommand{\headrulewidth}{0pt}

\newtheorem{applemma}{lemma}[section]
\newtheorem{prelemma}{Lemma}[section]   

\addcontentsline{toc}{section}{Appendix}
\setcounter{equation}{0}
\numberwithin{equation}{section}

\section{Projection Geometry and Auxiliary Results}
\subsection{Projection Spillover} \label{subsec:proj-spillover}
\begin{applemma}[Individual off-diagonal mass]
\label{applem:indiv-offdiag}
Let $M = I - H$ be the OLS residual-maker, symmetric and idempotent. For any observation $i$,
\begin{equation}
    \sum_{j \neq i} M_{ij}^2 = h_{ii}(1 - h_{ii}) = \|P_i\|^2\cdot\|Q_i\|^2, \label{appeq:indiv-offdiag}
\end{equation}
where $P_i = He_i$ and $Q_i = Me_i$ are the $i$-th columns of $H$ and $M$.
\end{applemma}

\begin{proof}
Idempotency of $M$ gives $(M^2)_{ii} = M_{ii}$, i.e.\ $\sum_j M_{ij}^2 = M_{ii}$. Separating the diagonal term:
    \[ \sum_{j \neq i} M_{ij}^2 = M_{ii} - M_{ii}^2 = M_{ii}(1 - M_{ii}) = (1 - h_{ii})\,h_{ii}. \]
The geometric form follows from $\|P_i\|^2 = (H^2)_{ii} = H_{ii} = h_{ii}$ and $\|Q_i\|^2 = (M^2)_{ii} = M_{ii} = 1-h_{ii}$, both by idempotency.
\end{proof}

\begin{remark}
The expression $h_{ii}(1-h_{ii})$ is zero at both extremes $h_{ii} \in \{0,1\}$ and maximized at $h_{ii} = 1/2$. Non-negligible off-diagonal mass at observation $i$ requires both a non-trivial shadow $P_i$ (leverage) and a non-trivial residual direction $Q_i$ (unexplained variation).
\end{remark}

\begin{applemma}[Aggregate off-diagonal mass] \label{applem:agg-offdiag}
Under the same setup,
\begin{equation}
    \sum_{i=1}^{n} \sum_{i \neq j} M_{ij}^2 = p - \sum_i h_{ii}^2 \leq p\!\left(1 - \frac{p}{n} \right), \label{appeq:agg-offdiag}
\end{equation}
with equality under uniform leverage $h_{ii} = p/n$ for all $i$.
\end{applemma}

\begin{proof}
By \ref{applem:indiv-offdiag}, the off-diagonal mass equals $\sum_i (h_{ii} - h_{ii}^2) = \mathrm{tr}(H) - \sum_i h_{ii}^2 = p - \sum_i h_{ii}^2$. Jensen's inequality applied to $f(x) = x^2$ with constraint $\sum_i h_{ii} = p$ gives $\sum_i h_{ii}^2 \geq p^2/n$, so $\sum_i \sum_{i\neq j} M_{ij}^2 \leq p(1 - p/n) \to cn(1-c) = O(n)$.
\end{proof}

\begin{remark}
The bound $p(1-p/n)$ is $O(n)$ when $p/n \to c \in (0,1)$, confirming first-order aggregate off-diagonal mass in the many-regressors regime. Lemma~\ref{applem:indiv-offdiag} and Lemma~\ref{applem:agg-offdiag} characterize the two channels: the aggregate bound \eqref{appeq:agg-offdiag} governs the bulk channel driven by the size of $p$, while the individual formula \eqref{appeq:indiv-offdiag} governs the pairwise channel driven by observations with large $h_{ii}$ concentrated on shared directions.
\end{remark}
\subsection{Fixed-Dimensional Projection} \label{subsec:fixed-dim}
\begin{applemma}[Low-rank perturbation of the residual projection] \label{applem:M-MW}
Suppose $\dim(X)=k$ is fixed Assumption~\ref{asmp:A4} holds. Let $M := M_{\widetilde X} M_W = M_{[X,W]}$. Then
    \[ M = M_W + \Delta_M, \qquad \frac{1}{n}\|\Delta_M\|_F^2 = o_p(1). \]
Consequently,
    \[ \frac{1}{n} \big\| (M\!\odot\! M) - (M_W\!\odot\! M_W) \big\|_F^2 = o_p(1). \]
\end{applemma}

\begin{proof}
Recall that $M := M_{\widetilde X} M_W$ and $M_{\widetilde X} := I_n - H_{\widetilde X}$ with $H_{\widetilde X} = \widetilde X (\widetilde X' \widetilde X)^{-1} \widetilde X'$. Therefore,
    \[ \Delta_M := M - M_W = (M_{\widetilde X} - I_n) M_W = -H_{\widetilde X} M_W . \]

\medskip
We first bound $\Delta_M$ in Frobenius norm. Using $\widetilde X'\widetilde X = n\Gamma_n$ and Assumption~\ref{asmp:A4},
    \[ H_{\widetilde X} = \widetilde X (n\Gamma_n)^{-1} \widetilde X' = \frac{1}{n}\,\widetilde X \Gamma_n^{-1} \widetilde X', \qquad \|\Gamma_n^{-1}\|_{\mathrm{op}}=O_p(1). \]
Since $\dim(X)=k$ is fixed, $\mathrm{rank} (H_{\widetilde X}) = \mathrm{rank} (\widetilde X) \le k$. Because $H_{\widetilde X}$ is a symmetric idempotent projection,
    \[ \|H_{\widetilde X}\|_F^2 = \mathrm{tr} (H_{\widetilde X}) = \mathrm{rank} (H_{\widetilde X}) \le k, \qquad \text{so} \quad \|H_{\widetilde X}\|_F \le \sqrt{k}. \]
Moreover, $M_W$ is an orthogonal projection and hence $\|M_W\|_{\mathrm{op}} \le 1$. Therefore,
    \[ \|\Delta_M\|_F = \|H_{\widetilde X}M_W\|_F \le \|H_{\widetilde X}\|_F \|M_W\|_{ \mathrm{op}} \le \sqrt{k}. \]
Consequently,
    \[ \frac{1}{n}\|\Delta_M\|_F^2 \le \frac{k}{n} = o(1). \]

\medskip
Next, consider the Hadamard squares. Using
    \[ (M\odot M)-(M_W\odot M_W) = (M-M_W)\odot (M+M_W) = \Delta M \odot (M+M_W), \]
and the inequality $\|A\odot B\|_F \le \|A\|_F \|B\|_{\max}$ where $\|B\|_{\max} := \max_{i,j} |B_{ij}|$, we obtain
    \[ \|(M\odot M)-(M_W\odot M_W)\|_F \le \|\Delta M\|_F\,\|M+M_W\|_{\max}. \]
Since $M$ and $M_W$ are orthogonal projections, $|M_{ij}|\le 1$ and $|(M_W)_{ij}|\le 1$ for all $i,j$, so $\|M+M_W\|_{\max}\le 2$. Therefore,
    \[ \|(M\odot M)-(M_W\odot M_W)\|_F \le 2\|\Delta M\|_F, \]
and hence
    \[ \frac1n\|(M\odot M)-(M_W\odot M_W)\|_F^2 \le \frac{4}{n}\|\Delta M\|_F^2 \le \frac{4k}{n} = o(1). \]
\end{proof}

\begin{remark}[Implication for inverse--Hadamard solutions] \label{rem:invhad-proof}
Lemma~\ref{applem:M-MW} provides a perturbation bound for the projection matrix. To translate this result into equivalence of scalar-cluster Riesz \emph{weights},  additional stability of the inverse--Hadamard mapping is required.

Let $A_n := M\odot M$ and $B_n := M_W\odot M_W$, and define $d_A := A_n^{-1} r_n$ and $d_B := B_n^{-1} r_n$. Using the resolvent identity \( A_n^{-1} - B_n^{-1} = A_n^{-1} (B_n-A_n) B_n^{-1} \), we obtain
    \[ \|d_A-d_B\|_2 \;\le\; \|A_n^{-1}\|_{\mathrm{op}} \, \|A_n-B_n\|_{\mathrm{op}} \, \|B_n^{-1}\|_{\mathrm{op}} \, \|r_n\|_2 \;\le\; \|A_n^{-1}\|_{\mathrm{op}} \, \|A_n-B_n\|_{F} \, \|B_n^{-1}\|_{\mathrm{op}} \, \|r_n\|_2. \]
In the scalar-cluster design ($N_g=1$), Lemma~\ref{pre:riesz-bdd} implies $\|r_n\|_2 = O_p(\sqrt n)$, while Lemma~\ref{applem:M-MW} yields $\|A_n-B_n\|_F = o_p(\sqrt n)$. If $\|A_n^{-1}\|_{\mathrm{op}}$ and $\|B_n^{-1}\|_{\mathrm{op}}$ are uniformly bounded, then
    \[ \|d_A-d_B\|_2 = o_p(n). \]
Equivalently, in the Riesz geometry $\|D\|_{\mathcal K_n}=(1/n)\|d\|_2$ for scalar clusters, so
    \[ \|D_A-D_B\|_{\mathcal K_n}=o_p(1). \]

Moreover, by Lemma~\ref{pre:A-operator-bdd} and the fact that orthogonal projection is a contraction,\footnote{ Recall that $\mathcal A_n^{\mathrm{diag}} = P_{\mathrm{diag}}\circ \mathcal A_n$, where $P_{\mathrm{diag}}$ is the orthogonal projection onto the diagonal subspace of $\mathcal K_n$. Since $\|P_{\mathrm{diag}}\|_{\mathrm{op}}=1$ and $\|\mathcal A_n^\ast\|_{\mathrm{op}}=1$, it follows that $\|(\mathcal A_n^{\mathrm{diag}})^\ast\|_{\mathrm{op}}\le 1$.} we have
    \[ \big\|(\mathcal A_n^{\mathrm{diag}})^\ast(D_A-D_B)\big\|_{\mathcal H_n} \;\le\; \|D_A-D_B\|_{\mathcal K_n} \;=\; o_p(1). \]
This verifies the Partial-Approx condition in the scalar-cluster setting (where $\mu_n=n$ and $\sup_g N_g=1$). Hence, under inverse--Hadamard stability, low-rank perturbations of the projection matrix do not affect the variance estimator at first order.
\end{remark}

\medskip
\begin{applemma}[Diagonal Riesz expansion under fixed-dimensional projection] \label{applem:diag-expansion}
Suppose $\dim(X)=k$ is fixed and Assumptions~\ref{asmp:A2} and~\ref{asmp:A4} hold. Then, uniformly in $i$,
    \[ d_i = \frac{r_i}{(1-h_{ii})^2} + O_p\!\left(n^{-1}\,\|r\|_1\right). \]
\end{applemma}

\begin{proof}
In the scalar design, the diagonal partial Riesz weights satisfy
    \[ (M\odot M)d = r, \qquad\text{equivalently}\qquad \sum_{j=1}^n M_{ij}^2 d_j = r_i, \quad i=1,\ldots,n. \]
Fix $i$ and separate the diagonal term:
    \[ M_{ii}^2 d_i + \sum_{j\neq i} M_{ij}^2 d_j = r_i, \]
which implies
    \[ d_i = \frac{r_i}{M_{ii}^2} - \frac{1}{M_{ii}^2} \sum_{j\neq i} M_{ij}^2 d_j. \]
Since $M_{ii}=1-h_{ii}$, the leading term equals
    \[ \frac{r_i}{M_{ii}^2} = \frac{r_i}{(1-h_{ii})^2}. \]

\medskip
To bound the remainder, apply the triangle inequality:
    \[ \Big| \frac{1}{M_{ii}^2} \sum_{j\neq i} M_{ij}^2 d_j \Big| \le \frac{1}{M_{ii}^2} \Big( \sum_{j\neq i} M_{ij}^2 \Big)\, \|d\|_1. \]
Under fixed-dimensional and well-conditioned designs, standard projection geometry implies $\max_i h_{ii}=O_p(n^{-1})$. Since $M$ is a symmetric idempotent projection,
   \[ \sum_{j\neq i} M_{ij}^2 = M_{ii} - M_{ii}^2 = h_{ii}(1-h_{ii}) = O_p(n^{-1}) \quad\text{uniformly in } i. \]
Moreover, since $h_{ii}=o_p(1)$, there exists $c>0$ such that $\min_i M_{ii}^2 \ge c^2$ with probability approaching one. Thus the Hadamard square $M\odot M$ is asymptotically diagonally dominant by columns. By Varah’s bound for strictly diagonally dominant matrices,
    \[ \|(M\odot M)^{-1}\|_1 = O_p(1). \]

Since $d=(M\odot M)^{-1}r$, submultiplicativity of the induced $\ell_1$ norm yields
    \[ \|d\|_1 \le \|(M\odot M)^{-1}\|_1\,\|r\|_1 = O_p(\|r\|_1). \]
Combining the above bounds,
    \[ \Big| \frac{1}{M_{ii}^2} \sum_{j\neq i} M_{ij}^2 d_j \Big| = O_p\!\left(n^{-1}\|r\|_1\right), \]
uniformly in $i$. Substituting back completes the proof.
\end{proof}

\medskip
\begin{applemma}[Block-diagonal approximation for the inverse Khatri--Rao system] \label{applem:block-KR-approx}
Consider the diagonal-moment partial Riesz system \eqref{eq:Khatri-Rao}, stacked as $r=Bd$ with $B_{gh} := M_{gh}\otimes M_{gh}$. Let $B_0$ denote the block-diagonal part of $B$ and $E:=B-B_0$. Suppose:
\begin{enumerate}[(i)]
    \item \textbf{Uniform stability of the diagonal blocks:} $B_0$ is invertible and $\|B_0^{-1}\|_{\mathrm{op}} = O_p(1)$;
    \item \textbf{Two-part leakage control:} 
        \[ \|B_0^{-1}E\|_{\mathrm{op}} =o_p(1) \quad \text{and} \quad o_p\!\Big(\mu_n^{-1}\frac{n}{\sup_g N_g}\Big). \]
\end{enumerate}
Then
    \[ \|d-B_0^{-1}r\|_2 = o_p\!\Big(\mu_n^{-1}\frac{n}{\sup_g N_g}\Big)\,\|r\|_2 . \]
\end{applemma}

\begin{proof}
Write $B=B_0+E$. Since $B_0$ is invertible, we may factor
    \[ B = B_0(I+B_0^{-1}E), \qquad\text{so that} \qquad B^{-1}=(I+B_0^{-1}E)^{-1}B_0^{-1}, \]
whenever $I+B_0^{-1}E$ is invertible. 

\medskip
By assumption (ii), $\|B_0^{-1}E\|_{\mathrm{op}} = o_p(1)$, hence $\|B_0^{-1}E\|_{\mathrm{op}}<1$ with probability approaching one. On this event, $I+B_0^{-1}E$ is invertible and the Neumann series yields
    \[ (I + B_0^{-1} E)^{-1} = \sum_{k=0}^{\infty} (-B_0^{-1} E)^k, \]
so that
    \[ B^{-1} - B_0^{-1} = \Big( (I+B_0^{-1} E)^{-1} - I \Big) B_0^{-1} = -\, (I + B_0^{-1} E)^{-1} (B_0^{-1} E) B_0^{-1}. \]
Taking operator norms and using $\|(I+K)^{-1}\|_{\mathrm{op}} \le (1-\|K\|_{\mathrm{op}})^{-1}$ when $\|K\|_{\mathrm{op}}<1$, we obtain on the event $\{\|B_0^{-1}E\|_{\mathrm{op}}<1\}$,
    \[ \|B^{-1}-B_0^{-1}\|_{\mathrm{op}} \le \frac{\|B_0^{-1}E\|_{\mathrm{op}}}{1-\|B_0^{-1}E\|_{\mathrm{op}}}\, \|B_0^{-1}\|_{\mathrm{op}}. \]
By (i), $\|B_0^{-1}\|_{\mathrm{op}}=O_p(1)$, and by (ii), $\|B_0^{-1}E\|_{\mathrm{op}}=o_p(1)$, so the prefactor $\big(1-\|B_0^{-1}E\|_{\mathrm{op}}\big)^{-1}=O_p(1)$. Hence
    \[ \| B^{-1} - B_0^{-1} \|_{\mathrm{op}} = O_p(1) \, \| B_0^{-1} E \|_{\mathrm{op}}. \]
Multiplying by $\|r\|_2$ gives
    \[ \|d - B_0^{-1} r\|_2 \le \|B^{-1} - B_0^{-1} \|_{\mathrm{op}} \, \|r\|_2 = O_p(1) \, \| B_0^{-1} E \|_{\mathrm{op}} \, \|r\|_2. \]
Finally, using the second part of (ii), $\| B_0^{-1} E \|_{\mathrm{op}} = o_p \! \big( \mu_n^{-1} n / \sup_g N_g \big)$, we conclude
    \[ \| d - B_0^{-1} r \|_2 = o_p \! \Big( \mu_n^{-1} \frac{n}{\sup_g N_g} \Big) \, \|r\|_2, \]
as claimed.
\end{proof}

\newpage
\section{Preliminary Lemmas}
\begin{prelemma} \label{pre:bound}
Suppose Assumption~\ref{asmp:A2} holds. Then, for any $1 \le \theta \le 2+\lambda$,
\begin{enumerate}[(i)]
    \item \(\displaystyle \sup_{g} \; N_g^{-\theta} \, \mathbb{E} \big\| \widetilde{X}_g' \widetilde{X}_g \big\|_{F}^{\theta} = O(1)\) and \(\displaystyle \sup_{g} \; N_g^{-\theta} \, \big\| \widetilde{X}_g' \widetilde{X}_g \big\|_{F}^{\theta} = O_p(1)\);
    \item \(\displaystyle \sup_{g}\; N_g^{-\theta}\, \mathbb{E}\!\big[ \|U_g U_g'\|_F^{\theta} \mid X,W \big] = O(1)\);
    \item \(\displaystyle \sup_{g}\; N_g^{-\theta}\, \mathbb{E}\!\big[ \|\widetilde X'_g U_g\|^{\theta} \mid X,W \big] = O_p(1).\)
\end{enumerate}
\end{prelemma}

\begin{proof} 
\textbf{(i)} Since $\|vv'\|_F = \|v\|^2$ for any vector $v$ and the Frobenius norm satisfies the triangle inequality,
    \[ \big\|\widetilde{X}_g' \widetilde{X}_g\big\|_F = \Big\|\sum_{i=1}^{N_g} \tilde x_{gi}\tilde x_{gi}'\Big\|_F \le \sum_{i=1}^{N_g} \|\tilde x_{gi}\tilde x_{gi}'\|_F = \sum_{i=1}^{N_g} \|\tilde x_{gi}\|^2. \]
For any $\theta \ge 1$ and scalars $a_i$, \( \big( \sum_{i=1}^{N_g} a_i \big)^\theta 
\le N_g^{\theta-1}\sum_{i=1}^{N_g}|a_i|^\theta.\) Applying this with $a_i = \|\tilde x_{gi}\|^2$ gives
    \begin{equation} \label{eq:L1.1}
        N_g^{-\theta}\big\|\widetilde{X}_g' \widetilde{X}_g\big\|_F^\theta \;\le\; N_g^{-\theta} \bigg(\sum_{i=1}^{N_g} \|\tilde  x_{gi}\|^2 \bigg)^{\theta} \;\le\; \frac{1}{N_g}\sum_{i=1}^{N_g}\|\tilde x_{gi}\|^{2\theta}. \tag{L1.1}
    \end{equation}
For $1 \le \theta \le 2+\lambda$ we have $2\theta \le 4+2\lambda$. By Lyapunov’s inequality, for each $(g,i)$,
    \[ \mathbb{E}\|\tilde x_{gi}\|^{2\theta} \;\le\; \big(\mathbb{E}\|\tilde x_{gi}\|^{4+2\lambda}\big)^{\frac{2\theta}{4+2\lambda}} = \big(\mathbb{E}\|\tilde x_{gi}\|^{4+2\lambda}\big)^{\frac{\theta}{2+\lambda}}. \]
Since $0 < \frac{\theta}{2+\lambda} \le 1$, the map $t \mapsto t^{\theta/(2+\lambda)}$ is concave, thus by Jensen’s inequality and Assumption~\ref{asmp:A2},
    \[ \frac{1}{N_g}\sum_{i=1}^{N_g}\mathbb{E}\|\tilde x_{gi}\|^{2\theta} 
    \le \frac{1}{N_g} \sum_{i=1}^{N_g} \big(\mathbb{E}\|\tilde x_{gi}\|^{4+2\lambda}\big)^{\frac{\theta}{2+1\lambda}} 
    \le \bigg( \frac{1}{N_g}\sum_{i=1}^{N_g}\mathbb{E}\|\tilde x_{gi}\|^{4+2\lambda} \bigg)^{\frac{\theta}{2+\lambda}} 
    \le C_{\tilde x}^{\frac{\theta}{2+\lambda}} < \infty. \]
Therefore, taking expectations in~\eqref{eq:L1.1} and then the supremum over $g$ yields
    \[ \sup_g N_g^{-\theta}\,\mathbb{E}\big\|\widetilde{X}_g' \widetilde{X}_g\big\|_F^\theta  \;\le\; \sup_g \frac{1}{N_g}\sum_{i=1}^{N_g}\mathbb{E}\|\tilde x_{gi}\|^{2\theta} = O(1). \]

For the stochastic bound, for any $M>0$, by Markov’s inequality:
    \[ \sup_g \mathbb{P}\!\left( N_g^{-\theta}\big\|\widetilde{X}_g' \widetilde{X}_g\big\|_F^\theta > M \right)  \;\le\; \frac{\sup_g \mathbb{E}\!\left[N_g^{-\theta}\big\|\widetilde{X}_g' \widetilde{X}_g\big\|_F^\theta\right]}{M} = \frac{O(1)}{M} \;\to\; 0 \quad\text{as } M\to\infty. \]
Hence $N_g^{-\theta}\|\widetilde{X}_g' \widetilde{X}_g\|_F^\theta = O_p(1)$ uniformly in $g$, and thus 
\(\sup_g N_g^{-\theta}\big\|\widetilde{X}_g' \widetilde{X}_g\big\|_F^\theta = O_p(1)\).

\noindent
\textbf{(ii)} This argument parallels that of (i). Since $\|U_g U_g'\|_F = \|U_g\|^2 = \sum_{i=1}^{N_g} u_{gi}^2$, we have for all $\theta\ge1$,
    \[ \|U_g U_g'\|_F^\theta = \Big( \sum_{i=1}^{N_g} u_{gi}^2 \Big)^\theta \le N_g^{\theta-1} \sum_{i=1}^{N_g} |u_{gi}|^{2\theta}. \]
Taking conditional expectations and using the conditional Lyapunov inequality,
    \[ N_g^{-\theta}\, \mathbb{E}\!\big[ \|U_g U_g'\|_F^\theta \mid X,W \big]
    \le \frac{1}{N_g} \sum_{i=1}^{N_g} \mathbb{E}\!\big( |u_{gi}|^{2\theta} \mid X,W \big) 
    \le \frac{1}{N_g} \sum_{i=1}^{N_g} \Big( \mathbb{E}\!\big( |u_{gi}|^{4+2\lambda} \mid X,W \big) \Big)^{\frac{\theta}{2+\lambda}}.\]
By Jensen's inequality and Assumption~\ref{asmp:A2},
    \[ N_g^{-\theta}\, \mathbb{E}\!\big[ \|U_g U_g'\|_F^\theta \mid X,W \big]  \le \bigg( \frac{1}{N_g} \sum_{i=1}^{N_g} \mathbb{E}\!\big[ |u_{gi}|^{4+2\lambda} \mid X,W \big] \bigg)^{\frac{\theta}{2+\lambda}} \le C_u^{\frac{\theta}{2+\lambda}} <\infty. \]
Therefore,
    \[ \sup_g\; N_g^{-\theta}\, \mathbb{E}\!\big[ \|U_g U_g'\|_F^\theta \mid X,W \big] = O(1). \]

\medskip \noindent
\textbf{(iii)} Note that for any matrix $A$ and vector $v$, \(\|A'v\|^2 = v'AA'v = \operatorname{tr}(AA'vv') \le \|AA'\|_F \, \|vv'\|_F. \) Applying this with $A=\widetilde X_g$ and $v=U_g$, and raising both sides to the power $\theta/2 \ge 1/2$,
    \[ \|\widetilde X_g' U_g\|^2 \;\le\;  \|\widetilde X_g' \widetilde X_g\|_F \, \|U_g U_g'\|_F \quad \Longrightarrow \quad \|\widetilde X_g' U_g\|^{\theta} \;\le\; \|\widetilde X_g' \widetilde X_g\|_{F}^{\theta/2}  \; \|U_g U_g'\|_{F}^{\theta/2}. \]
Taking conditional expectation and applying Cauchy-Schwarz,
    \begin{align*}
        N_g^{-\theta}\, \mathbb{E}\!\left( \|\widetilde X_g' U_g\|^{\theta} \,\middle|\, X,W \right)
        &\le N_g^{-\theta}\, \|\widetilde X_g' \widetilde X_g\|_{F}^{\theta/2} \, \mathbb{E}\!\left[ \|U_g U_g'\|_{F}^{\theta/2} \mid X,W \right] \\
        &\le N_g^{-\theta}\, \|\widetilde X_g' \widetilde X_g\|_{F}^{\theta/2} \, \Big( \mathbb{E}\!\left[ \|U_g U_g'\|_{F}^{\theta} \mid X,W \right] \Big)^{1/2} \\
        &= \Big( N_g^{-\theta}\|\widetilde X_g' \widetilde X_g\|_F^{\theta} \Big)^{1/2} \Big( N_g^{-\theta}\, \mathbb{E}\!\left[ \|U_g U_g'\|_F^{\theta} \mid X,W \right] \Big)^{1/2}.
    \end{align*}
Taking supremum over $g$,
   \[  \sup_g N_g^{-\theta}\, \mathbb{E}\!\left( \|\widetilde X_g' U_g\|^{\theta} \mid X,W \right) \le \Big( \sup_g N_g^{-\theta}\|\widetilde X_g' \widetilde X_g\|_F^{\theta} \Big)^{1/2} \Big( \sup_g N_g^{-\theta}\, \mathbb{E}\!\left( \|U_g U_g'\|_F^{\theta} \mid X,W \right) \Big)^{1/2}. \]
Therefore, by parts \emph{(i)} and \emph{(ii)} of the lemma: 
    \[ \sup_g N_g^{-\theta}\, \mathbb{E}\!\left( \|\widetilde X_g' U_g\|^{\theta} \mid X,W \right) = O_p(1)^{1/2} \, O(1)^{1/2} = O_p(1). \]
\end{proof}

\medskip
\begin{prelemma}[Boundedness of $Q_n$ and $R_n^\star$] \label{pre:riesz-bdd}
Let $Q_n := \Gamma_n^{-1} a a'\Gamma_n^{-1}$ with $\Gamma_n = n^{-1}\sum_{g=1}^G \widetilde X_g'\widetilde X_g $, and $R_n^\star := \{\widetilde X_g Q_n \widetilde X_g'\}_{g=1}^G \in \mathcal H_n $. Under Assumptions~\ref{asmp:A2} and \ref{asmp:A4},
    \[ \|Q_n\|_F = O_p(1) \qquad\text{and}\qquad  \|R_n^\star\|_{\mathcal H_n} = O_p \bigg( \sqrt{ \frac{\sup_g N_g}{n}} \bigg). \]
\end{prelemma}

\noindent\emph{Remark.} 
This lemma shows that the Riesz representer $R_n^\star$ is stochastically controlled at the natural cluster-growth rate. Combined with the stability condition on $\mathcal A_n^\ast$ in Assumption~\ref{asmp:A5}, this will yield uniform bounds on the minimum-norm solutions $D_n^\star$.

\begin{proof}
\emph{Step 1: Boundedness of $Q_n$.}  
Assumption~\ref{asmp:A4} implies $\Gamma_n \xrightarrow{p} \Gamma_0 \succ 0$, so the eigenvalues of $\Gamma_n$ are bounded away from zero and infinity with probability approaching one. Therefore
    \[ \|\Gamma_n^{-1}\|_{\mathrm{op}} = O_p(1). \]
Since $a$ is fixed with $\|a\|=1$,
    \[ Q_n  = \Gamma_n^{-1} a a' \Gamma_n^{-1}  \quad\Longrightarrow\quad  \|Q_n\|_F \le \|\Gamma_n^{-1}\|_{\mathrm{op}}^2 \,\|aa'\|_F  = \|\Gamma_n^{-1}\|_{\mathrm{op}}^2 = O_p(1). \]

\medskip \noindent
\emph{Step 2: Boundedness of $R_n^\star$.} By definition,
    \[ \|R_n^\star\|_{\mathcal H_n}^2  = \frac{1}{n^2}\sum_{g=1}^G \|R_{g,n}^\star\|_F^2, \qquad   R_{g,n}^\star := \widetilde X_g Q_n \widetilde X_g'. \]
Using the identity \(\|B\|_F^2 = \operatorname{tr}(B'B)\) and the cyclicity of the trace, we obtain
    \[ \|R_{g,n}^\star\|_F^2 = \operatorname{tr}\big( (\widetilde X_g Q_n \widetilde X_g')' (\widetilde X_g Q_n \widetilde X_g') \big) = \operatorname{tr}\big( Q_n \widetilde X_g' \widetilde X_g Q_n \widetilde X_g' \widetilde X_g\big). \]
Let $A_g := \widetilde X_g'\widetilde X_g$, and by $\operatorname{tr}(BC) \le \|B\|_{\mathrm{op}} \, \|C\|_F$ and $\|Q_n A_g Q_n\|_F \le \|Q_n\|_{\mathrm{op}}^2 \|A_g\|_F$, we have
    \[ \|R_{g,n}^\star\|_F^2  = \operatorname{tr}(Q_n A_g Q_n A_g) \le \|Q_n A_g Q_n\|_F \,\|A_g\|_F \le \|Q_n\|_{\mathrm{op}}^2 \,\|A_g\|_F^2 = \|Q_n\|_{\mathrm{op}}^2 \,\|\widetilde X_g'\widetilde X_g\|_F^2. \]
Hence
    \[ \|R_n^\star\|_{\mathcal H_n}^2 \le \frac{\|Q_n\|_{\mathrm{op}}^2}{n^2} \sum_{g=1}^G \|\widetilde X_g'\widetilde X_g\|_F^2. \]

By Lemma~\ref{pre:bound}(i) with $\theta = 2$, \(\|\widetilde X_g'\widetilde X_g\|_F^2 = O_p(N_g^2)\) uniformly in $g$. Therefore
    \[ \frac{1}{n^2} \sum_{g=1}^G \|\widetilde X_g'\widetilde X_g\|_F^2 = O_p\!\left( \frac{1}{n^2} \sum_{g=1}^G N_g^2 \right) = O_p\!\left( \frac{\sup_g N_g}{n} \right). \]
Combining this with $\|Q_n\|_{\mathrm{op}}=O_p(1)$ from Step~1 yields
    \[ \|R_n^\star\|_{\mathcal H_n}^2 = O_p(1)\cdot O_p\!\left( \frac{\sup_g N_g}{n} \right) = O_p\!\left( \frac{\sup_g N_g}{n} \right), \]
and hence
    \[ \|R_n^\star\|_{\mathcal H_n} = O_p\!\left( \sqrt{\frac{\sup_g N_g}{n}} \right). \]
\end{proof}

\medskip
\begin{prelemma}[Boundedness of the minimum-norm dual solution]\label{pre:Dstar-bdd}
For each $n$, let $D_n^\star \in \mathcal K_n$ be the minimum-Frobenius-norm solution of \(\mathcal A_n^\ast D = R_n^\star,\) as in~\eqref{eq:min-norm}. Under Assumption~\ref{asmp:A5} and Lemma~\ref{pre:riesz-bdd},
    \[ \|D_n^\star\|_{\mathcal K_n} = O_p \bigg( \sqrt{\frac{\sup_g N_g}{n}} \bigg). \]
\end{prelemma}

\noindent\emph{Remark.}
This lemma shows that the Riesz weights $D_n^\star$ remain stochastically bounded at the appropriate rate. This is essential for controlling quadratic forms of the type $\langle \widehat C_n - C_n, D_n^\star\rangle_{\mathcal K_n}$, which appear in the asymptotic analysis of the variance estimator.

\begin{proof}
By definition, $D_n^\star$ is the minimum-Frobenius-norm solution to
    \[ \mathcal A_n^\ast D = R_n^\star \quad\text{in } \mathcal K_n. \]
In any finite-dimensional Hilbert space, the minimum-norm solution of a linear system is given by the Moore--Penrose pseudoinverse. Assumption~\ref{asmp:A5} states that $R_n^\star \in \mathrm{Range}(\mathcal A_n^\ast)$ for all sufficiently large $n$, and that the restricted pseudoninverse of $\mathcal A_n^\ast$ to its range is uniformly bounded:
    \[ \sup_n \big\|(\mathcal A_n^\ast|_{\mathrm{Range}(\mathcal A_n^\ast)})^{+}\big\|_{\mathrm{op}} \le C < \infty. \]
Thus, for all large $n$, \(D_n^\star = (\mathcal A_n^\ast|_{\mathrm{Range}(\mathcal A_n^\ast)})^{+}\big(R_n^\star\big), \) and therefore
    \[ \|D_n^\star\|_{\mathcal K_n} \;\le\; \big\|(\mathcal A_n^\ast|_{\mathrm{Range}(\mathcal A_n^\ast)})^{+}\big\|_{\mathrm{op}} \,\|R_n^\star\|_{\mathcal H_n} \;\le\; C\,\|R_n^\star\|_{\mathcal H_n}. \]

By Lemma~\ref{pre:riesz-bdd}, $\|R_n^\star\|_{\mathcal H_n} = O_p \bigg( \sqrt{\frac{\sup_g N_g}{n}} \bigg)$. Combining the two displays gives the desired result
   \[ \|D_n^\star\|_{\mathcal K_n}  = O_p \bigg( \sqrt{\frac{\sup_g N_g}{n}} \bigg). \]
\end{proof}

\medskip
\begin{prelemma}[Operator norms of $\mathcal A_n$ and $\mathcal A_n^\ast$]\label{pre:A-operator-bdd}
Let $\mathcal A_n:\mathcal H_n\to\mathcal K_n$ be defined as in \eqref{eq:A-def} and $\mathcal A_n^\ast: \mathcal K_n \to \mathcal H_n$ be its adjoint as in \eqref{eq:Astar-def}. Under Assumption~\ref{asmp:A4}, for all $n$,
    \[ \|\mathcal A_n\|_{\mathrm{op}} \;\le\; 1, \qquad \|\mathcal A_n^\ast\|_{\mathrm{op}} \;\le\; 1. \]
\end{prelemma}

\noindent\emph{Remark.} 
This lemma shows that the covariance-to-moment operator and its adjoint are uniformly bounded contractions. It will be used to control LLN/CLT errors involving $\mathcal A_n(\Sigma)$ and to relate LSQR residuals to errors in $D_n$.

\begin{proof}
Assumption~\ref{asmp:A4} implies that $W'W$ and $\widetilde X'\widetilde X$ are nonsingular for large $n$. Thus the residual-maker matrices $M_W := I_n - W(W'W)^{-1}W'$ and $M_{\widetilde X} := I_n - \widetilde X(\widetilde X'\widetilde X)^{-1}\widetilde X'$ are well defined. By the Frisch--Waugh--Lovell theorem,
    \[ M := M_{\widetilde X} M_W = I_n - [X\; W]\big([X\; W]'[X\; W]\big)^{-1}[X\; W]', \]
so $M$ is the residual-maker from the joint regression on $[X\;W]$. Hence $M$ is symmetric and idempotent with $\|M\|_{\mathrm{op}}=1$.

\medskip
Given $\Sigma=\{\Sigma_g\}_{g=1}^G\in\mathcal H_n$, define the block-diagonal $n\times n$ matrix $\widetilde\Sigma := \mathrm{diag}(\Sigma_1,\ldots,\Sigma_G)$. Then
    \[ \|\widetilde\Sigma\|_F^2 = \sum_{g=1}^G \|\Sigma_g\|_F^2 \quad\Longrightarrow\quad \|\Sigma\|_{\mathcal H_n}^2 = \frac{1}{n^2}\|\widetilde\Sigma\|_F^2. \]
The $(h,h')$ block of $M \widetilde\Sigma M'$ is
    \[ [M\widetilde\Sigma M']_{hh'}  = \sum_{g=1}^G M_{hg}\,\Sigma_g\,M_{gh'} = [\mathcal A_n(\Sigma)]_{hh'}. \]
Hence the full $n\times n$ matrix $M\widetilde\Sigma M'$ has block structure given by $\mathcal A_n(\Sigma)$, and therefore
    \[ \|\mathcal A_n(\Sigma)\|_{\mathcal K_n}^2 = \frac{1}{n^2}\sum_{h,h'} \big\|[\mathcal A_n(\Sigma)]_{hh'}\big\|_F^2 = \frac{1}{n^2}\|M\widetilde\Sigma M'\|_F^2. \]

\medskip
For any conformable matrices $B,C$, we have $\|BC\|_F \le \|B\|_{\mathrm{op}} \|C\|_F$. Applying this twice,
    \[ \|M\widetilde\Sigma M'\|_F \le \|M\|_{\mathrm{op}} \,\|\widetilde\Sigma M'\|_F \le \|M\|_{\mathrm{op}}^2 \,\|\widetilde\Sigma\|_F. \]
Thus
    \[ \| \mathcal A_n(\Sigma) \|_{\mathcal K_n}^2 = \frac{1}{n^2} \| M \widetilde\Sigma M'\|_F^2 \le \frac{ \|M\|_{\mathrm{op}}^4}{n^2} \|\widetilde\Sigma\|_F^2 = \|M\|_{\mathrm{op}}^4 \, \|\Sigma\|_{\mathcal H_n}^2. \]
Taking square roots,
    \[ \|\mathcal A_n(\Sigma)\|_{\mathcal K_n} \le \|M\|_{\mathrm{op}}^2 \|\Sigma\|_{\mathcal H_n} = \|\Sigma\|_{\mathcal H_n}, \quad\Rightarrow\quad \|\mathcal A_n\|_{\mathrm{op}} \le 1. \]

In Hilbert spaces, the operator norms of a bounded linear operator and its adjoint coincide:
    \[ \|\mathcal A_n^\ast\|_{\mathrm{op}} = \|\mathcal A_n\|_{\mathrm{op}}.\]
Thus $\|\mathcal A_n^\ast\|_{\mathrm{op}} \le 1$.
\end{proof}

\medskip
\begin{prelemma}[Boundedness of $\Sigma_n$ and $C_n$] \label{pre:Sigma-C-bdd}
Suppose Assumptions~\ref{asmp:A1}--\ref{asmp:A2} hold, then
    \[ \|\Sigma_n\|_{\mathcal H_n} = O_p\!\Big(\sqrt{\tfrac{\sup_g N_g}{n}}\Big), \qquad \|C_n\|_{\mathcal K_n} \le \|\Sigma_n\|_{\mathcal H_n} = O_p\!\Big(\sqrt{\tfrac{\sup_g N_g}{n}}\Big). \]
\end{prelemma}

\begin{proof}
Fix $g$. Since $\Sigma_g=\mathbb{E} [U_gU_g'\mid X,W]$ is positive semidefinite,
   \[ \|\Sigma_g\|_F \le \mathrm{tr} (\Sigma_g) = \mathbb{E} [\|U_g\|^2\mid X,W]. \]
By Assumption~\ref{asmp:A2} (take $\theta=1$ in Lemma~\ref{pre:bound}(ii), or replicate its argument),
    \[ \sup_g N_g^{-1} \mathbb{E} [\|U_g\|^2\mid X,W] = O_p(1), \]
so uniformly in $g$,
    \[ \|\Sigma_g\|_F = O_p(N_g) \quad\Longrightarrow\quad \|\Sigma_g\|_F^2 = O_p(N_g^2). \]
Therefore,
    \[ \|\Sigma_n\|_{\mathcal H_n}^2 = \frac{1}{n^2} \sum_{g=1}^G \|\Sigma_g\|_F^2 = O_p \! \left( \frac{1}{n^2} \sum_{g=1}^G N_g^2 \right) \le O_p \! \left( \frac{\sup_g N_g}{n} \right), \]
using $\sum_g N_g^2 \le (\sup_g N_g)\sum_g N_g = (\sup_g N_g)\,n$. Taking square roots yields
    \[ \|\Sigma_n\|_{\mathcal H_n} = O_p\!\big(\sqrt{\sup_g N_g/n}\big). \]

Finally, by Lemma~\ref{pre:A-operator-bdd}, $\|\mathcal A_n\|_{\mathrm{op}} \le 1$, hence
    \[ \|C_n\|_{\mathcal K_n} = \|\mathcal A_n(\Sigma_n)\|_{\mathcal K_n} \le \|\mathcal A_n\|_{\mathrm{op}} \, \|\Sigma_n\|_{\mathcal H_n} \le \|\Sigma_n\|_{\mathcal H_n}. \]
\end{proof}

\newpage
\section{Proof of Lemmas}
\paragraph{Proof of Lemma~\ref{lem:riesz}}

\begin{proof}
For any fixed sample size $n$, the space $(\mathcal H,\langle\cdot,\cdot\rangle_{\mathcal H})$ is a finite-dimensional Hilbert space with norm $\|\Sigma\|_{\mathcal H}^2=\langle \Sigma,\Sigma\rangle_{\mathcal H}$. Thus, to apply the Riesz representation theorem, it suffices to verify that $\Phi(\cdot)$ is linear and bounded.

\medskip\noindent
\emph{Step 1: Linearity.} For arbitrary $\{\Sigma_g\},\{\Sigma_g'\}\in\mathcal H$ and scalars $\alpha,\beta\in\mathbb R$,
    \[ \Phi\big(\{\alpha\Sigma_g+\beta\Sigma'_g\}\big) = \frac{1}{n^2}\sum_{g=1}^G  \operatorname{tr}\!\big((\alpha\Sigma_g+\beta\Sigma'_g)\,\widetilde X_g Q \widetilde X_g'\big) = \alpha\,\Phi(\{\Sigma_g\})+\beta\,\Phi(\{\Sigma'_g\}), \]
so $\Phi$ is linear.

\medskip\noindent
\emph{Step 2: Boundedness.} For any $\Sigma=\{\Sigma_g\}\in\mathcal H$, the trace identity implies
    \[ \Phi(\Sigma) = \frac{1}{n^2}\sum_{g=1}^G \operatorname{tr}\!\big(\Sigma_g\,\widetilde X_g Q \widetilde X_g'\big) = \big\langle \Sigma,\{\widetilde X_g Q \widetilde X_g'\}\big\rangle_{\mathcal H}. \]
Hence, by Cauchy--Schwarz in $\mathcal H$,
    \[ |\Phi(\Sigma)| \;\le\; \|\Sigma\|_{\mathcal H}\; \big\|\{\widetilde X_g Q \widetilde X_g'\}\big\|_{\mathcal H}. \]
Conditional on a fixed realization of $(X,W)$, the quantity $\|\{\widetilde X_g Q \widetilde X_g'\}\|_{\mathcal H}$ is finite because each block $\widetilde X_g Q \widetilde X_g'$ is a finite matrix and $\mathcal H$ is finite-dimensional. Thus $\Phi$ is a bounded linear functional on $\mathcal H$.

\medskip
By the Riesz representation theorem, there exists a unique $R^\star=\{R_g^\star\}\in\mathcal H$ such that \( \Phi(\Sigma)=\langle \Sigma, R^\star\rangle_{\mathcal H} \) for all $\Sigma\in\mathcal H$.   Matching the expression for $\Phi(\Sigma)$ block-by-block yields $R_g^\star=\widetilde X_g Q \widetilde X_g'$ for each $g$.
\end{proof}

\begin{remark}[Asymptotic relevance]
Lemma~\ref{lem:riesz} is a finite-sample result: for each fixed $n$, the functional $\Phi$ admits a unique Riesz representer $R_n^\star=\{\widetilde X_g Q_n \widetilde X_g'\}$. For asymptotic analysis, we only require that this sequence remain stochastically bounded. Lemma~\ref{pre:riesz-bdd} shows that $\|R_n^\star\|_{\mathcal H_n}=O_p(1)$, which ensures the representers remain well behaved as $n\to\infty$.
\end{remark}

\paragraph{Proof of Lemma~\ref{lem:adjoint}}
\begin{proof}
Because $(\mathcal H,\langle\cdot,\cdot\rangle_{\mathcal H})$ and $(\mathcal K,\langle\cdot,\cdot\rangle_{\mathcal K})$ are finite-dimensional Hilbert spaces, every linear operator $\mathcal A:\mathcal H\to\mathcal K$ admits a unique adjoint $\mathcal A^\ast:\mathcal K\to\mathcal H$.

To compute it explicitly, recall that for $\Sigma=\{\Sigma_g\}\in\mathcal H$,
    \[ \big[\mathcal A(\Sigma)\big]_{hh'} = \sum_{g=1}^G M_{hg}\,\Sigma_g\,M_{gh'}. \]
Using the scaled Frobenius inner product on $\mathcal K$,
    \[ \langle \mathcal A(\Sigma),D\rangle_{\mathcal K} = \frac{1}{n^2}\sum_{h,h'} \operatorname{tr}\!\left( \Big(\sum_{g=1}^G M_{hg}\Sigma_g M_{gh'}\Big)' D_{hh'} \right). \]

Since $M$ is symmetric ($M_{hg}=M_{gh}'$) and each $\Sigma_g$ is symmetric,
    \[ \Big(\sum_g M_{hg}\Sigma_g M_{gh'}\Big)' = \sum_{g=1}^G M_{h'g}\,\Sigma_g\,M_{gh}. \]
Substituting and rearranging sums,
    \[ \langle \mathcal A(\Sigma),D\rangle_{\mathcal K} = \frac{1}{n^2}\sum_{g=1}^G  \operatorname{tr}\!\Big( \Sigma_g\,\sum_{h,h'} M_{gh}\,D_{hh'}\,M_{h'g} \Big). \]

On the other hand, the inner product on $\mathcal H$ is
    \[ \langle \Sigma,\mathcal A^\ast(D)\rangle_{\mathcal H} = \frac{1}{n^2}\sum_{g=1}^G  \operatorname{tr}\!\left(\Sigma_g\,[\mathcal A^\ast(D)]_g\right). \]
Identifying the two expressions block-by-block yields
    \[ \mathcal A^\ast(D)]_g = \sum_{h=1}^G\sum_{h'=1}^G M_{gh}\,D_{hh'}\,M_{h'g}. \]
This is the block representation in~\eqref{eq:Astar-def}.
\end{proof}

\begin{remark}[Asymptotic relevance]
Lemma~\ref{lem:adjoint} characterizes the adjoint operator $\mathcal A_n^\ast$ for each fixed $n$. For asymptotic analysis, the corresponding inverse problem $\mathcal A_n^\ast D = R_n^\star$ must remain well behaved as $n\to\infty$. Assumption~\ref{asmp:A5} ensures this by requiring that 
$R_n^\star$ belongs to the range of $\mathcal A_n^\ast$ and that the restricted pseudoinverse $(\mathcal A_n^\ast)^+$ is uniformly bounded. Together with the stochastic boundedness of $R_n^\star$, this yields the stability of the minimum-norm solution $D_n^\star$ established in Lemma~\ref{pre:Dstar-bdd}.
\end{remark}

\paragraph{Proof of Lemma~\ref{lem:AAstar-psd}}
\begin{proof}
We use only standard properties of adjoint operators on Hilbert spaces.

\medskip\noindent
\emph{(i) Self-adjointness and positive semidefiniteness.} First, self-adjointness follows from the fact that $(\mathcal A\mathcal A^\ast)^\ast = (\mathcal A^\ast)^\ast \mathcal A^\ast = \mathcal A \mathcal A^\ast$, since the adjoint is an involution: $(\mathcal A^\ast)^\ast = \mathcal A$.

For positive semidefiniteness, take any $D\in\mathcal K$ and compute
    \[ \langle D, \mathcal A\mathcal A^\ast D\rangle_{\mathcal K} = \langle \mathcal A^\ast D, \mathcal A^\ast D\rangle_{\mathcal H} = \|\mathcal A^\ast D\|_{\mathcal H}^2 \;\ge\; 0, \]
using the defining property of the adjoint: $\langle \mathcal A x, y\rangle_{\mathcal K} = \langle x, \mathcal A^\ast y\rangle_{\mathcal H}$ for all $x\in\mathcal H$, $y\in\mathcal K$.  This proves that $\mathcal A\mathcal A^\ast$ is self-adjoint and positive semidefinite on $\mathcal K$.

\medskip\noindent
\emph{(ii) Strict positivity on $\mathrm{range}(\mathcal A)$.} Suppose $D\in\mathrm{range}(\mathcal A)$ and $\langle D,\mathcal A\mathcal A^\ast D\rangle_{\mathcal K}=0$. By (i),
    \[ 0 = \langle D,\mathcal A\mathcal A^\ast D\rangle_{\mathcal K} = \|\mathcal A^\ast D\|_{\mathcal H}^2, \]
so $\mathcal A^\ast D = 0$. But $\ker(\mathcal A^\ast)$ is the orthogonal complement of $\mathrm{range}(\mathcal A)$, that is
    \[ \ker(\mathcal A^\ast) = \mathrm{range}(\mathcal A)^\perp. \]
Thus $\mathcal A^\ast D = 0$ implies $D\perp \mathrm{range}(\mathcal A)$. Since we also have $D\in\mathrm{range}(\mathcal A)$, the only possibility is $D=0$.  Hence
    \[ \langle D,\mathcal A\mathcal A^\ast D\rangle_{\mathcal K} = 0, \ D\in\mathrm{range}(\mathcal A)  \quad\Longrightarrow\quad D=0, \]
so $\mathcal A\mathcal A^\ast$ is strictly positive definite on $\mathrm{range}(\mathcal A)$.

\medskip\noindent
\emph{(iii) Consequences for the normal equations and Krylov iterates.} For any $R\in\mathcal H$, the right-hand side of the coefficient-space normal equations \((\mathcal A\mathcal A^\ast)D = \mathcal A R \) lies in $\mathrm{range}(\mathcal A)$ by definition of the operator $\mathcal A$. Moreover, $\mathrm{range}(\mathcal A)$ is invariant under $\mathcal A\mathcal A^\ast$: if $D = \mathcal A x$ for some $x\in\mathcal H$, then
    \[ \mathcal A\mathcal A^\ast D = \mathcal A\mathcal A^\ast(\mathcal A x) = \mathcal A(\mathcal A^\ast\mathcal A x) \in \mathrm{range}(\mathcal A). \]
Hence, for any initial $D^{(0)}\in\mathrm{range}(\mathcal A)$ (in particular $D^{(0)}=0$), all Krylov iterates
    \[ D^{(t)} \in \mathcal K_t(\mathcal A\mathcal A^\ast,\mathcal A R) = \mathrm{span}\{\mathcal A R, (\mathcal A\mathcal A^\ast)\mathcal A R,\ldots\} \]
remain in $\mathrm{range}(\mathcal A)$, where $\mathcal A\mathcal A^\ast$ is positive definite. This ensures that conjugate-gradient–type methods and LSQR are well defined on the relevant subspace and converge to the unique minimum-norm solution.
\end{proof}

\paragraph{Proof of Lemma~\ref{lem:lsqr-residual-error}}
\begin{proof}
Let $e^{(t)} := D^{(t)} - D^\star$ and $r^{(t)} := R - \mathcal{A}^\ast D^{(t)}$. Since $D^\star$ is the minimum--norm solution of $\mathcal A^\ast D = R$, it satisfies the coefficient-space normal equations
    \[ \mathcal A\mathcal A^\ast D^\star = \mathcal A R \quad\Longleftrightarrow\quad \mathcal T D^\star = \mathcal A R, \qquad \mathcal T := \mathcal A\mathcal A^\ast. \]
For the iterate $D^{(t)}$ we have $\mathcal A^\ast D^{(t)} = R - r^{(t)}$ by definition, hence
    \[ \mathcal T D^{(t)} = \mathcal A \mathcal A^\ast D^{(t)} = \mathcal A(R - r^{(t)}) = \mathcal A R - \mathcal A r^{(t)}. \]
Subtracting the normal equations for $D^\star$ gives
    \begin{equation} \label{eq:lem_lsqr_1}
        \mathcal T e^{(t)} = \mathcal T(D^{(t)} - D^\star) = -\,\mathcal A r^{(t)}.
    \end{equation}

By construction of LSQR with $D^{(0)}=0$, all iterates $D^{(t)}$ lie in $\mathrm{range}(\mathcal A)$, and so does $D^\star$; therefore $e^{(t)}\in\mathrm{range}(\mathcal A)$. On $\mathrm{range}(\mathcal A)$, the operator $\mathcal T = \mathcal A\mathcal A^\ast$ is self-adjoint and strictly positive definite by Lemma~\ref{lem:AAstar-psd}. Thus there exist finite constants $0< \lambda_{\min} \le \lambda_{\max} < \infty$ such that for every $x \in \mathrm{range} (\mathcal A)$,
    \[ \lambda_{\min}\,\|x\|_{\mathcal K}^2 \;\le\; \langle x,\mathcal T x\rangle_{\mathcal K} = \|x\|_{\mathcal T}^2 \;\le\; \lambda_{\max}\,\|x\|_{\mathcal K}^2. \]
In particular,
    \begin{equation} \label{eq:lem_lsqr_2}
        \|x\|_{\mathcal K} \;\le\; \lambda_{\min}^{-1/2}\,\|x\|_{\mathcal T} \qquad \text{for all } x \in \mathrm{range} (\mathcal A).
    \end{equation}

Now compute, using \eqref{eq:lem_lsqr_1}:
    \[ \|e^{(t)}\|_{\mathcal T}^2 = \langle e^{(t)}, \mathcal T e^{(t)} \rangle_{\mathcal K} = -\,\langle e^{(t)}, \mathcal A r^{(t)} \rangle_{\mathcal K}. \]
Applying Cauchy--Schwarz and \eqref{eq:lem_lsqr_2},
    \[ \|e^{(t)}\|_{\mathcal T}^2 \le \|e^{(t)}\|_{\mathcal K}\,\|\mathcal A r^{(t)}\|_{\mathcal K} \le \lambda_{\min}^{-1/2}\,\|e^{(t)}\|_{\mathcal T}\, \|\mathcal A\|_{\mathrm{op}}\,\|r^{(t)}\|_{\mathcal H}. \]
If $e^{(t)}=0$ the desired inequality is trivial. Otherwise, dividing both sides by $\|e^{(t)}\|_{\mathcal T}>0$ yields
    \[ \|e^{(t)}\|_{\mathcal T} \le \lambda_{\min}^{-1/2}\,\|\mathcal A\|_{\mathrm{op}}\, \|r^{(t)}\|_{\mathcal H}. \]
Setting $\mathcal{C} := \lambda_{\min}^{-1/2}$ with the stopping rule $\|r^{(t)}\|_{\mathcal{H}} \le \tau_n \|R\|_{\mathcal{H}}$ yields the first bound in \eqref{eq:lsqr-det-bound-mu}.

Finally, under Assumption~\ref{asmp:A4} we have $\|\mathcal A\|_{\mathrm{op}}\le1$ by Lemma~\ref{pre:A-operator-bdd}, and Lemma~\ref{pre:riesz-bdd} implies
    \[ \|R\|_{\mathcal H} = O_p\!\left(\sqrt{\frac{\sup_g N_g}{n}}\right). \]
Substituting these bounds into above yields
    \[ \|D^{(t)} - D^\star\|_{\mathcal T} = O_p\!\left(\tau_n\sqrt{\frac{\sup_g N_g}{n}}\right), \]
as claimed.
\end{proof}

\paragraph{Proof of Lemma~\ref{lem:part-exact}}
\begin{proof}
By definition of the Riesz representer $D_n^\star$,
    \begin{equation} \label{eq:lem_exact_1}
        \langle A_n(\Sigma_n), D_n^\star\rangle_{K_n} = \langle \Sigma_n, A_n^* D_n^\star\rangle_{H_n} = \langle \Sigma_n, R_n^\star\rangle_{H_n}.
    \end{equation}

By assumption, $(A_n^{\mathrm{diag}})^* D_n^{\mathrm{diag}} = R_n^\star$, hence
    \begin{equation} \label{eq:lem_exact_2}
        \langle \Sigma_n, R_n^\star\rangle_{H_n} = \langle \Sigma_n, (A_n^{\mathrm{diag}})^* D_n^{\mathrm{diag}}\rangle_{H_n} = \langle A_n^{\mathrm{diag}}(\Sigma_n), D_n^{\mathrm{diag}}\rangle_{K_n^{\mathrm{diag}}},
    \end{equation}
where the last equality follows from the adjoint identity for $A_n^{\mathrm{diag}}$. Combining \eqref{eq:lem_exact_1} and \eqref{eq:lem_exact_2} yields the result.
\end{proof}

\paragraph{Proof of Lemma~\ref{lem:part-firstorder}}
\begin{proof}
First note that since $D_n^{\mathrm{diag}}\in \mathcal K_n^{\mathrm{diag}}$ and $\widehat C_n^{\mathrm{diag}}=P_{\mathrm{diag}}(\widehat C_n)$, the inner product ignores off-diagonal blocks, so
    \[ \langle \widehat C_n^{\mathrm{diag}}, D_n^{\mathrm{diag}}\rangle_{\mathcal K_n^{\mathrm{diag}}} = \langle \widehat C_n, D_n^{\mathrm{diag}}\rangle_{\mathcal K_n}. \]
Hence
    \[ \widehat V_{\partial,n}-\widehat V_{\mathrm{Riesz},n} = \big\langle \widehat C_n,\; D_n^{\mathrm{diag}}-D_n^\star\big\rangle_{\mathcal K_n}. \]
Add and subtract $C_n:=\mathcal A_n(\Sigma_n)$ to obtain
    \begin{equation}\label{eq:diag-full-decomp-proof}
        \widehat V_{\partial,n}-\widehat V_{\mathrm{Riesz},n} = \underbrace{\big\langle C_n,\; D_n^{\mathrm{diag}}-D_n^\star\big\rangle_{\mathcal K_n}}_{(I)} + \underbrace{\big\langle \widehat C_n-C_n,\; D_n^{\mathrm{diag}}-D_n^\star\big\rangle_{\mathcal K_n}}_{(II)}.
    \end{equation}

\medskip\noindent
\emph{Step 1: population term $(I)$.} Since $C_n=\mathcal A_n(\Sigma_n)$, the adjoint identity yields
    \[ (I) = \big\langle \Sigma_n,\; \mathcal A_n^\ast(D_n^{\mathrm{diag}}-D_n^\star)\big\rangle_{\mathcal H_n}. \]
Since $\mathcal A_n^{\mathrm{diag}}=P_{\mathrm{diag}}\circ\mathcal A_n$ and $P_{\mathrm{diag}}:\mathcal K_n\to\mathcal K_n^{\mathrm{diag}}$ is an orthogonal projection, it is self-adjoint. Hence
    \[ (\mathcal A_n^{\mathrm{diag}})^\ast = \mathcal A_n^\ast \circ P_{\mathrm{diag}} \qquad \text{and} \qquad  (\mathcal A_n^{\mathrm{diag}})^\ast D_n^{\mathrm{diag}} = \mathcal A_n^\ast D_n^{\mathrm{diag}} \]
as $P_{\mathrm{diag}} D_n^{\mathrm{diag}}=D_n^{\mathrm{diag}}$ for any $D_n^{\mathrm{diag}}\in\mathcal K_n^{\mathrm{diag}}$. Therefore,
    \[ \mathcal A_n^\ast(D_n^{\mathrm{diag}}-D_n^\star) = (\mathcal A_n^{\mathrm{diag}})^\ast D_n^{\mathrm{diag}}-R_n^\star, \]
since $\mathcal{A}^* D^\star = R^\star$, and hence
    \[ (I) = \big\langle \Sigma_n,\; (\mathcal A_n^{\mathrm{diag}})^\ast D_n^{\mathrm{diag}}-R_n^\star \big\rangle_{\mathcal H_n}. \]

By Cauchy--Schwarz with boundedness of $\|\Sigma_n\|_{\mathcal H_n}$ in Lemma~\ref{pre:Sigma-C-bdd} and Assumption~\eqref{asmp:PartApprox}
   \[ |(I)| \le \|\Sigma_n\|_{\mathcal H_n}\, \big\| (\mathcal A_n^{\mathrm{diag}})^\ast D_n^{\mathrm{diag}} - R_n^\star \big\|_{\mathcal H_n} = = O_p\!\Big( \sqrt{ \tfrac{\sup_g N_g}{n}} \Big) \cdot o_p \! \Big( \mu_n^{-1} \sqrt{\tfrac{n}{\sup_g N_g}} \Big) = o_p(\mu_n^{-1}). \]

\medskip\noindent
\emph{Step 2: sample term $(II)$.}
Let $\widetilde D_n^\Delta$ be the $n\times n$ block matrix with $(h,h')$ block $D_{n,hh'}^\Delta := D_{n,hh'}^{\mathrm{diag}}-D_{n,hh'}^\star$, and define
    \[ H_n^\Delta := \frac{1}{n^2} M'\widetilde D_n^\Delta M. \]
Exactly as in Step~2 of the proof of Theorem~\ref{thm:riesz-t}, the block-trace identity implies 
    \begin{equation}\label{eq:ZDelta-trace-proof}
        Z_n^\Delta := \big\langle \widehat C_n-C_n,\; D_n^{\mathrm{diag}}-D_n^\star\big\rangle_{\mathcal K_n} = \mathrm{tr}\!\Big( H_n^\Delta\big(U_nU_n' - \mathbb E[U_nU_n'\mid X_n,W_n]\big) \Big).
    \end{equation}
Conditional on $(X_n,W_n)$, $\mathbb E[Z_n^\Delta\mid X_n,W_n]=0$. Hence by (conditional) Chebyshev, for any $\varepsilon>0$,
    \[ \Pr\big(|\mu_n Z_n^\Delta|>\varepsilon\mid X_n,W_n\big) \le \frac{\mathbb E[\mu_n^2(Z_n^\Delta)^2\mid X_n,W_n]}{\varepsilon^2}, \]
so it suffices to show
    \begin{equation} \label{eq:var-target-proof}
        \mathbb E[\mu_n^2(Z_n^\Delta)^2\mid X_n,W_n]=o_p(1).
    \end{equation}

Proceeding exactly as in \eqref{eq:chey_bound}-\eqref{eq:chey_sum_bound} of Theorem~\ref{thm:riesz-t} with $H_n$ replaced by $H_n^\Delta$, cluster independence yields the same orthogonality relations and gives
    \begin{equation} \label{eq:second-moment-proof}
        \mathbb E[(Z_n^\Delta)^2\mid X_n,W_n] \;\lesssim\; (\sup_g N_g)^2\,\|H_n^\Delta\|_F^2.
    \end{equation}
Since $\|M\|_{\mathrm{op}}=1$,
    \[ \|H_n^\Delta\|_F^2 = \frac{1}{n^4}\|M'\widetilde D_n^\Delta M\|_F^2 \le \frac{1}{n^4}\|\widetilde D_n^\Delta\|_F^2. \]
Using $\|\widetilde D_n^\Delta\|_F^2\le 2\|\widetilde D_n^{\mathrm{diag}}\|_F^2+2\|\widetilde D_n^\star\|_F^2$ and $\|\widetilde D\|_F^2=n^2\|D\|_{\mathcal K_n}^2$, we obtain
    \[ \|\widetilde D_n^\Delta\|_F^2 \le 2n^2\|D_n^{\mathrm{diag}}\|_{\mathcal K_n}^2 + 2n^2\|D_n^\star\|_{\mathcal K_n}^2. \]
By Lemma~\ref{pre:Dstar-bdd} and Assumption~\eqref{asmp:PartApprox},
    \[ \|D_n^\star\|_{\mathcal K_n} = O_p\!\Big(\sqrt{\tfrac{\sup_g N_g}{n}}\Big), \qquad \|D_n^{\mathrm{diag}}\|_{\mathcal K_n} = O_p\!\Big(\sqrt{\tfrac{\sup_g N_g}{n}}\Big).\]
Therefore,
    \[ \|\widetilde D_n^\Delta\|_F^2 = O_p\big(n\,\sup_g N_g\big), \qquad\text{hence} \quad \|H_n^\Delta\|_F^2 = O_p\!\Big(\frac{\sup_g N_g}{n^3}\Big). \]
Combining with \eqref{eq:second-moment-proof} yields
    \[ \mathbb E[(Z_n^\Delta)^2\mid X_n,W_n] = O_p\!\Big(\Big(\frac{\sup_g N_g}{n}\Big)^3\Big). \]
Consequently,
    \[ \mathbb E[\mu_n^2(Z_n^\Delta)^2\mid X_n,W_n] = O_p\!\Big(\mu_n^2\Big(\frac{\sup_g N_g}{n}\Big)^3\Big) = o_p(1) \]
by Assumption~\ref{asmp:A3}. This establishes \eqref{eq:var-target-proof} and hence $(II)=Z_n^\Delta=o_p(\mu_n^{-1})$.

\medskip
From \eqref{eq:diag-full-decomp-proof}, $(I)=o(\mu_n^{-1})$ and $(II)=o_p(\mu_n^{-1})$ imply
    \[ \widehat V_{\partial,n}-\widehat V_{\mathrm{Riesz},n} = o_p(\mu_n^{-1}), \qquad\text{equivalently}\qquad \mu_n(\widehat V_{\partial,n}-\widehat V_{\mathrm{Riesz},n})=o_p(1). \]
The statement about the same studentized asymptotic distribution follows from Slutsky's theorem.
\end{proof}

\newpage
\section{Proof of Theorems}
\paragraph{Proof of Theorem~\ref{thm:riesz-t}}
\begin{proof}
\emph{Step 1: Asymptotic normality.} According to the Slutsky theorem, the left-hand side of (\ref{eq:normality}) is
    \[ V^{-1/2} a' (\widehat{\beta} - \beta) = \nu_a^{-1/2} \mu_n^{1/2} n^{-1} a' \Gamma_0^{-1} \sum_{g=1}^{G} \widetilde X_g' U_g (1 + o_p(1)), \]
due to $\mu_n V \to \nu_a$ by Assumption~\ref{asmp:A3} and $\Gamma_n \to \Gamma_0$ by Assumption~\ref{asmp:A4}.
 Thus, it is enough to show
     \begin{equation} \label{eq:asym_norm}
         \nu_a^{-1/2}  \mu_n^{1/2} n^{-1} a' \Gamma_n^{-1} \sum_{g=1}^{G} \widetilde X'_g U_g \; \xlongrightarrow{d} \; \mathcal{N} (0,1). 
     \end{equation}

Define $Z_g = \nu_a^{-1/2} \mu_n^{1/2} n^{-1} a' \Gamma_0^{-1} \widetilde{X}'_g U_g$. By Assumptions~\ref{asmp:A1}, it is an independent sequence conditional on $(X,W)$ with mean
    \[ \mathbb{E} [Z_g \mid X,W] = \nu_a^{-1/2} \mu_n^{1/2} n^{-1} a' \Gamma_0^{-1} \widetilde{X}'_g \, \mathbb{E} [U_g \mid X,W] = 0, \]
and variance
    \[ \mathbb{E} [Z^2_g \mid X,W] = \nu_a^{-1} \mu_n n^{-2} a' \Gamma_0^{-1} \widetilde{X}'_g \Sigma_g \widetilde{X}_g \Gamma_0^{-1} a.\]
And by Assumption \ref{asmp:A3} and \ref{asmp:A4}, 
    \[ \sum_{g=1}^{G} \mathbb{E} [Z_g^2 \mid X,W] = \nu_a^{-1} \mu_n a' \Gamma_0^{-1} \bigg( \frac{1}{n^2} \sum_{g=1}^{G} \widetilde{X}'_g \Sigma_g \widetilde X_g \bigg) \Gamma_0^{-1} a' = \nu_a^{-1} (\mu_n V) \; \xrightarrow{p} \; 1. \]

Then \eqref{eq:asym_norm} follows from the Lyapunov Central Limit Theorem for heterogeneous, independent random variables, if for some $\xi > 0$, it holds that $\sum_{g=1}^{G} \mathbb{E} [ |Z_g|^{2+\xi} \mid X,W]\to 0$ (Lyapunov's condition). Set $\xi = \lambda>0$, then
    \begin{align}
        \sum_{g=1}^{G} \mathbb{E} \big[ |Z_g|^{2+\xi} \mid X,W \big] 
        &\leq \nu_a^{-\frac{2+\xi}{2}} \mu_n^{\frac{2+\xi}{2}} \|a' \Gamma_0^{-1}\|^{2+\xi} \, n^{-2-\xi} \sum_{g=1}^{G} \mathbb{E} \left[\|\widetilde X'_g U_g \|^{2+\xi} \mid X,W \right] \notag \\ 
        &\lesssim O_p \bigg( \mu_n^{\frac{2+\xi}{2}} \frac{\sum_{g=1}^{G} N_g^{2+\xi}}{n^{2+\xi}} \bigg) \lesssim O_p \bigg( \Big( \mu_n^{\frac{2+\xi}{2+2\xi}} \frac{\sup_g N_g}{n} \Big)^{1+\xi} \bigg) = o_p(1),
    \end{align}
where the second inequality is due to the positive definiteness of $\Gamma_0$ in Assumption ~\ref{asmp:A4} and Lemma~\ref{pre:bound} (with $\theta = \xi+2$), and the convergence is due to the cluster size restriction in Assumption~\ref{asmp:A3}. \\

\medskip \noindent
\emph{Step 2: Consistency of $\widehat V_{\mathrm{Riesz}}$}. By Lemmas~\ref{pre:riesz-bdd}, \ref{pre:Dstar-bdd}, and~\ref{pre:A-operator-bdd}, the Riesz representers $R_n^\star$, the dual solutions $D_n^\star$, and the operators $(\mathcal A_n,\mathcal A_n^\ast)$ remain well defined and uniformly bounded for all $n$, so the Riesz representation and the adjoint equation $\mathcal A_n^\ast D_n^\star = R_n^\star$ hold along the entire asymptotic sequence.

By construction of the oracle variance functional in~\eqref{eq:oracle-var} and Lemma~\ref{lem:riesz}, together with the definition of the cross-moments $C_n := \mathcal A_n(\Sigma_n)$ and minimum-Frobenius-norm dual representer $D_n^\star$, the adjoint identity gives, for each $n$,
    \[ V = \Phi_n(\Sigma_n) = \big\langle \Sigma_n, R_n^\star \big\rangle_{\mathcal H_n} = \langle \Sigma_n, \mathcal{A}^*_n D_n^\star \rangle_{\mathcal H_n} = \big\langle \mathcal{A}_n (\Sigma_n), D_n^\star \big\rangle_{\mathcal K_n} = \big\langle C_n, D_n^\star \big\rangle_{\mathcal K_n}. \]
Thus \eqref{eq:consistency} is equivalent to
    \begin{equation} \label{eq:V_converge}
        \mu_n (\widehat V_{\mathrm{Riesz}} - V) =  \mu_n \big\langle \widehat C_n - C_n, D_n^\star\big\rangle_{\mathcal K_n} \;\xrightarrow{p}\; 0.
    \end{equation}

Let $\widetilde D_n$ be the block matrix with $(h,h')$ block $D_{n,hh'}^\star$ and define the  $n\times n$ kernel  
    \[ H_n := \frac{1}{n^2} M' \widetilde D_n M, \]
where $M$ is the global residual-maker matrix from Section~\ref{sec:model}. Then
     \[ Z_n := \big\langle \widehat C_n - C_n, D_n^\star\big\rangle_{\mathcal K_n} = \mathrm{tr}\Big( H_n\big(U_n U_n' - \mathbb E[U_n U_n' \mid X_n,W_n]\big) \Big). \]
Conditional on $(X_n, W_n)$,
    \[ \mathbb{E} [\mu_n Z_n \mid X_n, W_n] = \mu_n \, \mathrm{tr} \Big( H_n \, \mathbb{E} \big[ U_n U_n' - \mathbb E[U_n U_n' \mid X_n,W_n] \big| X_n, W_n \big] \Big) = 0. \]
Hence by (conditional) Chebyshev,
   \[\mathrm{Pr} \big( |\mu_n Z_n| > \varepsilon \mid X, W \big) \leq \frac{\mathbb{E} [\mu_n^2 Z_n^2 \mid X_n, W_n]}{\varepsilon^2}, \]
so it suffices to show
    \begin{equation} \label{eq:chey_goal}
        \mathbb{E} \left[ \mu_n^2 Z_n^2 \mid X_n, W_n \right] = o_p(1).  
    \end{equation}

\medskip
Using the block-trace identity $\mathrm{tr}(AB)=\sum_{g,h} \mathrm{tr}(A_{gh} B_{hg})$, we can write
    \[ Z_n = \sum_{g,h} \mathrm{tr} \Big( H_{gh} \big( U_h U'_g - \mathbb{E} [U_h U'_g \mid X_n, W_n] \big) \Big) = \sum_g A_g + 2 \sum_{g<h} B_{gh} \]
where 
    \[A_g = \mathrm{tr} \big( H_{gg} ( U_g U'_g - \Sigma_g) \big), \qquad B_{gh} = \mathrm{tr} \big( H_{gh} U_h U'_g \big) = U'_g H_{gh} U_h \quad (g<h). \]
By construction, $\mathbb{E} [A_g \mid X_n, W_n]=0$ and $\mathbb{E} [B_{gh}\mid X_n, W_n]=0$, and independence of $\{U_g\}$ across $g$ yields the orthogonality relations: 
\begin{enumerate}[(i)]
    \item $\mathbb{E}[A_g A_h \mid X_n, W_n] = 0$ for $g\neq h$:
        \[  \mathbb{E} [A_g A_h \mid X_n, W_n] = \big( \mathbb{E} [A_g \mid X_n, W_n] \big)\, \big(\mathbb{E} [A_h \mid X_n, W_n]\big) = 0. \]
    \item $\mathbb{E} [A_g B_{g'h'} \mid X_n, W_n] = 0$ for all $g' < h'$:
    \begin{enumerate}[(a)]
        \item For $h' \neq g$: $\mathbb{E} [A_g U'_g H_{g'h'} U_{h'} \mid X_n, W_n] = \big(\mathbb{E} [A_g U'_g H_{g'h'} \mid X_n, W_n]\big) \big( \mathbb{E} [U_{h'} \mid X_n, W_n] \big) = 0$;
        \item For $g' \neq g$: $\mathbb{E} [B_{g'h'} A_g \mid X_n, W_n] = \big( \mathbb{E} [U'_{g'} \mid X_n, W_n] \big) \big( \mathbb{E}[H_{g'h'} U_{h'} A_g \mid X_n, W_n] \big) = 0$.
    \end{enumerate}
    \item $\mathbb{E} [B_{gh} B_{g'h'} \mid X_n, W_n] = 0$ whenever $g\neq g'$ or $h \neq h'$: \\ (assume $g \le g'$ without loss of generality)
    \begin{enumerate}[(a)]
        \item For $g<g'$, we have $g\neq h'$ since $g'<h'$, then
            \[ \mathbb{E}[B_{gh} B_{g'h'} \mid X_n, W_n] = \big( \mathbb{E}[ U'_g \mid X_n, W_n] \big) \, \big( \mathbb{E}[ H_{gh} U_h U'_{g'} H_{g'h'} U_{h'} \mid X_n, W_n] \big) = 0. \]
        \item For $g=g'$ and $h \neq h'$, we have
            \[ \mathbb{E}[B_{gh} B_{g'h'} \mid X_n, W_n] = \big( \mathbb{E}[ U'_g H_{gh} U_h U'_{g'} H_{g'h'} \mid X_n, W_n] \big) \, \big( \mathbb{E}[U_{h'} \mid X_n, W_n] \big) = 0. \]
    \end{enumerate}
\end{enumerate}
Therefore, 
    \begin{align*}
        \mathbb{E} \left[ \mu_n^2 Z_n^2 \mid X_n, W_n \right] 
        &= \mu_n^2 \, \mathbb{E} \bigg[ \Big( \sum_g A_g + 2 \sum_{g<h}  B_{gh} \Big)^2 \Big| X_n, W_n \bigg] \\
        &= \mu_n^2 \, \bigg( \sum_g \mathbb{E} [ A_g^2 \mid X_n, W_n] + 4 \sum_{g<h}  \mathbb{E} [B_{gh}^2 \mid X,W] \bigg).
    \end{align*}
Write $B_{gg} := \mathrm{tr} (H_g U_g U'_g) = U'_g H_{gg} U_g$, so that $A_g = B_{gg} - \mathbb{E}[B_{gg} \mid X_n, W_n]$ and hence
    \[ \mathbb{E} [A_g^2 \mid X_n, W_n] = \mathbb{E} [B_{gg}^2 \mid X_n, W_n] - \big( \mathbb{E} [B_{gg} \mid X_n, W_n] \big)^2 \leq \mathbb{E} [B_{gg}^2 \mid X_n, W_n]. \]
Thus
    \begin{align} \label{eq:chey_bound}
        \mathbb{E} \left[ \mu_n^2 Z_n^2 \mid X_n, W_n \right] 
        &\le \mu_n^2 \, \bigg( \sum_g \mathbb{E} [ B_{gg}^2 \mid X_n, W_n] + 4 \sum_{g<h}  \mathbb{E} [B_{gh}^2 \mid X,W] \bigg) \notag  \\
        &\le 2 \mu_n^2 \bigg( \sum_{g,h} \mathbb{E} [B_{gh}^2 \mid X,W] \bigg). 
    \end{align}
    
Next, for any $g,h=1,\cdots, G$, by Cauchy-Schwarz inequality in the trace form
    \[ B_{gh}^2 = \big| \mathrm{tr} (H_{gh} U_h U'_g) \big|^2 \leq \|H_{gh}\|_F^2 \, \|U_h U'_g\|_F^2, \]
and then taking conditional expectation
    \[ \mathbb{E} [B^2_{gh} \mid X_n, W_n] \leq \|H_{gh}\|_F^2 \, \mathbb{E} \left[ \|U_h U'_g\|^2_F \mid X_n, W_n \right] \leq \|H_{gh}\|_F^2 \, \big( \sup_g \mathbb{E} \left[ \|U_g U'_g\|^2_F \mid X_n, W_n \right] \big), \]
where the last inequality comes from
    \begin{align*}
         \mathbb{E} \left[ \|U_h U'_g\|^2_F \mid X_n, W_n \right] &= \mathbb{E} \left[ \|U_h\|^2 \|U_g\|^2 \mid X_n, W_n \right] \le \big( \mathbb{E}[\|U_h\|^4 \mid X_n, W_n]\big)^{1/2} \big( \mathbb{E}[\|U_g\|^4 \mid X_n, W_n]\big)^{1/2} \\
         &\le \sup_g \mathbb{E}[\|U_g\|^4 \mid X_n, W_n] = \sup_g \mathbb{E} \left[ \|U_g U'_g\|^2_F \mid X_n, W_n \right].
    \end{align*}
Hence, by Lemma~\ref{pre:bound}(ii) with $\theta=2$ and summing over $(g,h)$,
    \begin{equation} \label{eq:chey_sum_bound}
        \sum_{g,h=1}^{G} \mathbb{E} [B^2_{gh} \mid X_n, W_n] \leq \Big( \sup_g \, \mathbb{E} \left[ \|U_g U'_g \|_F^2 \mid X_n, W_n \right] \Big) \bigg( \sum_{g,h} \|H_{gh}\|_F^2 \bigg) = O_p \Big( (\sup_g N_g)^2 \Big) \, \|H_n\|_F^2.
    \end{equation}

Because $\|M\|^2_{\mathrm{op}}=1$ and $\|\widetilde D_n \|_F^2 = n^2 \|D^\star\|^2_{\mathcal{K}_n} = O_p(n \sup_g N_g)$ by Lemma~\ref{pre:Dstar-bdd} 
    \[ \|H_n\|_F^2 = \frac{1}{n^4} \|M' \widetilde D_n M\|_F^2 \leq \frac{1}{n^4} \|M\|^4_{\mathrm{op}} \, \| \widetilde D_n \|_F^2 = O_p \bigg( \frac{\sup_g N_g}{n^3} \bigg). \]
Combining,
    \[  \sum_{g,h=1}^{G} \mathbb{E} [B^2_{gh} \mid X_n, W_n] = O_p \Big( (\sup_g N_g)^2 \Big) \, O_p \bigg( \frac{\sup_g N_g}{n^3} \bigg) = O_p \bigg( \Big( \frac{\sup_g N_g}{n} \Big)^3 \bigg). \]
Therefore, from \eqref{eq:chey_bound}, and by Assumption~\ref{asmp:A3} with $0<\lambda \leq 2$
   \[ \mathbb{E} \left[ \mu_n^2 Z_n^2 \mid X_n, W_n \right] \le 2 \mu_n^2 \bigg( \sum_{g,h=1}^{G} \mathbb{E} [B^2_{gh} \mid X_n, W_n] \bigg) = O_p \bigg(\mu_n^2 \Big( \frac{\sup_g N_g}{n} \Big)^3 \bigg) = o_p(1), \]
 Thus, \eqref{eq:chey_goal} holds, and we conclude \eqref{eq:V_converge}, which is equivalent to \eqref{eq:consistency}.

\medskip \noindent
\emph{Step 3: Riesz-$t$ statistic}. By Slutsky's theorem applied to \eqref{eq:normality} and \eqref{eq:consistency}
    \[ t_{\mathrm{Riesz}} := \widehat V^{-1/2}_{\mathrm{Riesz}} a' (\widehat \beta - \beta) = \big(V_{\mathrm{Riesz}}/V \big)^{-1/2} \, \left[ V^{-1/2} a' (\widehat \beta - \beta) \right] \, \longrightarrow \mathcal{N}(0,1), \]
which is \eqref{eq:t-Riesz}. This completes the proof.
\end{proof}

\paragraph{Proof of Theorem~\ref{thm:feasible}}
\begin{proof} 
The key is to compare the LSQR approximation with the exact Riesz estimator. By definition,
    \[ \widehat V_{\mathrm{Riesz},n}^{\mathrm{LSQR}} - \widehat V_{\mathrm{Riesz},n}^\star = \langle \widehat C_n,\, \widehat D_n - D_n^\star \rangle_{\mathcal K_n} \;\le\; \|\widehat C_n\|_{\mathcal K_n}\, \|\widehat D_n - D_n^\star\|_{\mathcal K_n}. \]

By~\eqref{eq:lsqr-det-bound-mu} in Lemma~\ref{lem:lsqr-residual-error}, the stopping rule \eqref{eq:lsqr-stop} implies a bound in the $\mathcal T_n$-norm,
    \[ \|\widehat D_n - D_n^\star\|_{\mathcal T_n} \;=\; O_p\!\left(\tau_n \sqrt{\frac{\sup_g N_g}{n}}\right), \qquad \mathcal T_n := \mathcal A_n\mathcal A_n^\ast. \]
Moreover, the \emph{proof} of Lemma~\ref{lem:lsqr-residual-error} (in particular, inequality~\eqref{eq:lem_lsqr_2} shows that for all $x\in\mathrm{range}(\mathcal A_n)$,
    \[ \|x\|_{\mathcal K_n} \;\le\; \mathcal C\,\|x\|_{\mathcal T_n}, \]
with $\mathcal C = \lambda_{\min}^{-1/2}$, where $\lambda_{\min}>0$ is the minimal eigenvalue of $\mathcal T_n$ on $\mathrm{range}(\mathcal A_n)$. Since both $\widehat D_n$ and $D_n^\star$ lie in $\mathrm{range}(\mathcal A_n)$ (LSQR is initialized at $0$ and Lemma~\ref{lem:AAstar-psd} ensures invariance of $\mathrm{range}(\mathcal A_n)$), we obtain
    \begin{equation} \label{eq:D-error}
        \|\widehat D_n - D_n^\star\|_{\mathcal K_n} \;\le\; \mathcal C\,\|\widehat D_n - D_n^\star\|_{\mathcal T_n} \;=\; O_p\!\left(\tau_n \sqrt{\frac{\sup_g N_g}{n}}\right).  
    \end{equation}

To bound $\|\widehat C_n\|_{\mathcal K_n}$, note that
    \[ \|\widehat C_n\|_{\mathcal K_n}^2 = \frac{1}{n^2}\sum_{h,h'}\|\widehat U_h \widehat U_{h'}'\|_F^2 = \Big(\frac{1}{n}\sum_{h=1}^G\|\widehat U_h\|^2\Big)^2 \quad \Rightarrow \quad \|\widehat C_n\|_{\mathcal K_n} = \frac{1}{n}\sum_{h=1}^G \|\widehat U_h\|^2. \]
Because $\widehat U = M U$ and $M$ is an orthogonal projection with $\|M\|_{\mathrm{op}} = 1$, we have $\sum_h \|\widehat U_h\|^2 = \|\widehat U\|^2 \le \|U\|^2 = \sum_h \|U_h\|^2$. By Lemma~\ref{pre:bound}(ii) with $\theta=1$, \(\mathbb E[\|U_g\|^2 \mid X,W] \le \mathcal{C} N_g \) uniformly in $g$, hence
    \[ \mathbb E\!\left[\frac{1}{n}\sum_{h=1}^G \|U_h\|^2 \, \Big| \, X, W \right] \le \frac{\mathcal{C}}{n} \sum_{h=1}^G N_h = \mathcal{C}. \]
Markov's inequality then gives
    \[\|\widehat C_n\|_{\mathcal K_n} = \frac{1}{n}\sum_{h=1}^G\|\widehat U_h\|^2 \;\le\; \frac{1}{n}\sum_{h=1}^G\|U_h\|^2 = O_p(1). \]
Combining this with \eqref{eq:D-error} yields
    \[ \big|\widehat V_{\mathrm{Riesz},n}^{\mathrm{LSQR}} - \widehat V_{\mathrm{Riesz},n}^\star \big|  = O_p(1)\cdot O_p\!\left(\tau_n \sqrt{\frac{\sup_g N_g}{n}}\right) = O_p\!\left(\tau_n \sqrt{\frac{\sup_g N_g}{n}}\right). \]
Under condition~\eqref{eq:tau-rate},
    \[ \mu_n \, \big|\widehat V_{\mathrm{Riesz},n}^{\mathrm{LSQR}} - \widehat V_{\mathrm{Riesz},n}^\star \big| = O_p\!\left(\mu_n \tau_n \sqrt{\frac{\sup_g N_g}{n}}\right) = o_p(1). \]
Thus
    \[ \mu_n \, \big|\widehat V_{\mathrm{Riesz},n}^{\mathrm{LSQR}} - V \big| \le \mu_n \, \big|\widehat V_{\mathrm{Riesz},n}^{\mathrm{LSQR}} - \widehat V_{\mathrm{Riesz},n}^\star \big| + \mu_n \big|\widehat V_{\mathrm{Riesz},n}^\star - V \big| = o_p(1) \]
where the second term is due to Theorem~\ref{thm:riesz-t}. This establishes \eqref{eq:consistency-lsqr}
    \[ \widehat V_{\mathrm{Riesz},n}/V \to 1. \]
 
By Slutsky theorem,
   \[ t_{\mathrm{Riesz}}^{\mathrm{LSQR}} := \widehat V_{\mathrm{Riesz},n}^{-1/2} a' (\widehat \beta - \beta) = \big(\widehat V_{\mathrm{Riesz},n}/V \big)^{-1/2} \, \left[ V^{-1/2} a' (\widehat \beta - \beta) \right] \, \longrightarrow \mathcal{N}(0,1), \]
which establishes \eqref{eq:t-lsqr}.
\end{proof}

\newpage
\section{Additional Simulation Results}


\end{document}